\documentclass[11pt, a4paper, onecolumn, accepted=2024-07-22]{quantumarticle}
\pdfoutput=1

\usepackage{fullpage}

\usepackage[persection]{belov_paper}
\usepackage[svgnames]{xcolor}
\numberwithin{equation}{section}

\usepackage[framemethod=TikZ]{mdframed}
\usepackage{chngcntr}

\newmdenv[%
  roundcorner=5pt,
  linecolor=blue!15,
  linewidth=2pt,
  subtitlebackgroundcolor=blue!15,
  subtitleaboveskip=0pt,
  subtitlebelowskip=0pt,
  subtitleinneraboveskip=0pt,
  innerbottommargin=0pt,
  subtitlefont=\normalfont
]{mdfigure}

\counterwithin{figure}{section}
\def\myfigureInternal#1#2#3{
\refstepcounter{figure} #1
\begin{mdfigure}[
  frametitle={
    \tikz[baseline=(current bounding box.east),outer sep=0pt]
    \node[anchor=east,rectangle,fill=blue!15,rounded corners]
    {Figure~\thefigure};},
  frametitleaboveskip=-10pt
  innertopmargin=0pt,
  innerbottommargin=0pt,
  roundcorner=5pt,
  linecolor=blue!15,
  linewidth=2pt,
  subtitlebackgroundcolor=blue!15,
  subtitleaboveskip=0pt,
  subtitlebelowskip=0pt,
  subtitleinneraboveskip=0pt,
  subtitlefont=\normalfont
]
#3
\mdfsubtitle{\medskip #2}
\end{mdfigure}
}

\def\myfigure#1#2#3{
\begin{figure}[htb]
\myfigureInternal{#1}{#2}{#3}
\negbigskip
\end{figure}
}

\newcommand{\OO}{\mathcal{O}}

\def\reg#1{\mathcal{#1}}

\newcommand{\Adv}{\mathop{\mathrm{Adv}^\pm}}

\usepackage{mathabx}
\newcommand{\maps}[1]{\;\stackrel{#1}{\longmapsto}\;}
\newcommand{\transduce}[1]{\stackrel{#1}{\rightsquigarrow}}
\newcommand{\DownTransduce}{\scalebox{0.7}{\rotatebox[origin=c]{270}{$\rightsquigarrow$}}}

\PassOptionsToPackage{svgnames}{xcolor}
\usepackage{tikz}
\usetikzlibrary{graphs, shapes.geometric}

\def\sS[#1]{\vcenter{\hbox{$\scriptstyle ($}} #1 \vcenter{\hbox{$\scriptstyle )$}}}

\begin{document}
\title{Taming Quantum Time Complexity}
\author{Aleksandrs Belovs}
\affiliation{Center for Quantum Computing Science, Faculty of Science and Technology, University of Latvia}
\orcid{0009-0004-1625-108X}
\author{Stacey Jeffery}
\affiliation{QuSoft, CWI \& University of Amsterdam}
\orcid{0000-0003-0046-5089}
\author{Duyal Yolcu}
\affiliation{\url{https://github.com/qudent}}

\maketitle

\begin{abstract}
    Quantum query complexity has several nice properties with respect to composition. First, bounded-error quantum query algorithms can be composed without incurring log factors through error reduction (\emph{exactness}). Second, through careful accounting (\emph{thriftiness}), the total query complexity is smaller if subroutines are mostly run on cheaper inputs -- a property that is much less obvious in quantum algorithms than in their classical counterparts. 
    While these properties were previously seen through the model of span programs (alternatively, the dual adversary bound), a recent work by two of the authors (Belovs, Yolcu 2023) showed how to achieve these benefits without converting to span programs, by defining \emph{quantum Las Vegas query complexity}. Independently, recent works, including by one of the authors (Jeffery 2022), have 
    worked towards bringing thriftiness to the more practically significant 
    setting of quantum \emph{time} complexity.

    In this work, we show how to achieve both exactness and thriftiness in the setting of time complexity. We generalize the quantum subroutine composition results of Jeffery 2022 so that, in particular, no error reduction is needed. We give a time complexity version of the well-known result in quantum query complexity, $Q(f\circ g)=\OO(Q(f)\cdot Q(g))$, without log factors.

We achieve this by employing a novel approach to the design of quantum algorithms based on what we call \emph{transducers}, and which we think is of large independent interest.
While a span program is a completely different computational model, a transducer is a direct generalisation of a quantum algorithm, which allows for much greater transparency and control.
Transducers naturally characterize general state conversion, rather than only decision problems;
 provide a very simple treatment of other quantum primitives such as quantum walks; and lend themselves well to time complexity analysis.
\end{abstract}

\clearpage
{
\tableofcontents
}
\clearpage

\pagestyle{plain}

\mycutecommand{\q}{{\mathrm{q}}}
\mycutecommand{\w}{{\mathrm{w}}}

\definecolor{applegreen}{rgb}{0.55, 0.71, 0.0}

\def\xicolor{red}
\def\taucolor{applegreen}
\def\witnesscolor{blue}
\def\nonquerycolor{gray}
\def\querycolor{orange}

\section{Introduction}

\label{sec:introPrior}
Since the introduction of span programs into quantum query complexity by Reichardt and \v Spalek
\cite{reichardt:formulae, reichardt:spanPrograms}%
\footnote{However, this particular connection is attributed to Troy Lee in the first of these two papers.},
the quantum query world became a nicer place to be.
A span program (alternatively, a feasible solution to the dual adversary bound) is an idealised computational model, in the sense that no real device corresponds to it.
Nonetheless, it has a strong connection to quantum query complexity.
We can turn any quantum query algorithm into a span program,%
    \footnote{Originally, this construction was non-constructive as it came from the dual of a lower bound.  Section 3 of~\cite{reichardt:spanPrograms} contains a constructive construction for algorithms with one-sided error.  This was later extended to two-sided error in~\cite{jeffery:spanProgramsSpace}.}
or we can construct one from scratch.
They can be composed as usual quantum subroutines.
At the end of the day, the span program can be transformed back into a quantum query algorithm.
We identify two main points of advantage of span programs compared to usual quantum subroutines.

\begin{itemize}
\item 
We call the first one \emph{exactness}.
A span program evaluates a function exactly even if the quantum query algorithm it was converted from had a bounded error.
Therefore, span programs can be composed without any error reduction.
The conversion from the span program to a quantum query algorithm does introduce an error, and we do \emph{not} get an exact quantum query algorithm at the end.
However, the error is introduced only \emph{once} per the whole algorithm, and does not accumulate even in the case of many non-precise subroutines.

This idea was the major driving force behind this line of research starting with the algorithm for the iterated NAND function~\cite{farhi:nandTree, ambainis:formulaeEvaluation} to general iterated Boolean functions~\cite{reichardt:formulae, reichardt:unbalancedFormulas}.

\item
We call the second one \emph{thriftiness}.
Span programs have a natural predisposition towards accurate bookkeeping.
If an execution of a subroutine happens to be cheap on a particular input, this cheaper cost will contribute to the total complexity of the whole span program.
This is in contrast to the usual execution of quantum subroutines, where the maximal cost is to be paid no matter how easy the input is in a particular execution.
This is because we generally cannot measure a subroutine, which may be run only in some branch of a superposition, to see whether it has ended its work or not.%
\footnote{
In \emph{specific} cases, Ambainis' variable-time framework~\cite{ambainis:searchVariableTimes, ambainis:amplificationVariableTimes, ambainis:variableTimeNew} allows for this particular approach.
We consider variable-time framework as a piece of motivation towards importance of thus defined thriftiness.  See, e.g., the results mentioned in the introduction of~\cite{ambainis:variableTimeNew}.
}
At the end, when converting to a quantum query algorithm, we still have to account for the maximal possible complexity, but this is done only \emph{once} per the whole algorithm, thus amortising complexities of individual subroutines.

This point has been emphasised to a lesser extent than the exactness property, but we find it equally important.
In particular, this ability results in some rather interesting super-polynomial speed-ups~\cite{zhan:treesWithHiddenStructure}.
We consider the thriftiness property in more detail in \rf{sec:conceptualPreliminaries}.
\end{itemize}

In a recent paper by a subset of the authors~\cite{belovs:LasVegas}, it was shown that one does not even have to change the model to obtain these benefits.
It suffices to keep the model of a quantum query algorithm, and only change the way query complexity is defined.
Namely, one can count total squared norm of all the states processed by the input oracle (coined quantum Las Vegas query complexity in~\cite{belovs:LasVegas} as it can be interpreted as the expected number of queries) instead of the number of executions of the input oracle (which is the usual definition, called Monte Carlo in~\cite{belovs:LasVegas}).
Note that quantum Las Vegas query complexity is an idealised complexity measure, but it has an operational meaning: at the end of the day, a quantum Las Vegas query algorithm can be turned into a Monte Carlo query algorithm with a constant increase in complexity and introducing a small error.

A major drawback of the above results is that while accounting of query complexity is very natural and important, analysis of the more meaningful \emph{time complexity} 
of the resulting algorithm can be quite difficult.
While the time complexity of some nicely structured span program algorithms has been successfully analyzed~(e.g.~\cite{belovs:learningClaws, reichardt:formulae}), in other cases, time complexity analyses have proven much more challenging, perhaps most notably in~\cite{belovs:learningKDist}, where a query algorithm was given, but a time complexity analysis remained elusive for more than a decade~\cite{jeffery:kDist}.

There has been a line of work attempting to extend some of the benefits of span programs to time complexity.
In \cite{cornelissen:spanProgramsTime}, 
it was shown how span programs could be made to capture time complexity of quantum algorithms, not only query complexity. 
Essentially what \cite{cornelissen:spanProgramsTime} did was to extend the algorithm-to-span-program conversion of~\cite{reichardt:spanPrograms} so that the span program still encodes information about the gate structure of the original algorithm.
Thus, this structure can be recovered when the span program is converted back into an algorithm.
The resulting span programs could be manipulated in a limited number of ways.
For instance, this allowed for an alternative implementation of Ambainis' variable-time search algorithm~\cite{ambainis:searchVariableTimes}.
This idea was further exploited in~\cite{jeffery:subroutineComposition} by one of the authors (in the framework of multidimensional quantum walks~\cite{jeffery:kDist}) to give a general composition result for time complexity of quantum algorithms, similar to the query results of~\cite{belovs:LasVegas}, but lacking in generality.
While~\cite{jeffery:subroutineComposition} achieved thriftiness, it failed to achieve exactness.  Instead, this work assumes that subroutines have had their success probability amplified through majority voting so that errors occur with inverse polynomial probability, which results in logarithmic factor overhead. 

One of the main contributions of this paper is a framework that achieves both exactness and thriftiness for quantum \emph{time} complexity in a systematic manner.
In order to do so, we introduce \emph{transducers}.
While quantum Las Vegas query complexity keeps the model of a quantum algorithm the same but changes the definition of the query complexity, transducers still keep the same model but change the definition of computation.
We consider two objects: a \emph{transducer} that is a usual quantum algorithm in a larger Hilbert space, and its \emph{transduction action} that is an idealised transformation obtained by ``additively tracing out'' a part of the Hilbert space of the transducer.
A simple procedure converts a transducer into an algorithm implementing its transduction action, introducing a small error in the process.
Monte Carlo query complexity of the resulting algorithm is Las Vegas query complexity of the transducer.
Its time complexity is a product of two things: time complexity of the transducer (considered as a usual quantum algorithm) and its \emph{transduction complexity}, which measures the extent of the ``traced-out'' part of the system. 
Compared to span programs, which most naturally model computations of decision problems,  transducers naturally model arbitrary state conversion computations, and as a transducer is itself a quantum computation, its time complexity analysis is more immediate.

For the composition results, we follow the same overall strategy as for span programs.
We can transform any quantum program into a transducer, and we can compose transducers in essentially the same way as we compose quantum programs.
The advantage is that both query complexity and time complexity (in the guise of transduction complexity) are composed in a thrifty manner.
At the end, we convert the composed transducer back into a quantum algorithm.
In parcitular, we obtain the following results:
\begin{description}
\item[Exact and thrifty composition of quantum algorithms:] We show how to compose quantum algorithms in a way that is thrifty\footnote{For an idea of what we mean by thrifty composition, see \rf{eqn:randomComposition} and \rf{eqn:randomComposition2}, and the surrounding discussion of the ``gold standard'' for composition.} in time complexity, but also exact (\rf{thm:introCompositionQRAG}). This improves on \cite{jeffery:subroutineComposition} in several ways.
	\begin{enumerate}
	\item Exactness: We do not need to assume the composed algorithms have inverse polynomial errors, which saves log factors in the overall complexity.
	\item The results apply to composing quantum algorithms that solve general state conversion problems, whereas \cite{jeffery:subroutineComposition} only applies to algorithms that decide Boolean functions.
	\item In addition to achieving thrifty time complexity, the composed algorithm is also thrifty in its number of queries to the input oracle.
	\end{enumerate}
The one way in which our results, as stated, are not more general than \cite{jeffery:subroutineComposition} is that \cite{jeffery:subroutineComposition} considers quantum algorithms that work in the variable-time model, meaning their running times can be random variables, achieved by performing intermediate partial measurements that indicate if the algorithm is finished. While variable-time algorithms are also compatible with transducers, we omit their explicit treatment here, for the sake of simplicity, and leave it for future work. 
We also give results that apply purely to the circuit model (without quantum random access gates (QRAGs)), but with worse complexity (\rf{thm:introCompositionCircuit}). 
We also explicitly consider multiple layers of composition (\rf{thm:introCompositionTree}) in the QRAG model, and its special case of iterated functions in the circuit model (\rf{thm:introIterated}).

\item[Quantum analogue of majority voting, and time-efficient function composition:] It\\ is known that the bounded-error quantum query complexity, $Q$, of a composed function $f \circ g$ is $Q(f\circ g)=\OO(Q(f)Q(g))$. This statement would be obvious with log factors, by using the standard technique of majority voting to reduce the error of each call to $g$ to inverse polynomial, but the fact that this holds without log factors is somewhat surprising. In \rf{sec:introPurifier}, we shed some light on this surprising result, and show a similar result for time complexity. We extend a notion from \cite{belovs:variations}, called a \emph{purifier}, to show how to take a quantum algorithm with constant error, and convert it to a transducer with any error $\eps>0$ with only \emph{constant} overhead on the query and time complexities. This, in turn, allows us to prove a log-factor-free composition result for time complexity in the QRAG model (\rf{thm:introCompositionFunctionQRAG}). Specifically, if there is an algorithm for $f$ that makes $Q$ queries and takes $T_f$ additional time, and an algorithm for $g$ that takes $T_g$ time, then there is an algorithm for $f\circ g$ that takes time $\OO(T_f+QT_g)$. As with the aforementioned query result, this would be obvious with log factors on the second term, but the fact that it holds without log factor overhead is surprising. This implies log factor improvements in the time complexity of quantum algorithms obtained by composition, such as~\cite{jeffery:kDist}.

\end{description}

We stress that in addition to our concrete results, a major contribution is conceptual. 
Along with 
bringing the beautiful existing work on quantum query complexity to the real world of quantum time complexity, transducers achieve many existing results in a much simpler and cleaner way, and we feel that they are a novel and potentially instructive way of understanding quantum algorithms. 

For example, discrete-time quantum walks are an important technical tool.
However, understanding their internal workings requires some background: the notion of the spectral gap, phase estimation subroutine, non-trivial spectral analysis.
In \rf{sec:introWalks} we use transducers to devise a very simple implementation of a quantum walk that completely avoids all this background.

Let us end this section with a few remarks on the model of computation.
The previous results towards thriftiness in time complexity,~\cite{cornelissen:spanProgramsTime} and~\cite{jeffery:subroutineComposition}, assumed the QRAG model in a fundamental way.
It is based around the quantum random access gate (QRAG), which allows index access to an array of \emph{quantum} registers in superposition.
This should not be confused with quantum random access memory (QRAM), which assumes such access to an array of \emph{classical} registers.%
\footnote{This nomenclature is not completely standardised and has been a source of confusion, as different authors use the same name to refer to different models.}
QRAG is a stronger model than QRAM, but it is utilised in a large number of time-efficient quantum algorithms, Ambainis' element distinctness algorithm~\cite{ambainis:distinctness} being a noticeable example.

Our general attitude towards the QRAG model is one of ambivalence.
On the one hand, the QRAG model is a very natural quantization of the classical RAM machine, and, thus, makes a lot of sense to study theoretically.
On the other hand, the assumption that we can swap a large number of qubits in superposition in essentially constant time seems far-fetched given the current state of the art in development of quantum computers.
Because of that, we have chosen to pursue both directions.  
We prove results involving the QRAG model, as they tend to have more natural formulations, and continue the aforementioned line of research in~\cite{cornelissen:spanProgramsTime, jeffery:subroutineComposition}, but we also design algorithms in the circuit model, which, while not being as efficient as in the QRAG model, are of significant interest.

We now give a technical overview of our results in \rf{sec:conceptualPreliminaries} and \rf{sec:overview}, before fleshing out full details in the remaining sections.

\section{Conceptual Preliminaries: Quantum Las Vegas Complexity}
\label{sec:conceptualPreliminaries}
Quantum Las Vegas query complexity~\cite{belovs:LasVegas}, also known as the total query weight~\cite{jeffery:subroutineComposition}, is a cornerstone in understanding this paper for a number of reasons.
First, the definition of the query complexity of a transducer is intrinsically the Las Vegas one.
Second, when converting a quantum program into a transducer, we take into account its Las Vegas query complexity.
Finally, we use composition results for the Las Vegas \emph{query} complexity as a model that we strive to achieve for \emph{time} complexity.

This section serves as a very brief overview of the composition properties for randomised and quantum Las Vegas complexity.
The main goal of this section is to highlight the connection between quantum and randomised Las Vegas complexity, and to explain the composition properties we are interested in this paper.
More technical exposition of these topics will be done in Sections~\ref{sec:preliminaries} and~\ref{sec:properties}, respectively.

The randomised results in this section are folklore, and the quantum ones are either from~\cite{belovs:LasVegas} or can be obtained similarly.
Both are for purely illustrative purpose and will serve as reference points to the results obtained later in the paper.
The reader may choose to skip to \rf{sec:overview}, and to return to this section if needed.

\subsection{Types of Composition}
\label{sec:conceptualTypes}
The composition property we have in mind is as follows.
Assume we have an algorithm $A$ with an oracle $O'$, and an algorithm $B$ that implements the oracle $O'$.
The functional composition of the two is an algorithm $A\circ B$, where the algorithm $B$ is used as a subroutine to process the queries made by $A$ to $O'$.
The subroutine $B$ has access to some oracle $O$, which is also the input to $A\circ B$.
For simplicity we will now assume that $A$ has no access to $O$, but we will drop this assumption shortly.
The main question is a bound on the complexity of $A\circ B$ in terms of complexities of $A$ and $B$.

One particular example studied extensively both classically and quantumly is given by composed functions, i.e., functions of the form
\begin{equation}
\label{eqn:randomComposedFunction}
f\sA[
g_1(z_{1,1}, \dots,z_{1,m}), 
g_2(z_{2,1}, \dots,z_{2,m}),
\dots,
g_n(z_{n,1}, \dots,z_{n,m})].
\end{equation}
(For simplicity of notation, we assume all the functions $g_i$ are on $m$ variables, which is without loss of generality.)
Then, $O$ encodes the input string $z$, and $A$ evaluates the function $f$.
Concerning $B$, it is a parallel composition $B=\bigoplus_i B_i$ of algorithms evaluating $g_i$.
On query $i$, $B$ returns the output of $B_i$, which is $g_i(z_{i,1}, \dots,z_{i,m})$.

In general, the parallel composition $B=\bigoplus_i B_i$, which we also call the direct sum, has queries of the form $(i,j)$, on which it responds with the output of $B_i$ on query $j$.
We assume that all $B$ and $B_i$ have access to the same oracle $O$.
In the example above, all $j$ are absent.
Also, although each $B_i$ has access to $O$, it only uses the substring $z_i = (z_{i,1}, \dots,z_{i,m})$ thereof.

\subsection{Randomised Complexity}

Recall the distinction between Las Vegas and Monte Carlo randomised algorithms. A Monte Carlo algorithm is a randomised algorithm that takes at most some fixed number $T$ of time steps, and outputs the correct answer with bounded error. 
The complexity of such an algorithm is simply $T$. This is analogous to the usual \emph{bounded-error} quantum query model. 
In contrast, a Las Vegas algorithm is a randomised algorithm whose number of steps is a random variable $T$, and that never outputs an incorrect answer (though it may run forever). The complexity of such an algorithm is $\E[T]$. More generally, one might consider the Las Vegas complexity, $\E[T]$, of an algorithm whose running time is a random variable $T$, even if the algorithm has some probability of erring (that is, it is not strictly a Las Vegas algorithm). 

Randomised Las Vegas complexity behaves nicely with respect to composition.
Here we sketch the corresponding notions and results in order to facilitate the forthcoming introduction of the related quantum notions, and to have a reference point against which we can gauge our quantum composition results.

\mycutecommand{\rand}{_{\mathrm{rand}}}

Let us start with functional composition.
Let $T\rand(A,O')$ be the complexity of the algorithm $A$ on the oracle $O'$.
Let $B(O)$ denote the action of the algorithm $B$ on the oracle $O$, and $T\rand(B, O, i)$ the complexity of $B(O)$ on query $i$.
Denote by $p^{(i)}\rand(A,O')$ the probability the algorithm $A$ will give $i$ as a query to the oracle $O'$ at some point during the execution of the algorithm (we assume each query is given at most once).
Then, it is not hard to see that the total complexity of $A\circ B$ on oracle $O$ is given by
\begin{equation}
\label{eqn:randomComposition}
T\rand(A\circ B,O) = T\rand\sA[A, B(O)] + \sum_i p^{(i)}\rand\sA[A, B(O)]\cdot T\rand(B, O, i).
\end{equation}

Let us rewrite this formula to illuminate transition to the quantum case.
Define the total query vector of the algorithm $A$ as a formal linear combination 
\[
q\rand(A, O') = \sum_i p^{(i)}\rand (A,O')\; e_i
\]
for $e_i$ the elements of the standard basis.
Let by definition the complexity of $B$ on such a vector be
\begin{equation}
\label{eqn:randomComplexityExtended}
T\rand\sB[B, O, \sum_i p_i e_i] = \sum_i p_i\, T\rand\sA[B, O, i].
\end{equation}
Then, we can rewrite~\rf{eqn:randomComposition} as
\begin{equation}
\label{eqn:randomComposition2}
T\rand(A\circ B, O) =  T\rand\sA[A,B(O)] + T\rand\sA[B, O, q\rand(A,B(O))],
\end{equation}
which says that the complexity of the composed program on the oracle $O$ equals the complexity of $A$ by itself plus the complexity of $B$ on its total query vector.
This equation will serve as our ``gold standard'' of thriftiness in functional composition.

Similarly, thriftiness holds for parallel composition.
Let $\xi = \sum_{i,j} \xi_{i,j} e_{i,j}$ be a query vector for $\bigoplus_i B_i$.
We can break it down as
\begin{equation}
\label{eqn:randomXiDecomposition}
\xi = \bigoplus_i \xi_i,
\end{equation}
where $\xi_i = \sum_j \xi_{i,j} e_{i,j}$ is the corresponding query vector for the constituent $B_i$.
Then, we have
\begin{equation}
\label{eqn:randomParallelComplexity}
T\rand\sB[\bigoplus_i B_i, O, \xi] = \sum_{i} T\rand\sA[B_i, O, \xi_i]. 
\end{equation}
(We assume here that relaying from $B$ to the corresponding $B_i$ is done instantly.)

These results can be combined in various ways.
For instance, assume the oracle $O'$ in the functional composition settings is a direct sum $O'=\bigoplus_{i=1}^n O^{(i)}$ of $r$ independent oracles $O^{(i)}$.
Let $B_i$ implement $O^{(i)}$, so that $B = \bigoplus_i B_i$ implements $O'$.
Similarly as in~\rf{eqn:randomXiDecomposition}, we can decompose the corresponding total query vector
\[
q\rand(A, O') = \bigoplus_i q\rand^{(i)}(A,O')
\]
into partial query states corresponding to the $i$-th oracle.
In this way, the multiple-oracle case differs from the single-oracle case only by this change of perspective.
Combining~\rf{eqn:randomComposition2} and~\rf{eqn:randomParallelComplexity}, we obtain
\begin{equation}
\label{eqn:randomComposition3}
T\rand(A\circ B, O) = T\rand\sA[A,B(O)] + \sum_i T\rand\sA[B_i, O, q\rand^{(i)}(A,B(O))].
\end{equation}

\mycutecommand\Ic{I^\circ}
\mycutecommand\Ib{I^\bullet}
\mycutecommand\Iw{I^\circ}
\mycutecommand\cHw{\cH^\circ}
\mycutecommand\cHq{\cH^\bullet}
\mycutecommand\cHt{\cH^\uparrow}

\subsection{Quantum Complexity}
\label{sec:conceptualQuantum}
The usual definition of quantum query algorithms does not allow for different complexities on different inputs, hence, we cannot even properly define a meaningful analogue of~\rf{eqn:randomComposition}.
However, it is possible for \emph{Las Vegas} complexity of quantum query algorithms defined in~\cite{belovs:LasVegas}.
We will demonstrate here that it possesses properties essentially identical to those of the randomised case.
For a more formal definition of the model of query algorithms and the Las Vegas complexity, refer to \rf{sec:prelimQuery}.

Let $A = A(O)$ be a quantum algorithm in space $\cH$ with an oracle $O$ in space $\cM$.
We denote by $A(O)$ the action of $A$ on a specific input oracle $O$.
We use $T(A)$ to denote time complexity of $A$, and $Q(A)$ to denote the number of queries made by $A$.

What is important, is that the queries to the input oracle $O$ are conditional.
This means that the query applies $O$ to a subspace $\cHq\subseteq\cH$ of the workspace $\cH$, and is the identity elsewhere.
The total query state $q(A,O,\xi)$ records the history of all the queries given by the algorithm $A$ to the input oracle $O$ on the initial state $\xi$.
In other words, it is the direct sum
\[
q(A,O,\xi) = \bigoplus_{t=1}^{Q(A)} \psi^\bullet_t = \sum_{t=1}^{Q(A)} \ket|t> \ket |\psi^\bullet_t>,
\]
where $\psi^\bullet_t\in \cHq$ is the state given to the input oracle on the $t$-th query.
We have $q(A,O,\xi) \in \cE\otimes \cM$ for some space $\cE$.
The Las Vegas query complexity is defined as $L(A,O,\xi) = \|q(A,O,\xi)\|^2$.
This definition has an operational meaning: the algorithm can be modified to use $\OO(L)$ queries, where $L$ is the worst-case Las Vegas query complexity, by introducing some small error.

One convention of this paper is that we usually only allow \emph{unidirectional} access to $O$.  The algorithm can only execute $O$, but not its inverse $O^*$.
This is without loss of generality, as \emph{bidirectional access} to $O$, when the algorithm can execute both $O$ and $O^*$, is equivalent to unidirectional access to $O\oplus O^*$.

As in the randomised case, $r$ input oracles can be combined into one as follows:
\begin{equation}
\label{eqn:introMultipleOracles}
O = O^{(1)}\oplus O^{(2)}\oplus \cdots \oplus O^{(r)}.
\end{equation}
The \emph{partial query state} $q^{(i)}(A,O,\xi)$ of the $i$-th input oracle is defined as the direct sum of all the states given to that particular oracle $O^{(i)}$.
In particular,
$
q(A, O, \xi) = \bigoplus_i q^{(i)} (A, O,\xi).
$
Similarly as before, the Las Vegas query complexity of the $i$-th input oracle is 
$
L^{(i)}(A, O, \xi) = \normA|q^{(i)}(A,O,\xi)|^2.
$

\paragraph{Parallel Composition}
The parallel composition is straightforward.
For programs $B_1,\dots, B_n$, all on the input oracle $O$, its direct sum is $\bigoplus_i B_i$, which executes $B_i$ on orthogonal parts of the space.
It is not hard to show that
\begin{equation}
\label{eqn:quantumParallel}
q\sB[\bigoplus_i B_i, O, \bigoplus_i \xi_i] = \bigoplus_i q(B_i, O, \xi_i),
\end{equation}
where we implicitly assume the correct arrangement of the entries in the corresponding direct sums.
A direct consequence is the following counterpart of~\rf{eqn:randomParallelComplexity}:
\[
L\sB[\bigoplus_i B_i, O, \bigoplus_i \xi_i] = \sum_i L(B_i, O, \xi_i).
\]

\paragraph{Functional Composition}
Let us now derive a quantum query analogue of~\rf{eqn:randomComposition2}.
First, though, we define a counterpart of~\rf{eqn:randomComplexityExtended}.
For $\xi'\in \cE\otimes \cH$, we can write $\xi' = \xi_1\oplus \xi_2 \oplus \cdots \oplus \xi_m$ with $\xi_t\in \cH$, and define
\begin{equation}
\label{eqn:queryStateExtended}
q(A,O,\xi') = \bigoplus_t q(A, O, \xi_t).
\end{equation}
This is precisely the total query state we will get if we tensor-multiply $A$ by the identity in the register $\cE$ and execute it on $\xi'$.
The corresponding Las Vegas query complexity is 
\begin{equation}
\label{eqn:LasVegasExtended}
L(A, O, \xi') = \|q(A, O, \xi')\|^2.
\end{equation}

Now consider \rf{fig:composition}, which depicts composition of two quantum programs, and which goes along the lines of our previously discussed randomised case.
This time, however, we consider a more general case when $A$ has access to the input oracle $O$ as well.

\myfigure{\label{fig:composition}}
{
General case of functional composition of two programs.
The outer program $A$ has two oracles $O$ and $O'$, which we identify with the oracle $O\oplus O'$.
The oracle $O'$ is implemented by a program $B$ with access to $O$.
The diagram specifies the corresponding query states, where we use the upper indices ${}^{(0)}$ and ${}^{(1)}$ in relation to $O$ and $O'$, respectively.
}
{
\negbigskip
\[
\begin{tikzpicture}[every path/.append style={thick,->}]
    \draw[\xicolor] (-1, 1) node[left]{$\xi$} to (0,1);
    \draw (0,0) rectangle (3,1.5) node[pos=0.5] {$A(O\oplus O')$};
    \draw (4,0) rectangle (7,-1.5) node[pos=0.5] {$O' = B(O)$};
    \draw (0, -3) rectangle (7,-4.5) node[pos=0.5] {\Large $O$}; 
    \draw[purple] (1.5,0) to node[left] {\scriptsize $q^{(0)}(A, O\oplus O', \xi)$} (1.5,-3);
    \draw[purple, out=0, in=90] (3, 1) to node[above]{\scriptsize $q^{(1)}(A, O\oplus O', \xi)$} (5.5,0);
    \draw[purple] (5.5, -1.5) to node[right] {\scriptsize $q\sB[ B, O, q^{(1)}(A, O\oplus O', \xi)]$}(5.5,-3);
\end{tikzpicture}
\]
\negbigskip
}

The composed algorithm $A\circ B$ is implemented by replacing each execution of $O'$ by an execution of $B$.
Its action on the input oracle $O$ is equal to $A\sA[O\oplus B(O)]$.
It is not hard to show that
\begin{equation}
\label{eqn:quantumComposition}
q(A\circ B, O, \xi) = q^{(0)}\sB[A, O\oplus B(O), \xi]\; \oplus\; q\sB[ B, O, q^{(1)}{\sA[A, O\oplus B(O), \xi]}].
\end{equation}
Similarly to~\rf{eqn:randomComposition2}, this slightly complicated expression represents a very intuitive observation that the total query state of $A\circ B$ on the input oracle $O$ consists of the part of the query state of $A$ given directly to $O$ (denoted $q^{(0)}$), together with the query state of $B$ on the initial state composed of the part of the query state of $A$ given to $O'$ (denoted $q^{(1)}$). 

\myfigure{\label{fig:compositionMultiple}}
{
A version of \rf{fig:composition} where the oracle $O'$ is decomposed into $r$ input oracles $O' = O^{(1)}\oplus\cdots\oplus O^{(r)}$ and the program $B$ is accordingly decomposed into $B = B_1\oplus\cdots\oplus B_r$ so that $B_i$ implements $O^{(i)}$.
The diagram specifies the corresponding query states, where we use the upper index ${}^{(0)}$ in relation to $O$.
We note that it is without loss of generality to assume that $A$ and all $B_i$ use the same input oracle $O$.
Indeed, if this is not the case, we can define $O$ as the direct sum of the oracles used by $A$ and $B_i$.
}
{
\negbigskip\negbigskip\negbigskip
\[
\begin{tikzpicture}[every path/.append style={thick,->}]
    \draw[\xicolor] (-1, 2) node[left]{$\xi$} to (0,2);
    \draw (0,1) rectangle (3,2.5) node[pos=0.5] {$A(O\oplus O')$};
    \draw (3,0) rectangle (6,-1.5) node[pos=0.5] {$O^{(1)} = B_1(O)$};
    \node at (7.5, -0.75) {\Large$\cdots$};
    \draw (8.5,0) rectangle (11.5,-1.5) node[pos=0.5] {$O^{(r)} = B_r(O)$};
    \draw (0, -3) rectangle (11.5,-4.5) node[pos=0.5] {\Large $O$}; 
    \draw[purple] (1.5,1) to node[left] {\scriptsize $q^{(0)}(A, O\oplus O', \xi)$} (1.5,-3);
    \draw[purple, out=0, in=90] (3, 1.5) to node[right, pos=0.8]{\scriptsize $q^{(1)}(A, O\oplus O', \xi)$} (4.5,0);
    \draw[purple, out=0, in=90] (3, 2.25) to node[right, pos=0.95]{\scriptsize $q^{(r)}(A, O\oplus O', \xi)$} (10.5,0);
    \draw[purple] (4.5, -1.5) to node[right] {\scriptsize $q\sB[ B_1, O, q^{(1)}(A, O\oplus O', \xi)]$}(4.5,-3);
    \draw[purple] (10, -1.5) to node[right] {\scriptsize $q\sB[ B_r, O, q^{(r)}(A, O\oplus O', \xi)]$}(10,-3);
\end{tikzpicture}
\]
\negbigskip
}

It is quite often the case, that the oracle $O'$ above is composed of several input oracles, each implemented by its own subroutine $B_i$, see \rf{fig:compositionMultiple}.
In this case, we can use~\rf{eqn:quantumParallel} to obtain an analogue of~\rf{eqn:randomComposition3}:
\begin{equation}
\label{eqn:quantumCompositionMultiple}
q(A\circ B, O, \xi) = q^{(0)}\sB[A, O\oplus B(O), \xi]\; \oplus\; \bigoplus_i q\sB[B_i, O, q^{(i)}{\sA[A, O\oplus B(O), \xi]}].
\end{equation}

It is often convenient to define $L_{\max}(B,O)$ as the worst-case complexity of $L(B,O,\xi)$ as $\xi$ ranges over all unit vectors (or over all unit vectors in some admissible subspace of initial vectors).
Then, using linearity, we can obtain from~\rf{eqn:quantumComposition}: 
\begin{align}
L(A\circ B, O, \xi) 
&= L^{(0)}\sB[A, O\oplus B(O), \xi]\; +\; L\sB[ B, O, q^{(1)}{\sA[A, O\oplus B(O), \xi]}] \label{eqn:quantumComposition1}\\
&\le L^{(0)}\sB[A, O\oplus B(O), \xi]\; +\; L_{\max}(B,O) L^{(1)}\sA[A, O\oplus B(O), \xi]
\label{eqn:quantumComposition2}.
\end{align}
and from~\rf{eqn:quantumCompositionMultiple}
\begin{align}
L(A\circ B, O, \xi) 
&= L^{(0)}\sB[A, O\oplus B(O), \xi]\; +\; \sum_i L\sB[B_i, O, q^{(i)}{\sA[A, O\oplus B(O), \xi]}] \label{eqn:quantumCompositionMultiple1}\\
&\le L^{(0)}\sB[A, O\oplus B(O), \xi]\; +\; \sum_i L_{\max}(B_i,O) L^{(i)}\sA[A, O\oplus B(O), \xi]
\label{eqn:quantumCompositionMultiple2}.
\end{align}
In~\rf{eqn:quantumComposition2} and~\rf{eqn:quantumCompositionMultiple2}, it is assumed that $B$ and $B_i$ are only executed on the admissible initial states.

To summarise, we see that quantum Las Vegas query complexity satisfies composition properties very similar to the ``gold standard'' of the randomised Las Vegas complexity.
One of the goals of this paper is to approach these results for quantum \emph{time} complexity.
In~\cite{jeffery:subroutineComposition}, a result in this direction was obtained for evaluation of functions assuming the QRAG model of computation.
In this paper, we consider more general state conversion settings, and also obtain partial results for the circuit model.
We also think that the approach of this paper is less technical than the one taken in in~\cite{jeffery:subroutineComposition}.

\section{Overview of the Paper}
\label{sec:overview}
This section serves as an informal version of the whole paper, where we introduce all the main concepts, ideas, and sketch the proofs of the main results.
In the remaining paper, we fill in all the technical gaps.

\subsection{Transducers}
\label{sec:introTransducers}

In the current paper, we take a different approach to time complexity than in the two papers~\cite{cornelissen:spanProgramsTime, jeffery:subroutineComposition} mentioned in \rf{sec:introPrior}.
Instead of using span programs or quantum walks, we build on the key technical primitive from~\cite{belovs:LasVegas}, which we call a \emph{transducer}\footnote{Not to be confused with transducers from the theory of finite automata, however there are some connections between the two, as we discuss in \rf{sec:automaton}.} 
in this paper.

Transducers are based on the following mathematical observation 
we prove in \rf{sec:transducerDefinition}:
\begin{thm}[Transduction]
\label{thm:introTransduce}
Let $S$ be a unitary acting in a direct sum of two vectors spaces $\cH\oplus \cL$.
For every $\xi\in \cH$, there exist a unique $\tau = \tau(S,\xi)\in \cH$ and in some sense unique $v = v(S,\xi)\in \cL$ such that
\begin{equation}
\label{eqn:1transduce}
S\colon \xi \oplus v \mapsto \tau \oplus v.
\end{equation}
\end{thm}

We say in the setting of~\rf{eqn:1transduce} that $S$ \emph{transduces} $\xi$ into $\tau$, denoted $\xi\transduce{S} \tau$ or $S\colon\xi\transduce{} \tau$.
This defines a mapping $\xi\mapsto \tau$ on $\cH$, which turns out to be unitary.
We call it the \emph{transduction action} of $S$ on $\cH$, denoted by $S\DownTransduce_\cH$.
See~\rf{fig:transducer} for a schematic depiction.
The motivation behind this terminology is that while $S$ does \emph{not} literally map $\xi$ into $\tau$, having $S$ is a legitimate and fruitful way of implementing $S\DownTransduce_\cH$ on a quantum computer, as we show shortly.
If a unitary $S$ is designed primarily with this application in mind, we call it a \emph{transducer}, and say that it \emph{implements} $S\DownTransduce_\cH$. 

\myfigure{\label{fig:transducer}}
{
Schematic depiction of transducers.  To the left is the real action of $S$, which is interpreted as the action of $S\DownTransduce_\cH$ on $\cH$.\\
Note that parallel wires here denote direct sum of the corresponding subspaces, not tensor product.  The same applies to the other figures in this paper.
}
{
\negbigskip
\[
\begin{tikzpicture}
    \draw (0,0) rectangle (1, 2) node[pos=0.5] {$S$};
    \draw[->,\witnesscolor, thick] (-1,0.5) node[above]{$v$} to (0,0.5);
    \draw[->,\xicolor, thick] (-1,1.5) node[above]{$\xi$} to (0,1.5);
    \draw[->,\witnesscolor, thick] (1,0.5) to (2,0.5) node[above]{$v$};
    \draw[->,\taucolor, thick] (1,1.5) to (2,1.5) node[above]{$\tau$};
    \draw (6,1) rectangle (7, 2) node[pos=0.5] {$S\DownTransduce_\cH$};
    \draw[->,\xicolor, thick] (5,1.5) node[above]{$\xi$} to (6,1.5);
    \draw[->,\taucolor, thick] (7,1.5) to (8,1.5) node[above]{$\tau$};
\end{tikzpicture}
\]
\negmedskip
}

Let us note that this construction is not new.
It has appeared before as an additive trace in the category of isometries in finite-dimensional Hilbert spaces~\cite{bartha:quantumTuringAutomata, AndresMartinez:phd}.
Here we demonstrate that this construction has an operational meaning. 
It would be interesting to understand the connection between our construction and the one taken in these two references.

We will stick to the following terminology.
We call $\cH$ the \emph{public} and $\cL$ the \emph{private} space of $S$.
We say that the transducer $S$ is \emph{on} the space $\cH$, but works \emph{in} the space $\cH\oplus \cL$.
Also, we will call $\xi$ the \emph{initial state} of $S$, while $\xi\oplus v$ is 
the \emph{initial coupling}.
We call $v$ the \emph{catalyst} of the transduction~\rf{eqn:1transduce} because it helps in the transformation of $\xi$ into $\tau$, but is not changed in the process.
The role of the catalyst is similar to the role of the witness in span programs.
In particular, the \emph{transduction complexity} of the transducer $S$ on an initial vector $\xi\in \cH$ is given by its size:
\begin{equation}
\label{eqn:introWitnessSize}
W(S,\xi) = \|v(S,\xi)\|^2.
\end{equation}

Let $T(S)$ denote the usual time complexity of implementing $S$ as a unitary.
As a rule of thumb, among various transducers $S$ with the same transduction action, there is a trade-off between $W$ and $T$ so that the product
\begin{equation}
\label{eqn:tradeoff}
(1+W(S,\xi)) \;\cdot\; T(S)
\end{equation}
stays approximately the same.
The importance of this product can be readily seen from the following result.
\begin{thm}[Implementation of Transducer]
\label{thm:introImplementation}
Let spaces $\cH, \cL$, and parameters $W, \eps>0$ be fixed.
There exists a quantum algorithm that $\eps$-approximately transforms $\xi$ into $S\DownTransduce_\cH\, \xi$ for all transducers $S\colon \cH\oplus \cL \to \cH\oplus \cL$ and initial states $\xi\in\cH$ such that $W(S,\xi)\le W$.
The algorithm conditionally executes $S$ as a black box $K = \OO(1+W/\eps^2)$ times, and makes $\OO(K)$ other elementary operations.
\end{thm}

Since $S$ generally takes at least one elementary operation, the complexity of the algorithm is dominated by the executions of~$S$, which takes time~\rf{eqn:tradeoff} up to constant factors (assuming $\eps = \Theta(1)$).
The term 1 in the definition of $K$ is required as we can have non-trivial transducers with $W=0$, see, e.g., \rf{sec:introAlgorithm->Transducer}.
Also, as follows from the discussion in~\cite{belovs:LasVegas}, the dependence on $\eps$  is optimal.

\pfstart[Proof sketch of \rf{thm:introImplementation}]
We are given a copy of $\xi$, and our goal is to transform it into $\tau = S\DownTransduce_{\cH} \xi$ using $S$ as a black box.
Assume we are additionally given a copy of $v$ (in the sense of direct sum, \emph{cf.} \rf{fig:transducer}).  
Then, we can perform the required transformation $\xi \oplus v \mapsto \tau\oplus v$ using $S$.

There are two problems here.
First, the algorithm is not given $v$, and, second, $v$ can have a huge norm.
The second problem is solved by breaking $\xi$ down into $K$ copies of $\xi/\sqrt K$, that is, performing the transformation $\xi \mapsto \sum_{t=0}^{K-1} \ket |t>\ket|\xi>/\sqrt K$.
The key idea is that we can use the same scaled down catalyst $v/\sqrt{K}$ to perform $K$ scaled down transductions $\xi/\sqrt{K} \transduce{} \tau/\sqrt{K}$ as $v/\sqrt{K}$ does not change in the process.
See \rf{fig:pumping}, for an illustration.

The first problem is solved by ``guessing'' $v/\sqrt{K}$, i.e., using that $\xi$ is close to $\xi\oplus v/\sqrt{K}$ if $K$ is sufficiently large.
The larger the value of~$K$, the smaller the error imposed by guessing, but the larger the number of executions of $S$.
For a formal proof, see \rf{sec:implementation}.
\pfend

\myfigure{\label{fig:pumping}}
{
A graphical illustration of the construction of \rf{thm:introImplementation}.
The initial state $\xi$ is broken down into $K=4$ copies of $\xi/\sqrt{K}$, which are sequentially transformed into $\tau/\sqrt K$ using only one copy of the scaled-down catalyst $v/\sqrt{K}$.
}
{
\negbigskip
\def\witnessheight{3.5}
\def\inputheight{4.5}
\def\xiheight{6}
\tikzset{witness/.style={above, node contents={\scriptsize $\frac{v}{\sqrt{K}}$}}}
\newcommand{\OneIteration}[1]{
    \edef\indxx{#1}
    \begin{scope}[shift={(3*\indxx,0)}]
        \draw (0.5,3) rectangle (1.5, 5) node[pos=0.5] {\Large $S$};
        \draw (0.2,\inputheight+0.4) node[\xicolor] {\scriptsize $\frac{\xi}{\sqrt{K}}$} ;
        \draw (1.8,\inputheight+0.4) node[\taucolor] {\scriptsize $\frac{\tau}{\sqrt{K}}$} ;
        \draw[->, thick, \witnesscolor] (1.5,\witnessheight) to node[witness]{} (3.5,\witnessheight);
    \end{scope}
    \draw[->, thick, \xicolor] (-1, \xiheight)  .. controls (3*\indxx, \xiheight) and (3*\indxx-1, \inputheight) .. (3*\indxx+0.5,\inputheight);
    \draw[<-, thick, \taucolor] (12, \xiheight)  .. controls (3*\indxx+2, \xiheight) and (3*\indxx+3, \inputheight) .. (3*\indxx+1.5,\inputheight);
}
\[
\begin{tikzpicture}
\draw[thick, \xicolor] (-1.5,\xiheight) to node[above]{$\xi$} (-1,\xiheight);
\draw[thick, \taucolor] (12.5,\xiheight) to node[above]{$\tau$} (12,\xiheight);
\draw[->, thick, \witnesscolor] (-1.5,\witnessheight) to node[witness]{} (0.5,\witnessheight);
\OneIteration{0}
\OneIteration{1}
\OneIteration{2}
\OneIteration{3}
\end{tikzpicture}
\]
}

\subsection{Connection to Quantum Walks}
\label{sec:introWalks}

In this section, we take a short detour, and inspect the connection between transducers and quantum walks.
Like transducers, quantum walks~\cite{santha:walkBasedAlgorithms} replace the desired transformation with some other transformation that is easier to implement: one iteration of the quantum walk.
See \rf{fig:trans2walk} for an informal comparison between transducers and quantum walks, on which we will elaborate in this section.

\myfigure{\label{fig:trans2walk}}
{
An informal correspondence between transducers and quantum walks.
The main point of this comparison is to give some intuition about the roles of different objects from \rf{sec:introTransducers} by linking them to objects with similar functions in the context of quantum walks.
This comparison is purely indicative.
}
{
\negmedskip
\[
\begin{tabular}{r>{$\longleftrightarrow$}cl}
Transducer&&Quantum Walk\\
Execution of $S$&&One Iteration $R_2R_1$\\
Transformation $S\DownTransduce_{\cH}$ && Accept/reject of the initial state\\
$W(S,\xi)$ && Spectral Gap\\
\rf{thm:introImplementation} && Phase Estimation
\end{tabular}
\]
}

In this paper, we consider broadly interpreted discrete-time quantum walks.
We identify the two characteristic properties of such algorithms, where the first one is essential, and the second one is usual, but not, strictly speaking, necessary.

The first, essential property is that one iteration of a quantum walk is a product of two reflections $R_1$ and $R_2$.
The quantum walk either rejects 
the initial state $\xi$, when it is close to an eigenvalue-1 eigenvector of $R_2R_1$; 
or accepts it, when $\xi$ is mostly supported on eigenvectors with eigenvalues far from 1. 
The most standard implementation of quantum walks is a phase estimation of the product $R_2R_1$ on the initial state $\xi$.
The analysis of quantum walks involves spectral analysis, sometimes assisted by the effective spectral gap lemma~\cite{lee:stateConversion}.

The second, optional property is that each reflection, $R_1$ and $R_2$, is broken down as a product of local reflections that act on pairwise orthogonal subspaces.
This allows for their efficient implementation.
The walk is usually described using a bipartite graph (like in \rf{fig:walk}), where each edge corresponds to a portion of the space.
Local reflections are given by vertices, and they act on the direct sum of spaces corresponding to their incident edges.
The reflection $R_1$ executes all the local reflections for one part of the bipartite graph, and $R_2$ for the second. 
The two reflections, interleaving, transcend locality and form an involved global transformation.

Because of the second property, quantum walks find a large number of algorithmic applications.
This includes such basic primitives as Grover's algorithm~\cite{grover:search} and amplitude amplification~\cite{brassard:amplification}; as well as the element distinctness algorithm~\cite{ambainis:distinctness}, Szegedy quantum walks~\cite{szegedy:walk} and their various extensions, span programs~\cite{reichardt:advTight}, learning graphs~\cite{belovs:learning}, and others.

We will show that quantum walks are very often transducers of the following form.
Let $\cH$ and $\cL$ be the public and the private spaces, so that the initial state $\xi\in\cH$.
The transducer is the iteration of the walk: $S = R_2R_1$, where we additionally assume that the second reflection $R_2$ acts trivially on $\cH$.
If $\xi$ is negative, we have the following chain of transformations:
\begin{equation}
\label{eqn:QW_sequence_negative}
\xi \oplus v \maps{R_1} \xi \oplus v \maps{R_2} \xi \oplus v,
\end{equation}
certifying that $\xi\transduce{S}\xi$.
In the positive case, we have the following sequence of transformations:
\begin{equation}
\label{eqn:QW_sequence_positive}
\xi \oplus v \maps{R_1} -\xi \oplus -v \maps{R_2} -\xi \oplus v,
\end{equation}
certifying that $\xi\transduce{S}-\xi$.

Both sequences follow the standard practice of designing quantum walks.
In the negative case, $\xi\oplus v$ is a stationary vector of both $R_1$ and $R_2$, and, hence, $R_2R_1$.
In the positive case, $\xi\oplus v$ is the witness for the Effective Spectral Gap Lemma.
The transduction vantage point unites these asymmetric positive and negative analyses. 
An interesting artefact of this construction is that the transduction action of the corresponding quantum walk is exact: It transduces $\xi$  to either $\xi$ or $-\xi$ exactly.

We will illustrate this construction in more detail in \rf{sec:walks} by re-proving the main result of~\cite{belovs:electicityQuantumWalks} on electric quantum walks.
However, the same applies to any quantum walk that adheres to the same design principles, including algorithms derived from span programs \cite[Section 3.4]{belovs:phd}, and more generally, algorithms of the type formally defined in \cite[Section 3.1]{jeffery:kDist}.

To epitomise, we keep the iteration of the quantum walk intact, but replace the wrapping phase estimation by the algorithm of \rf{thm:introImplementation}.
In this way, we significantly simplify the construction by abandoning \emph{any} spectral analysis both from the implementation and the analysis of the algorithm.
We believe this is worthy from the pedagogical point of view, as the corresponding algorithms now require very little background knowledge.

\mycutecommand{\vw}{v^\circ} %v work
\mycutecommand{\vq}{v^\bullet} % v query
\mycutecommand{\Sw}{S^\circ} % S work

\mycutecommand{\cLw}{\cL^\circ}
\mycutecommand{\cLq}{\cL^\bullet}
\mycutecommand\cLt{\cL^\uparrow}

\subsection{Input Oracle and the Canonical Form}
\label{sec:introCanonical}
Our previous discussion in \rf{sec:introTransducers} did not consider the input oracle.
In this paper, we assume an approach similar to that of quantum query algorithms, where oracle executions and the remaining operations are separated.
Moreover, unlike the algorithms, it suffices to have one query for a transducer.

We define the \emph{canonical} form of a transducer $S=S(O)$ with the input oracle $O$ in \rf{fig:canonical}.
The private space $\cL = \cLw\oplus\cLq$ is decomposed into the work part $\cLw$ and the query part $\cLq$, with the imposed decomposition $v = \vw\oplus\vq$ of the catalyst $v$.
The query is the very first operation, and it acts only on $\vq$.
It is followed by a unitary $\Sw$ without queries.

\myfigure{\label{fig:canonical}}
{A schematic depiction of a transducer in the canonical form.
It consists of one application of the oracle $O$ and an input-independent unitary $\Sw$.
The catalyst $v\in\cL$ is separated into two parts $v=\vw\oplus \vq$ with $\vw\in\cLw$ and $\vq\in\cLq$. 
The first one is not processed by the oracle, and the second one is.
Note that the input oracle is not applied to the public space.
}
{
\negbigskip
\[
\begin{tikzpicture}[every path/.append style={thick,->}]
    \draw (2,0) rectangle (3, 3) node[pos=0.5] {$\Sw$};
    \draw (-0.2,0) rectangle (1.2, 1) node[pos=0.5] {$I\otimes O$};
    \draw[\nonquerycolor] (-1,0.5) node[above]{$\vq$} to (-0.2,0.5);
    \draw[\witnesscolor] (-1,1.5) node[above]{$\vw$} to (2,1.5);
    \draw[\xicolor] (-1,2.5) node[above]{$\xi$} to (2,2.5);
    \draw[\nonquerycolor] (3,0.5)  to (4,0.5) node[above]{$\vq$};
    \draw[\witnesscolor] (3,1.5) to (4,1.5) node[above]{$\vw$} ;
    \draw[\taucolor] (3,2.5) to (4,2.5) node[above]{$\tau$} ;
    \draw[\querycolor] (1.2,0.5) to (2,0.5);
\end{tikzpicture}
\]
\negmedskip
}

Canonical transducers are easier to deal with, and every transducer can be converted into the canonical form (see \rf{prp:canoning}).
We will generally assume our transducers are canonical.
We write $W(S, O, \xi)$ instead of $W\sA[S(O),\xi]$ for the transduction complexity $ W(S,O,\xi) = \|v\|^2 = \|\vw\|^2 + \|\vq\|^2$.
Also, the \emph{total query state} is defined by $q(S, O, \xi) = \vq$, and the \emph{query complexity} by $L(S, O, \xi) = \|\vq\|^2$.

This definition is compatible with the case when $O$ is combined of several input oracles like in~\rf{eqn:introMultipleOracles}.
In this case, we define the \emph{partial query state} $q^{(i)}(S, O, \xi)$ as the state processed by the oracle $O^{(i)}$, and $L^{(i)}(S, O, \xi) = \norm|q^{(i)}(S, O, \xi)|^2$.

The following result, which justifies the name ``query complexity'', is proven in \rf{sec:reducingOracle}:

\begin{thm}[Query-Optimal Implementation of Transducer]
\label{thm:introImplementationBetter}
Let spaces $\cH, \cL=\cLw\oplus \cLq$ be fixed.
Moreover, assume the transducer uses $r=\OO(1)$ input oracles combined as in~\rf{eqn:introMultipleOracles}.
Let $\eps, W, L^{(1)},\dots,L^{(r)}>0$ be parameters.
Then, there exists an algorithm that conditionally executes $\Sw$ as a black box $K=\OO(1+W/\eps^2)$ times, 
makes $\OO(L^{(i)}/\eps^2)$ queries to the $i$-th input oracle $O^{(i)}$, and uses $\OO(K)$ other elementary operations.
The algorithm $\eps$-approximately transforms $\xi$ into $\tau(S, O, \xi)$ for all $S$, $O^{(i)}$, and $\xi$ such that $W(S, O, \xi)\le W$ and $L^{(i)} (S, O, \xi) \le L^{(i)}$ for all $i$.
\end{thm}

\pfstart[Proof Sketch]
Let us first consider the special case of $r=1$ for simplicity.
The crucial new idea compared to \rf{thm:introImplementation} is that we guess not one but some $D$ copies of the state $\vq$.
They are all processed by one oracle call, and then gradually given to $\Sw$, see \rf{fig:implementationBetter}.
Thus, the transduction complexity becomes $\|\vw\|^2 + D \|\vq\|^2$, but we have to execute the input oracle only once in every $D$ iterations.
The correct choice of $D$ is around $\|\vw\|^2/ \|\vq\|^2$, so that the norms of the query and the non-query parts of the catalyst become equalised.
If there are $r=\OO(1)$ input oracles, we perform the same procedure for all of them.
The total transduction complexity grows by a factor of $r$, which is tolerable by our assumption of $r=\OO(1)$.
\pfend

\myfigure{\label{fig:implementationBetter}}
{
A graphical illustration of the construction of \rf{thm:introImplementationBetter} with the same parameters as in \rf{fig:pumping}, and $D=2$.
We write $O$ instead of $I \otimes O$ to save space.
\\
Note that one oracle execution and $D$ subsequent executions of $\Sw$ form a transducer of its own.
}
{
\negbigskip
\def\queryheight{3}
\def\witnessheight{4}
\def\inputheight{4.5}
\def\xiheight{6}
\tikzset{witness/.style={above, node contents={\scriptsize $\frac{\vw}{\sqrt{K}}$}}}
\newcommand{\OneIteration}[1]{
    \edef\indxx{#1}
    \begin{scope}[shift={(3*\indxx,0)}]
        \draw (0.5,2.5) rectangle (1.5, 5) node[pos=0.5] {\Large $\Sw$};
        \draw[-] (0.2,\inputheight+0.4) node[\xicolor] {$\frac{\xi}{\sqrt{K}}$} ;
        \draw[-] (1.8,\inputheight+0.4) node[\taucolor] {$\frac{\tau}{\sqrt{K}}$} ;
        \draw[-] (0.2,\queryheight+0.4) node[\querycolor] {$\frac{O\vq}{\sqrt{K}}$} ;
        \draw[-] (1.8,\queryheight+0.4) node[\nonquerycolor] {$\frac{\vq}{\sqrt{K}}$} ;
        \draw[\witnesscolor] (1.5,\witnessheight) to node[witness]{} (3.5,\witnessheight);
    \end{scope}
    \draw[\xicolor] (-1, \xiheight)  .. controls (3*\indxx, \xiheight) and (3*\indxx-1, \inputheight) .. (3*\indxx+0.5,\inputheight);
    \draw[<-,\taucolor] (12, \xiheight)  .. controls (3*\indxx+2, \xiheight) and (3*\indxx+3, \inputheight) .. (3*\indxx+1.5,\inputheight);
}
\newcommand{\Oracle}[1]{
    \edef\indxx{#1}
    \begin{scope}[shift={(3*\indxx,0)}]
        \draw (-0.9,1) rectangle (-0.1, 3.5) node[pos=0.5] {\Large $O$};
        \draw [\querycolor](-0.1, 3) to (0.5, 3); 
        \draw [\querycolor](-0.1, 1.5) .. controls (1,1.5) and (2.5, 3) .. (3.5, 3);
        \draw [\nonquerycolor](4.5, 3) to (5.1, 3); 
        \draw [<-, \nonquerycolor](5.1, 1.5) .. controls (4,1.5) and (2.5, 3) .. (1.5, 3);
    \end{scope}
}
\[
\begin{tikzpicture}[every node/.style={font=\scriptsize}, every path/.append style={thick,->}]
\draw[-, \xicolor] (-1.5,\xiheight) to node[above]{\normalsize$\xi$} (-1,\xiheight);
\draw[-, \taucolor] (12.5,\xiheight) to node[above]{\normalsize$\tau$} (12,\xiheight);
\draw[\witnesscolor] (-1.5,\witnessheight) to node[witness]{} (0.5,\witnessheight);
\draw[\nonquerycolor] (-1.5,\queryheight) to node[above]{\scriptsize $\frac{\vq}{\sqrt{K}}$} (-0.9,\queryheight);
\draw[thick, \nonquerycolor] (-1.5,1.5) to node[above]{\scriptsize $\frac{\vq}{\sqrt{K}}$} (-0.9,1.5);
\draw[\nonquerycolor] (11, 3) to (12.5,3);
\draw[\nonquerycolor] (11, 1.5) to (12.5,1.5);
\OneIteration{0}
\OneIteration{1}
\OneIteration{2}
\OneIteration{3}
\Oracle{0}
\Oracle{2}
\end{tikzpicture}
\]
}

This theorem does not hold for superconstant values of $r$, in which case a more technical \rf{thm:optimalImplementation} should be used.
Note how canonicity of the transducer $S$ is used here.
Indeed, the input oracle is executed only once in each execution of the transducer, reducing the total number of oracle calls.

The definition of the canonical form is inspired by the implementation of the adversary bound for state conversion from~\cite{belovs:LasVegas}.
Moreover, as we show in \rf{sec:stateConversion}, the adversary bound is essentially equivalent to the above construction with $\cLw$ being empty.
Also, we show in \rf{sec:function} how to implement the usual adversary bound for function evaluation: \begin{thm}[Adversary Bound]
\label{thm:introAdv}
For every function $f\colon D\to [p]$ with $D\subseteq[q]^n$, there exists a canonical transducer $S_f$ with input oracle $O_x$ encoding the input string $x$, such that, for every $x\in D$, $S_f$ transduces $\ket |0>\transduce{} \ket |f(x)>$ on the input oracle $O_x$, and 
\[
W\sA[S_f, O_x, \ket|0>] =
L\sA[S_f, O_x, \ket|0>] \le \Adv(f),
\]
where $\Adv(f)$ is the adversary bound of $f$, defined in \rf{sec:function}.
\end{thm}

We can draw the following parallels with span programs.
It is known that the dual adversary bound for Boolean functions is equivalent to a very special case of span programs~\cite{reichardt:spanPrograms}.
General span programs provide more flexibility, and thus are more suitable for time-efficient implementations~\cite{belovs:learningClaws, reichardt:formulae,jeffery:spanFormula,cornelissen:spanProgramsTime}. 
While it is possible to implement the dual adversary bound time-efficiently~\cite{belovs:gappedGroupTesting}, the constructions are more complicated.

Span programs sometimes model more general function evaluation~\cite{beigi:span-programs}; but the dual adversary has been extended much further to include arbitrary state conversion~\cite{lee:stateConversion} with general unitary input oracles~\cite{belovs:variations}.
Transducers treat state conversion with unitary input oracles very naturally.
Adding the non-query space $\cLw$ provides more flexibility compared to the dual adversary, which again is beneficial for time-efficient implementations.
It is also of great help that transducers are quantum algorithms themselves, which makes time analysis especially straightforward.

Summarising, there are \emph{three} basic complexity measures associated with a transducer:
\begin{itemize}
\item Time complexity $T(S)$, which is independent of the input oracle $O$ and the initial state.  Similarly to usual quantum algorithms, the precise value of $T(S)$ depends on the chosen model of quantum computation.
\item Transduction complexity $W(S,O,\xi)$.  It is defined mathematically, and does not depend on the model.
On the other hand, it depends on both the input oracle and the initial state.
\item Query complexity $L(S,O,\xi)$.  
It is also defined mathematically, and does not depend on the model.
The query state $q(S,O,\xi)$ provides more information.
\end{itemize}

Let us finish this section with a few technicalities.
Similarly to~\rf{eqn:queryStateExtended} and~\rf{eqn:LasVegasExtended}, we extend the above definitions to $\xi' = \xi_1\oplus\cdots\oplus \xi_m\in\cE\otimes \cH$ via
\begin{equation}
\label{eqn:transductionExtended}
q(S,O,\xi') = \bigoplus_t q(S, O,\xi_t),\quad
L(S,O,\xi') = \sum_t L(S, O,\xi_t),
\quad \text{and}\quad
W(S, O, \xi') = \sum_t W(S, O,\xi_t).
\end{equation}
Again, these can be interpreted as the complexities of the transducer $I_\cE\otimes S$, see \rf{cor:byIdentity}.

If for $\xi\in\cH$ we use a catalyst $v$, for $c\xi$ with $c\in\bC$, we can use the catalyst $cv$.
This yields
\begin{equation}
\label{eqn:introLinearity}
W(S, O, c\xi) = |c|^2 W(S, O, \xi),
\quad
q^{(i)} (S, O, c\xi) = c\, q^{(i)} (S, O, \xi),
\quad
L^{(i)} (S, O, c\xi) = |c|^2 L^{(i)} (S, O, \xi).
\end{equation} 

\mycutecommand{\Hgood}{\cH_{\mathrm{good}}}
\mycutecommand{\Wmax}{W_{\mathrm{max}}}
\mycutecommand{\Lmax}{L_{\mathrm{max}}}
It often makes sense to define the subspace $\Hgood\subseteq \cH$ of admissible initial vectors to the transducer $S(O)$,
and define $\Wmax(S, O)$ as the supremum of $W(S, O, \xi)$ as $\xi$ ranges over \emph{unit} vectors in $\Hgood$.
We define $\Lmax$ and $\Lmax^{(i)}$ similarly.
We say that the initial state $\xi'\in \cE\otimes \cH$ is admissible if it lies in $\cE\otimes \Hgood$.
For any such state, by~\rf{eqn:transductionExtended} and~\rf{eqn:introLinearity}, we have
\begin{equation}
\label{eqn:admissibleQuery}
W(S, O, \xi') \le \Wmax(S, O) \|\xi'\|^2
\qqand
L(S, O, \xi') \le \Lmax(S, O) \|\xi'\|^2.
\end{equation}

\subsection{Transducers from Quantum Algorithms}
\label{sec:introAlgorithm->Transducer}
In this section, we briefly explain how we achieve one of the points on our agenda: conversion of arbitrary quantum algorithms into transducers.
We consider both the QRAG model mentioned at the end of \rf{sec:introPrior}, and the usual circuit model.
While the stronger QRAG model allows for better and more intuitive exposition, we are still able to get some useful results in the circuit model.

Let
\begin{equation}
\label{eqn:introAlgorithm}
A(O) = G_TG_{T-1}\cdots G_2 G_1
\end{equation}
be a quantum algorithm, where $G_i$ are individual elementary operations (gates), which also include queries to the input oracle.
We call the mapping $i\mapsto G_i$ the \emph{description} of $A$.
The number of elementary operations $T = T(A)$ is the time complexity of the algorithm.

\paragraph{Trivial Transducer}
On the one extreme of the trade-off~\rf{eqn:tradeoff} is the trivial transducer $S=A$.
In this case, the catalyst $v=0$, hence, $W(S,O,\xi) = 0$, and we get that~\rf{eqn:tradeoff} equals $T(A)$, as expected.
Of course, this does not require any change of the model.

\paragraph{QRAG Transducer}
The other extreme of the trade-off~\rf{eqn:tradeoff} is covered by the following construction, which assumes the QRAG model.
Write the sequence of states the algorithm~\rf{eqn:introAlgorithm} goes through on the initial state $\xi$:
\begin{equation}
\label{eqn:introSequenceOfStates}
\xi = \psi_0 \maps{G_1} \psi_1 \maps{G_2} \psi_2 \maps{G_3}\cdots \maps {G_T} \psi_{T} = \tau.
\end{equation}
We utilize the following history state:
\begin{equation}
\label{eqn:introHistoryState}
v = \sum_{t=1}^{T-1} \ket |t> \ket |\psi_t>.
\end{equation}
And define the transducer $S_A$ as follows:
\begin{equation}
\label{eqn:introProgram->Transducer}
\xi \oplus v = \sum_{t=0}^{T-1} \ket |t> \ket |\psi_t> \maps{S_A} \sum_{t=1}^{T} \ket |t> \ket |\psi_t> = \tau \oplus v,
\end{equation}
where, in $S_A$, we first apply $G_t$ conditioned on the first register having value $t$, and then increment $t$ by one.
The first operation is possible assuming the QRAG model and QRAM access to the description of $A$, see \rf{sec:QRAG}.
The transduction action of $S_A(O)$ is exactly $A(O)$, and $W(S_A,O,\xi) = (T-1)\|\xi\|^2$, which we will simplify to $T\|\xi\|^2$ for brevity.

\mycutecommand{\timeR}{T_{\mathrm R}}

\paragraph{Comments on Time Complexity}
But what is the time complexity $T(S_A)$?
In $S_A$, we perform two operations: increment a word-sided register and execute one instruction from the program specified by the address $t$.
If this were a modern randomised computer, both operations would be elementary and took $\OO(1)$ time.
From the theoretical side, this is captured by the notion of RAM machine.
If we consider the scale of individual qubits instead of word-sided registers, the increment operation takes time $\OO(\log T)$, and the second operation at least as much.
In order to simplify the following discussion, let us assume that both operations take some time we denote $\timeR$.  That is
\begin{equation}
\label{eqn:timeR}
\text{$\timeR$: time required to perform basic word operations, including random access.}
\end{equation}
For simplicity, we do not discriminate between different word operations.
We may think of $\timeR$ as being 1, or $\OO(\log T)$, but either way, it is some fixed factor which denotes transition from the circuit model, where $A$ operates,%
\footnote{It is the usual assumption that $A$ is a bona fide quantum circuit, but we may also assume that $A$ uses QRAG gates.}
to the QRAG model, where $S_A$ is implemented.

\paragraph{Canonical Form}
Neither the first, nor the second transducer above are in the canonical form.
The latter is given by the following result, which we prove in \rf{sec:programs->transducers}.

\begin{thm}
\label{thm:introProg->Transducer}
For a quantum program $A$ on the input oracle $O$, there exists a canonical transducer $S_A = S_A(O)$, whose transduction action is identical to $A$, and whose complexity is given by \rf{fig:complexityTable}.
In the QRAG model, we assume QRAM access to the description of the program $A$.
\end{thm}

\myfigure{\label{fig:complexityTable}}
{
Complexity of canonical transducers derived from a quantum algorithm in terms of complexities of the algorithm.  Both the circuit and the QRAG model versions are considered.
}
{
    \[
    \begin{tabular}{rl|cc|}
    &&Circuit Model&QRAG model\\\hline
    Time& $T(S_A)$& $\OO(T(A))$ & $\OO(\timeR)$ \\
    Transduction &$W(S_A,O,\xi)$ & $L(A,O,\xi)$ & $T(A)\|\xi\|^2$\\
    Query state & $q(S_A,O,\xi)$ & $q(A, O,\xi)$ & $q(A, O,\xi)$\\
    \hline
    \end{tabular}
    \]
}

\pfsketch
The constructions in both circuit and the QRAG models are already sketched above.
We describe here how they can both be transformed into the canonical form so that the query state $\vq$ from~\rf{fig:canonical} becomes equal to the total query state $q(A, O, \xi)$ of the algorithm.

For the trivial transducer, the catalyst is the total query state, which is all processed by one query to the oracle.
Whenever the algorithm was making a query, the transducer switches the query state with the corresponding part of the catalyst, thus simulating a query.

For the QRAG transducer, the catalyst~\rf{eqn:introHistoryState} already contains of all the intermediate states of the program.
The transducer can then apply the oracle to all of them in one go, and then proceed as before without making a single additional query.
\pfend

\paragraph{Query Compression}
One consequence of the above results is that we can compress the number of queries of a quantum algorithm $A$ to match its worst-case Las Vegas query complexity.
The following result is an immediate corollary of Theorems~\ref{thm:introProg->Transducer} and~\ref{thm:introImplementationBetter}.

\begin{thm}[Query Compression]
\label{thm:introQueryCompression}
Assume $A=A(O)$ is a quantum algorithm with $r=\OO(1)$ input oracles as in~\rf{eqn:introMultipleOracles}.
Let $\eps, L^{(1)},\dots, L^{(r)}>0$ be parameters.
There exists a quantum algorithm $A'=A'(O)$ with the following properties:
\begin{itemize}\itemsep=0pt
\item It makes $\OO(L^{(i)}/\eps^2)$ queries to the $i$-th input oracle $O^{(i)}$.
\item For every normalised initial state $\xi$ and every input oracle $O=O^{(1)}\oplus O^{(2)}\oplus \cdots \oplus O^{(r)}$ as in~\rf{eqn:introMultipleOracles}, we have 
$\norm|A(O)\xi - A'(O)\xi| \le \eps$ as long as $L^{(i)}(A, O,\xi)\le L^{(i)}$ for all $i$.
\item In the QRAG model and assuming QRAM access to the description of $A$, its time complexity is $\OO(\timeR\cdot T(A)/\eps^2)$.
\item In the circuit model, its time complexity is $\OO(L\cdot T(A)/\eps^2)$, where 
$L = 1 + L^{(1)}+\cdots+ L^{(r)}$.
\end{itemize}
\end{thm}

This improves over the analogous result from~\cite{belovs:LasVegas} in two respects.
First, it essentially preserves the time complexity of the algorithm in the QRAG model, while~\cite{belovs:LasVegas} did not consider time complexity at all.
Second, it allows multiple input oracles, as long as its number is bounded by a constant, while~\cite{belovs:LasVegas} only allowed for a single input oracle.

\pfstart[Proof of \rf{thm:introQueryCompression}]
First, obtain the transducer $S_A$ as in \rf{thm:introProg->Transducer}.
Then, apply \rf{thm:introImplementationBetter} to $S_A$.

One key observation is that $q(S_A, O,\xi) = q(A, O,\xi)$ for all $O$ and $\xi$.
Hence, $S_A$ and $A$ have the same partial Las Vegas query complexities on all $O$ and $\xi$.
This yields the statement on the number of queries bounded by $L^{(i)}/\eps^2$.

In the QRAG model, $W(S_A, O, \xi) = T(A)$ and $T(S_A) = \OO(\timeR)$.
This gives the required runtime, as we can use that $T(A)\ge 1$ to remove the additive 1 factor.

In the circuit model, we use that $T(S_A) = \OO\sA[T(A)]$ and
\[
W(S_A, O, \xi) = L(A, O, \xi) = \sum_i L^{(i)} (A, O, \xi) \le \sum_i L^{(i)}.
\]
We take care of the additive 1 factor by adding it explicitly to $L$.
This covers the extreme case of $L$ being too small, in particular, 0.
\pfend

\subsection{Composition of Transducers}
\label{sec:introComposition}
Transducers can be composed just like usual quantum algorithms: we consider parallel, sequential, and functional composition of transducers.
Thus, from the design point of view, there is little difference between dealing with quantum algorithms and transducers.
The advantage is that in all these composition modes, the resources are more tightly accounted for than is the case for traditional quantum algorithms.

We will first sketch the composition results for transducers.
We will completely omit the case of sequential composition from this section.
It suffices to say, that it is very similar to the transformation of programs in \rf{sec:introAlgorithm->Transducer}.
After that, we will give few applications both in the circuit and the QRAG models.
Unlike other subsections of this section, we will be able to give complete proofs for most of the results.

\paragraph{Composition of Transducers}
The parallel composition of transducers is straightforward.  They are just implemented in parallel as usual quantum algorithms.
We have the following relations, akin to~\rf{eqn:randomParallelComplexity} and~\rf{eqn:quantumParallel}:
\begin{equation}
\label{eqn:introParallel}
W\sB[\bigoplus_i S_i, O, \bigoplus_i \xi_i] = \sum_i W(S_i, O, \xi_i)
\qqand
q\sB[\bigoplus_i S_i, O, \bigoplus_i \xi_i] = \bigoplus_i q(S_i, O, \xi_i).
\end{equation}
For the time complexity of implementing $\bigoplus_i \Sw_i$, we can say precisely as much as for usual quantum algorithms.
In some cases it is easy: when all $\Sw_i$ are equal, for example.
Also, it can be efficiently implemented assuming the QRAG model, see \rf{cor:selectProgram}.
In general, however, direct sum in the circuit model can take as much time as the total complexity of all $\Sw_i$ together.

The functional composition of transducers parallels \rf{fig:composition}, where we replace the programs $A$ and $B$ with transducers $S_A$ and $S_B$, respectively.
The functional composition of the two is a transducer $S_A\circ S_B$ whose transduction action on the oracle $O$ is equal to the transduction action of $S_A$ on the oracle $O\oplus O'$, where $O'$ is the transduction action of $S_B(O)$.
The following result, proven in~\rf{sec:functional}, parallels~\rf{eqn:quantumComposition}.

\begin{prp}[Functional Composition of Transducers]
\label{prp:introFunctional}
The functional composition $S_A\circ S_B$ can be implemented in the following complexity, where we use the extended notion of complexity from~\rf{eqn:transductionExtended}.
\begin{itemize}\itemsep=0pt
\item Its transduction complexity satisfies
\begin{equation}
\label{eqn:introFunctionalTransduction}
W(S_A\circ S_B, O, \xi) = W(S_A, O\oplus O', \xi) + W\sB[ S_B, O, q^{(1)}(S_A, O\oplus O', \xi)].
\end{equation}
\item Its total query state is
\begin{equation}
\label{eqn:introFunctionalQuery}
q(S_A\circ S_B, O, \xi) = q^{(0)}(S_A, O\oplus O', \xi) \oplus q\sB[ S_B, O, q^{(1)}(S_A, O\oplus O', \xi)].
\end{equation}
\item Its time complexity is the sum of the (conditional) time complexities of $S_A$ and $S_B$.
\end{itemize}
Here $q^{(0)}$ and $q^{(1)}$ denote the partial query states of $S_A$ to the oracles $O$ and $O'$, respectively.
\end{prp}

One can see that~\rf{eqn:introFunctionalTransduction} and~\rf{eqn:introFunctionalQuery} meet the form of the ``gold standard'' of~\rf{eqn:randomComposition2}.
The second one strongly resembles~\rf{eqn:quantumComposition}.
One unfortunate deviation from this is the time complexity, which afterwards gets multiplied by the transduction complexity in~\rf{eqn:tradeoff}.
But since the dependence is additive, we can generally tolerate it, see, e.g., \rf{thm:introIterated} below.

Let us state, for the ease of future referencing, a number of simple consequences of \rf{prp:introFunctional} similar to~\rf{eqn:quantumCompositionMultiple}--\rf{eqn:quantumCompositionMultiple2}.
First, if all the queries to $S_B$ are admissible, we have by~\rf{eqn:introFunctionalTransduction} and~\rf{eqn:admissibleQuery}:
\begin{equation}
\label{eqn:compositionTransductionUpper}
W(S_A\circ S_B, O, \xi) \le W(S_A, O\oplus O', \xi) + \Wmax(S_B, O) L^{(1)}\sA[S_A, O\oplus O', \xi].
\end{equation}
Also, from~\rf{eqn:introFunctionalQuery}:
\begin{align}
\label{eqn:compositionLasVegas}
L(S_A\circ S_B, O, \xi) 
&= L^{(0)}(S_A, O\oplus O', \xi) + L\sB[ S_B, O, q^{(1)}(S_A, O\oplus O', \xi)]\\
\label{eqn:compositionLasVegasUpper}
&\le L^{(0)}(S_A, O\oplus O', \xi) + \Lmax(S_B, O) L^{(1)}\sA[S_A, O\oplus O', \xi].
\end{align}
In the case of multiple input oracles, similar to \rf{fig:compositionMultiple}, with $O' = \bigoplus_i O^{(i)}$ and $S_B = \bigoplus S_{B_ i}$, so that $S_{B_i}$ implements $O^{(i)}$, we have by~\rf{eqn:introParallel}:
\begin{align}
\label{eqn:compositionTransductionMultiple}
W(S_A\circ S_B, O, \xi) 
&= W(S_A, O\oplus O', \xi) + \sum_i W\sB[ S_{B_i}, O, q^{(i)}(S_A, O\oplus O', \xi)]\\
\label{eqn:compositionTransductionMultipleUpper}
&\le W(S_A, O\oplus O', \xi) + \sum_i \Wmax( S_{B_i}, O) L^{(i)}\sA[S_A, O\oplus O', \xi],\\
\label{eqn:compositionQueryMultiple}
q(S_A\circ S_B, O, \xi) 
&= q^{(0)}(S_A, O\oplus O', \xi) \oplus \bigoplus_i q\sB[ S_{B_i}, O, q^{(i)}(S_A, O\oplus O', \xi)]
\end{align}
and
\begin{align}
\label{eqn:compositionLasVegasMultiple}
L(S_A\circ S_B, O, \xi) 
&= L^{(0)}(S_A, O\oplus O', \xi) + \sum_i L\sB[ S_{B_i}, O, q^{(i)}(S_A, O\oplus O', \xi)]\\
\label{eqn:compositionLasVegasMultipleUpper}
&\le L^{(0)}(S_A, O\oplus O', \xi) + \sum_i \Lmax( S_{B_i}, O) L^{(i)}\sA[S_A, O\oplus O', \xi].
\end{align}
In~\rf{eqn:compositionTransductionMultipleUpper} and~\rf{eqn:compositionLasVegasMultipleUpper}, we used~\rf{eqn:admissibleQuery} in assumption that all the queries to $S_B$ and $S_{B_i}$ are admissible.

\paragraph{Composition of Programs}
Let us give a few examples of composition of quantum programs using transducers.
We focus on the general state conversion here, evaluation of function postponed till \rf{sec:introPurifier}.
We consider both the circuit and the QRAG models.

\begin{thm}[Composition of Programs, circuit model]
\label{thm:introCompositionCircuit}
Assume the settings of \rf{fig:compositionMultiple}, where $r=\OO(1)$ and all $A$ and $B_i$ are in the circuit model.
Define $B = \bigoplus_i B_i$, and let $\eps, L^{(1)},\dots, L^{(r)}>0$ be parameters.

Then, there exists a quantum algorithm $A' = A'(O)$ such that $\normA|A'(O)\xi - A\sA[O\oplus B(O)]\xi|\le \eps$ for every normalised $\xi$ as long as $L^{(i)}(A, O\oplus B(O), \xi)\le L^{(i)}$ for all $i\ge 1$.
The program $A'$ can be implemented in the circuit model in time
\begin{equation}
\label{eqn:introCompositionCircuit}
\OO\s[\frac{L\cdot T(A) + \sum_i T\s[B_i] L^{(i)}}{\eps^2}],
\end{equation}
where $L = 1+L^{(1)}+\cdots+L^{(r)}$.
\end{thm}

\pfstart
Use the circuit version of \rf{thm:introQueryCompression} for the program $A$, where we treat calls to $O$ as ordinary operations (in other words, we assume $A$ has $r$ input oracles $O^{(1)},\dots, O^{(r)}$).
After that, replace each call to $O^{(i)}$ by the execution of $B_i$.
\pfend

It turns out that it is not efficient to use \rf{prp:introFunctional} here as this would increase time complexity of the resulting transducer.

\begin{rem}
In \rf{thm:introCompositionCircuit}, the emphasis is on the time complexity.
If we want to simultaneously bound query complexity, we treat $O$ as the input oracle in $A$ as well.
This gives time complexity~\rf{eqn:introCompositionCircuit} with $L = 1+ L^{(0)}+L^{(1)}+\cdots+L^{(r)}$ and the total number of queries
\[
\OO\s[\frac{L^{(0)} + \sum_i Q\s[B_i] L^{(i)}}{\eps^2}],
\]
where $Q(B_i)$ is the number of queries made by $B_i$. 
We additionally require that $L^{(0)}\sA[A, O\oplus B(O), \xi]\le L^{(0)}$.
\end{rem}

\begin{thm}[Composition of Programs, QRAG model]
\label{thm:introCompositionQRAG}
Assume the settings of \rf{fig:compositionMultiple}.
Let $\eps, T>0$ be parameters.
Assuming the QRAG model and QRAM access to an array with description of $A$ and all $B_i$, there exists a quantum algorithm $A' = A'(O)$ with time complexity $\OO(\timeR\cdot T/\eps^2)$ such that, for every normalised $\xi$, we have $\|A'(O)\xi - A\sA[O\oplus B(O)]\xi\| \le \eps$ as long as
\begin{equation}
\label{eqn:introCompositionQRAG}
T(A) + \sum_{i=1}^r T(B_i) L^{(i)}(A, O\oplus B(O), \xi) \le T.
\end{equation}
The algorithm makes $\OO(L/\eps^2)$ queries to the input oracle $O$, where $L$ is an upper bound on $L(A\circ B, O, \xi)$, given by~\rf{eqn:quantumCompositionMultiple1}.
\end{thm}

\pfstart
Convert $A$ and all the $B_i$ into transducers $S_A$ and $S_{B_i}$ as in \rf{thm:introProg->Transducer}.
We obtain the transducer $S_B = \bigoplus_i S_{B_i}$ for $B$.
Then compose $S_A\circ S_B$ using \rf{prp:introFunctional}.
By definition, its transduction action is identical to $A(B)$.

Since the transduction complexity of $S_A$ on a normalised initial state is bounded by $T(A)$ and that of $S_{B_i}$ by $T(B_i)$, we get from~\rf{eqn:compositionTransductionMultipleUpper} that for a normalised $\xi$:
\[
W(S_A\circ S_B, O, \xi) \le T(A) + \sum_i T(B_i) L^{(i)}\sA[A, O\oplus B(O), \xi].
\]

The main reason this construction is efficient in the QRAG model is that the time complexity of the transducers stays bounded by $\OO(\timeR)$ the whole time.
Indeed, such is the time complexity of the individual transducers obtained from $A$ and $B^{(i)}$.
Parallel composition can be performed efficiently in the QRAG model (\rf{prp:program->transducerParallel}), and the time complexity of the functional composition is the sum of its constituents.

The transducers $S_A$ and $S_{B_i}$ have the same query states as $A$ and $B_i$, respectively.
By~\rf{eqn:compositionQueryMultiple}, we get that the total query state of $S_A\circ S_B$ is identical to that of $A\circ B$, which is given by~\rf{eqn:quantumCompositionMultiple}.
The statement of the theorem follows from \rf{thm:introImplementationBetter}.
\pfend

Observe that \rf{thm:introCompositionQRAG}, while assuming a stronger model, gives a stronger result than \rf{thm:introCompositionCircuit}.
The differences are as follows.
First, we do not have to assume that $r=\OO(1)$.
Second, the $L^{(i)}(A, O\oplus B(O), \xi)$ in~\rf{eqn:introCompositionQRAG} are the actual values of the Las Vegas query complexity, while $L^{(i)}$ in~\rf{eqn:introCompositionCircuit} are only upper bounds on them.
This can be important if $L^{(i)}(A, O\oplus B(O), \xi)$ heavily fluctuates over different input oracles $O$.
Finally, the $T(A)$ term is oddly multiplied by $L$ in~\rf{eqn:introCompositionCircuit}.

Altogether, the expression in~\rf{eqn:introCompositionQRAG} is more natural.
It is also similar to the one in Section~1.2 of~\cite{jeffery:subroutineComposition}.
Our result is more general though, as it covers arbitrary state conversion, and not only function evaluation (we can assume $\eps=\Omega(1)$ in \rf{thm:introCompositionQRAG} as it gives bounded-error evaluation of a function).

On the other hand, if the estimates $L^{(i)}$ are sufficiently precise, and $T(A)$ is smaller than average $T(B^{(i)})$, the expression in~\rf{eqn:introCompositionCircuit} is quite close to~\rf{eqn:introCompositionQRAG}.

\paragraph{Multiple Layers of Composition}
Here we assume the QRAG model of computation.
Theorems~\ref{thm:introCompositionCircuit} and~\ref{thm:introCompositionQRAG} considered one layer of composition.
In the case of multiple layers, similarly as for the span programs, it is advantageous to perform all the compositions in the realm of transducers, and to transform the resulting transducer back into an actual algorithm only at the very end.

\myfigure{\label{fig:multipleLayers}}
{
One node $B_{t; i_1, i_2, \dots, i_t}$ in the composition tree.
Its initial state (in the general sense of~\rf{eqn:transductionExtended}) is given by $\xi_{t; i_1, i_2, \dots, i_t}$.
It has several subroutines of the form $B_{t+1; i_1; i_2,\dots,i_t, i_{t+1}}$ with the corresponding query state $\xi_{t+1; i_1; i_2,\dots,i_t, i_{t+1}}$.
}
{
\negbigskip
\negbigskip
\negbigskip
\negbigskip
\[
\begin{tikzpicture}[every path/.append style={thick,->}]
    \draw[\xicolor] (-2, -0.5) node[left]{$\xi_{t; i_1, i_2,\dots, i_t}$} to (-1,-0.5);
    \draw (-1,0) rectangle (1,-1) node[pos=0.5] {$B_{t; i_1; i_2,\dots,i_t}$};
    \draw (0,-2) rectangle (2.5,-3) node[pos=0.5] {$B_{t+1; i_1; i_2,\dots,i_t, 1}$};
    \draw (3,-2) rectangle (5.5,-3) node[pos=0.5] {$B_{t+1; i_1; i_2,\dots,i_t, 2}$};
    \node at (6.3, -2.5) {$\cdots$};
    \node at (10.5, -2.5) {$\cdots$};
    \draw (7,-2) rectangle (10,-3) node[pos=0.5] {$B_{t+1; i_1; i_2,\dots,i_t, i_{t+1}}$};
    \draw[purple, out=0, in=90] (1,-0.8) to node[left, pos=0.7] {$\xi_{t+1; i_1, i_2,\dots, i_t,1}$} (1.5,-2);
    \draw[purple, out=0, in=90] (1,-0.6) to node[left, pos=0.85] {$\xi_{t+1; i_1, i_2,\dots, i_t,2}$} (4.5,-2);
    \draw[purple, thick, -] (1,-0.4) to (1.2,-0.4);
    \draw[purple, thick, -] (1,-0.2) to (1.2,-0.2);
    \draw[purple, out=0, in=90] (1,-0.3) to node[left, pos=0.95] {$\xi_{t+1; i_1, i_2,\dots, i_t,i_{t+1}}$} (9,-2);
\end{tikzpicture}
\]
\negmedskip
}

Consider a composition tree of quantum subroutines.
The top layer is the algorithm $B_{0}$ that has several subroutines of the form $B_{1,i}$, like $B_i$ used to be for $A$ in \rf{fig:compositionMultiple}.
We define the tree downwards so that, in general, a subroutine $B_{t; i_1,i_2,\dots,i_t}$ has several subroutines of the form $B_{t+1; i_1, i_2, \dots, i_t, i_{t+1}}$, see \rf{fig:multipleLayers}.
Let $d$ be the maximal value of $t$ in $B_{t; i_1,i_2,\dots,i_t}$.
It is the depth of the composition tree.
We assume all the subroutines have access to some common oracle $O$.
Define the composition $B=B(O)$ of the whole tree in the obvious inductive way, so that the initial state $\xi$ of the composed algorithm is the initial state $\xi$ of $B_0$.

It is not hard to get the Las Vegas query complexity of $B$ using~\rf{eqn:quantumCompositionMultiple} inductively or from the general principles.
Let $q_{t; i_1,i_2,\dots,i_t}(O,\xi)$ be the query state given by $B_{t; i_1,i_2,\dots,i_t}$ to the input oracle $O$ when $B$ is executed on the initial state $\xi$.
Then,
\begin{equation}
\label{eqn:compositionTreeQueryComplexity}
q(B, O, \xi) = \bigoplus_{t; i_1,\dots, i_t} q_{t; i_1,i_2,\dots,i_t}(O,\xi).
\end{equation}
We obtain a similar result for the time complexity.
Let $\xi_{t; i_1,i_2,\dots,i_t}(O,\xi)$ be the total query state given \emph{to} the subroutine $B_{t; i_1,i_2,\dots,i_t}$.
In particular, $\xi_0(O,\xi) = \xi$.

\begin{thm}[Composition Tree]
\label{thm:introCompositionTree}
Assuming the QRAG model and QRAM access to the description of all $B_{t; i_1,\dots,i_t}$ as above, there exists a quantum algorithm $B' = B'(O)$ with time complexity $\OO(\timeR\cdot (d+1)\cdot T/\eps^2)$ such that, for every $\xi$, we have $\|B'(O)\xi - B(O)\xi\| \le \eps$ as long as
\begin{equation}
\label{eqn:introCompositionTree}
\sum_{t; i_1,\dots, i_t}  T(B_{t; i_1,i_2,\dots,i_t}) \|\xi_{t; i_1,i_2,\dots,i_t}(O,\xi)\|^2 \le T,
\end{equation}
where the summation is over all the vertices of the composition tree.
The algorithm makes $\OO(L/\eps^2)$ queries to $O$, where $L$ is an upper bound on the Las Vegas query complexity of $B$ as obtained from~\rf{eqn:compositionTreeQueryComplexity}.
\end{thm}

\pfstart
We use the induction on $d$ to show that, under the assumptions of the theorem, there exists a transducer $S_B$ whose transduction action is identical to $B$, whose transduction complexity is given by the left-hand side of~\rf{eqn:introCompositionTree}, and whose time complexity is $\OO\sA[(d+1)\cdot\timeR]$.

The base case is given by $d=0$, where we only have $B_0$, which we transform into a transducer using \rf{thm:introProg->Transducer}.
The transduction complexity of $S_{B_0}$ on a normalised initial state is $T(B_0)$. 
Hence on the initial state $\xi_0$, the transduction complexity is $T(B_0)\|\xi_0\|^2$ by~\rf{eqn:introLinearity}.

Assume the theorem is true for depth $d$.
For the depth $d+1$, we treat the nodes $B_{t; i_1,i_2,\dots,i_t}$ with $t\le d$ as forming a composition tree $A$ with depth $d$, and the nodes $B_{d+1; i_1,i_2,\dots,i_d, i_{d+1}}$ as input oracles to $A$.
In other words, $A$ has an input oracle $O\oplus O'$ with 
\[
O' = \bigoplus_{i_1,\dots,i_d, i_{d+1}} B_{d+1; i_1,\dots,i_d, i_{d+1}}(O).
\]

We use the induction assumption to obtain a transducer $S_A$ for the composition tree $A$, whose transduction complexity on $O \oplus O'$ and $\xi$ is given by the sum in~\rf{eqn:introCompositionTree}, where we restrict the sum to $t\le d$.
We convert each $B_{d+1; i_1,\dots,i_d, i_{d+1}}$ into a transducer and join them via direct sum to obtain a transducer $S_{B_{d+1}}$ whose transduction action on the oracle $O$ is identical to $O'$.
Then, we apply \rf{prp:introFunctional} to get the transducer $S_B = S_A\circ S_{B_{d+1}}$.
Its transduction complexity on $O$ and $\xi$ is given by the left-hand side of~\rf{eqn:introCompositionTree} with all the terms involved, where the term 
\[
T(B_{d+1; i_1,i_2,\dots,i_d, i_{d+1}}) \|\xi_{d+1; i_1,i_2,\dots,i_d, i_{d+1}}(O,\xi)\|^2
\]
is the contribution of $B_{d+1; i_1,i_2,\dots,i_d, i_{d+1}}$.

All the transducers in $S_{B_{d+1}}$ can be implemented in parallel due to the QRAG assumption (see \rf{prp:program->transducerParallel}), hence, its time complexity is $\OO(\timeR)$.
Thus, $T(S_B) = T(S_A) + T(S_{B_{d+1}}) = \OO\sA[(d+1)\timeR]$.
Finally, using~\rf{eqn:quantumCompositionMultiple} and~\rf{eqn:compositionQueryMultiple}, we get that $S_B$ and $B$ have the same query state.
The statement of the theorem again follows from \rf{thm:introImplementationBetter}.
\pfend

\paragraph{Iterated Functions}
Due to the discussion after \rf{thm:introCompositionCircuit}, one might think that several layers of composition are difficult for the circuit model.
However, this is not the case if the composition tree is sufficiently homogenous.
As an example, we consider evaluation of iterated functions in the circuit model.

Let $f\colon [q]^{n} \to [q]$ and $g\colon [q]^m\to [q]$ be total functions.
The \emph{composed function} $f\circ g\colon [q]^{nm} \to [q]$ is defined by
\begin{equation}
\label{eqn:introComposedFunction}
\begin{aligned}
\sS[f\circ g]&(z_{1,1}, \dots,z_{1,m},\;\; z_{2,1},\dots,z_{2,m},\;\;\dots\dots,\;\;z_{n,1},\dots,z_{n,m})\\
&= f\sA[
g(z_{1,1}, \dots,z_{1,m}), 
g(z_{2,1}, \dots,z_{2,m}),
\dots,
g(z_{n,1}, \dots,z_{n,m})],
\end{aligned}
\end{equation}
which is equivalent to~\rf{eqn:randomComposedFunction} with all the inner functions being equal.
The function composed with itself several times is called \emph{iterated function}.
We use the following notation
$f^{(1)} = f$ and $f^{(d+1)} = f^{(d)}\circ f = f\circ f^{(d)}$.

Iterated functions have been studied before both classically~\cite{snir:nand, saks:nand, santha:readOnceFormulae, magniez:maj3, leonardos:maj3} and quantumly~\cite{farhi:nandTree, ambainis:formulaeEvaluation, reichardt:unbalancedFormulas, reichardt:formulae}.
For the case of Boolean functions, an essentially optimal algorithm
was given by Reichardt and \v Spalek in~\cite{reichardt:unbalancedFormulas, reichardt:formulae} using span programs.%
\footnote{
Papers~\cite{reichardt:unbalancedFormulas, reichardt:formulae} are mostly known for their \emph{query} results, but they also contain statements on the time complexity of the resulting algorithms.
It is these time complexity statements that we extend in \rf{thm:introIterated}.
}
It is based on the use of span programs.
Similar results for the general case of non-Boolean functions can be easily obtained using composition of transducers.
While it is true that the time complexity grows with each layer of iteration, it only does so \emph{additively}, which is overshadowed by optimal \emph{multiplicative} handling of the query complexity.
We formally prove the following result in \rf{sec:iterated}.

\begin{thm}[Iterated Functions]
\label{thm:introIterated}
Let $f\colon [q]^n\to[q]$ be a total function.
There exists a bounded-error quantum algorithm that evaluates the iterated function $f^{(d)}$ in query complexity $\OO(\Adv(f)^d)$ and time complexity\footnote{We use $\OO_f(\cdot)$ to indicate that the suppressed constant may depend on the particular function $f$.} $\OO_f\sA[d\cdot\Adv(f)^d]$, where $\Adv(f)$ is the adversary bound of $f$.
The algorithm works in the circuit model.
\end{thm}

\pfsketch
We use induction to construct the transducer $S_{f^{(d)}}$ evaluating the function $f^{(d)}$ and having worst-case query complexity $\Adv(f)^d$, and transduction complexity $\OO\sA[\Adv(f)^d]$.

The base case is the transducer $S_f$ from \rf{thm:introAdv}.
For the inductive step, we use \rf{prp:introFunctional} with $S_A = S_{f^{(d)}}$ and $S_B = S_f$ to get $S_{f^{(d+1)}} = S_A\circ S_B$.
In notations of that theorem, $O$ encodes the input to $f^{(d+1)}$ and $O'$ the input to $f^{(d)}$ obtained by evaluating the lowest level of the composition tree.
First, $S_A$ does not make direct queries to $O$.
Second, $S_B = S_f$ has worst-case Las Vegas query complexity $\Adv(f)$ on a unit admissible vector. 
Hence, by~\rf{eqn:compositionLasVegasUpper} and the induction assumption:
\[
L\sA[S_A\circ S_B, O, \ket |0>] 
\le \Adv(f) L\sA[S_A, O', \ket|0>] \le \Adv(f)^{d+1}.
\]
Similarly, the worst-case transduction complexity of $S_B$ on a unit admissible vector is $\Adv(f)$, hence by~\rf{eqn:compositionTransductionUpper}:
\begin{align*}
W\sA[S_A\circ S_B, O, \ket |0>] 
&\le W\sA[S_A, O, \ket |0>] + \Adv(f) L(S_A, O', \ket |0>)\\
&\le \OO\sA[\Adv(f)^d] + \Adv(f)^{d+1} = \OO\sA[\Adv(f)^{d+1}]
\end{align*}
for a sufficiently large constant behind the big-Oh.

For the time complexity, we have by induction that $S_{f^{(d)}}$ has time complexity $d\cdot T(S_f)$.
The theorem follows from \rf{thm:introImplementationBetter}.
\pfend

Note that the transducer $S_{f^{(d)}}$ in the proof of \rf{thm:introIterated} is different from the transducer $S'_{f^{(d)}}$ we would get by applying \rf{thm:introAdv} to the adversary bound of $f^{(d)}$ obtained using the composition results for the adversary bound.
Indeed, for $S'_{f^{(d)}}$, its transduction and query complexities are equal, which is not the case for $S_{f^{(d)}}$.
Also, for $S'_{f^{(d)}}$, we have no guarantees on its running time.
In \rf{thm:introIterated}, we use the non-query part $\vw$ of the catalyst as a ``scaffolding'' to build a time-efficient iterative algorithm.

\subsection{Purifiers and Composition of Functions}
\label{sec:introPurifier}

Although we have studied thriftiness from \rf{sec:introPrior}, we have so far not touched much on exactness.
True, in most cases, like in \rf{sec:introWalks} on quantum walks, or \rf{thm:introAdv} on the adversary bound, the transduction action of the corresponding transducer is exact.
In this section, we will show how to get very close to general exactness starting from an arbitrary algorithm evaluating a function with bounded error.

We consider both Boolean and non-Boolean functions.
In the Boolean case, we abstract the action of the function-evaluating algorithm as an input oracle performing the following state generation:
\begin{equation}
\label{eqn:introPurifierInput}
O_\psi\colon \ket M |0> \mapsto \ket M|\psi> = \ket B|0>\ket N|\psi_0> + \ket B|1> \ket N|\psi_1>
\end{equation}
in some space $\cM = \cB\otimes \cN$ with $\cB = \bC^2$.
The action of the purifier only depends on the state $\psi$ in~\rf{eqn:introPurifierInput}, hence the notation $O_\psi$.
We allow the gap to be at any position $c$ between 0 and 1.
In other words, we assume there exist constants $0\le c-d < c + d \le 1$ such that
\begin{equation}
\label{eqn:introPurifierCases}
\text{either \qquad
$\|\psi_1\|^2 \le c-d$
\qquad  or \qquad 
$\|\psi_1\|^2 \ge c+d$.
}
\end{equation}
The first case is negative, the second one positive, or, $f(\psi)=0$ and $f(\psi)=1$, respectively.

In the non-Boolean case, the range is some $[p]$. 
For simplicity, we assume $p=\OO(1)$ here. An input oracle has the form
\begin{equation}
\label{eqn:introrPurifierInput2}
O_\psi\colon \ket M |0> \mapsto \ket M|\psi> 
= 
\sum_{j=0}^{p-1} \ket B|j>\ket N|\psi_j>,
\end{equation}
with $\cB = \bC^p$, and we assume there exists a (unique) $f(\psi)\in[p]$ such that
\begin{equation}
\label{eqn:introPurifierCases2}
\|\psi_{f(\psi)}\|^2 \ge \frac12+d
\end{equation}
for some constant $d>0$.

The traditional way of performing error reduction is via majority voting.
The following result is folklore.

\begin{thm}[Majority Voting]
\label{thm:majority}
For any $\eps>0$, there exists an algorithm with bidirectional access to an oracle like in~\rf{eqn:introrPurifierInput2} that has the following properties.
Assuming the oracle satisfies~\rf{eqn:introPurifierCases} or~\rf{eqn:introPurifierCases2}, the algorithm evaluates $f(\psi)$ with error at most $\eps$.
The query complexity of the algorithm is $\OO\sA[\log \frac1\eps]$ and its time complexity in the circuit model is polynomial in $\log\frac1\eps$.
\end{thm}

The majority-voting construction is a direct quantisation of a purely classical technique.
The logarithmic query complexity in the above theorem results in extra logarithmic factors that can be found in the analyses of a large variety of quantum algorithms.
In this paper, we develop an alternative, genuinely quantum approach to error reduction, that we call a \emph{purifier}.
The main feature is that, unlike majority voting, the query complexity of a purifier stays bounded by a constant \emph{no matter how small the error $\eps$ is}.
This is effectively equivalent to having an errorless algorithm (although, we cannot obtain an exact algorithm with a finite overhead in general).
We prove the following theorem in \rf{sec:purifiers}.

\mycutecommand{\purifier}{S_{\mathrm{pur}}}

\begin{thm}[Purifier]
\label{thm:introPurifier}
For any $\eps>0$, there exists a canonical transducer $\purifier$ with bidirectional access to an oracle like in~\rf{eqn:introrPurifierInput2} that has the following properties.
Assuming the oracle satisfies~\rf{eqn:introPurifierCases} or~\rf{eqn:introPurifierCases2}, the purifier transduces $\ket |0>$ into $\ket|f(\psi)>$ with error at most $\eps$ (in the sense to be made exact in \rf{sec:perturbedTransducers}).
Both its transduction and query complexities are bounded by a constant.
Its time complexity is $\OO\sA[s\log \frac1\eps]$ in the circuit model, and $\OO(\timeR)$ in the QRAG model, where $s$ is the number of qubits used by $\cM$.
\end{thm}

The purifier is inspired by the corresponding construction in~\cite{belovs:variations}, which used the dual adversary bound.
It had the same characteristic property of query complexity being bounded by a constant, but there are some differences.

\begin{itemize}
\item 
The purifiers in~\cite{belovs:variations} worked solely in the \emph{query} complexity settings.
The resulting dual adversary bound was for \emph{exact} function evaluation.
Also, by the nature of the adversary bound, it was implicitly assumed that there is only a \emph{finite} collection of possible input oracles, and, technically, for different collections, we obtain different purifiers.

\item
The purifiers in the current paper are constructed keeping both query and \emph{time} complexity in mind.
Also, the same purifier works for all (\emph{infinitely many}) possible input oracles.

\item 
Due to these improvements, the purifier ceases to be exact, but introduces a small error.
\end{itemize}

\mycutecommand{\Stoy}{S_{\mathrm{toy}}}
\mycutecommand{\wpsi}{\widetilde\psi}

\pfstart[Proof sketch of \rf{thm:introPurifier}]
We consider the Boolean case, as the general case can be easily obtained using the Bernstein-Vazirani algorithm~\cite{bernstein:quantumComplexity} as in, e.g.,~\cite[Section 4]{jeffery:kDist}.
Our construction is a quantum walk in the sense of \rf{sec:introWalks}.
In particular, it transduces $\ket |0>$ into $(-1)^{f(\psi)}\ket |0>$.

The overall structure is given by the following toy transducer $\Stoy$ in \rf{fig:introPurifier1}.
It is a quantum walk on the one-sided infinite line.
That is, $\Stoy = R_2R_1$, where the reflection $R_1$ is the product of the local reflections about the odd vertices $1,3,5,\cdots$, and $R_2$ about the positive even vertices $2,4,6,\cdots$.
The local reflection at the vertex $i>0$ is given by the $X$ operation, which maps
\begin{equation}
\label{eqn:toyOperations}
\ket |i-1> + \ket |i> \mapsto \ket |i-1> + \ket |i> 
\qqand
\ket |i-1> - \ket |i> \mapsto -\ket |i-1> + \ket |i>.
\end{equation}

\myfigure{\label{fig:introPurifier1}}
{
A toy transducer illustrating the overall construction of a purifier.
It is a quantum walk on the one-sided infinite line.
Each edge correspond to an element of the standard basis written below it.
The public space is given by $\ket |0>$.
We have two different catalysts (a) and (b) for the same initial state, the numbers above the edge giving the corresponding coefficients.
}
{
\[
\begin{tikzpicture}[auto]
\node at (-2,0) {(a)};
\node (0) at (0,0) [circle, draw] {$0$};
\foreach \x [remember=\x as \prevx (initially 0)] in {1,...,5}
{
    \node (\x) at (\x*2,0) [circle, draw] {$\x$};
    \draw (\prevx) to node [above, blue] {$1$} 
                     node [below, gray] {$\ket|\prevx>$} (\x);
}
\draw (5) to (11,0) node[right] {$\cdots$};
\begin{scope}[shift={(0,-2)}]
\node at (-2,0) {(b)};
\node (0) at (0,0) [circle, draw] {$0$};
\foreach \x [remember=\x as \prevx (initially 0)] in {1,...,5}
{
    \node (\x) at (\x*2,0) [circle, draw] {$\x$};
    \pgfmathparse{mod(\x,2) ? 1 : int(-1)}
    \draw (\prevx) to node [above, blue] {$\pgfmathresult$} 
                     node [below, gray] {$\ket|\prevx>$} (\x);
}
\draw (5) to (11,0) node[right] {$\cdots$};
\end{scope}
\end{tikzpicture}
\]
}

Now, it is not hard to see that in the negative and the positive case, we have the following mappings, which follow the general case of~\rf{eqn:QW_sequence_negative} and~\rf{eqn:QW_sequence_positive}:
\begin{equation}
\label{eqn:toyTransduce}
\ket |0> + \sum_{i=1}^{+\infty} \ket |i> \maps{\Stoy} \ket |0> + \sum_{i=1}^{+\infty} \ket |i>
\qqand
\ket |0> + \sum_{i=1}^{+\infty} (-1)^i\ket |i> \maps{\Stoy} - \ket |0> + \sum_{i=1}^{+\infty} (-1)^i\ket |i>.
\end{equation}
This formally gives us both $\ket |0> \transduce{\Stoy} \ket |0>$ and $\ket |0> \transduce{\Stoy} -\ket |0>$.
Of course, this does not contradict \rf{thm:introTransduce}, because the corresponding space has infinite dimension and both catalysts in~\rf{eqn:toyTransduce} have infinite norm.

Nonetheless, we will be able to utilise this general construction.
The first order of business is to reduce the norm of the catalysts.
Our next step towards a purifier will be a multidimensional quantum walk on the line, in the sense that each edge corresponds to a multidimensional subspace.
See \rf{fig:introPurifier2}.

\myfigure{\label{fig:introPurifier2}}
{
An improved transducer.
It is a multidimensional quantum walk.
Each edge corresponds to the element of the standard basis beneath it tensor multiplied by $\cM^{\otimes \infty}$.
The public space is given by $\ket |0>$.
For the negative and positive case, we have the catalysts like in (a) and (b), respectively, where the vector above the edge is placed in the corresponding subspace.
}
{
\[
\begin{tikzpicture}[auto]
\node at (-2,0) {(a)};
\node (0) at (0,0) [circle, draw] {$0$};
\foreach \x [remember=\x as \prevx (initially 0)] in {1,...,5}
{
    \node (\x) at (\x*2,0) [circle, draw] {$\x$};
    \draw (\prevx) to node [below, gray] {$\ket|\prevx>$} (\x);
}
\draw (0) to node[above, blue] {$1$} (1);
\draw (1) to node[above, blue] {$\wpsi$} (2);
\draw (2) to node[above, blue] {$\wpsi^{\otimes 2}$} (3);
\draw (3) to node[above, blue] {$\wpsi^{\otimes 3}$} (4);
\draw (4) to node[above, blue] {$\wpsi^{\otimes 4}$} (5);
\draw (5) to (11,0) node[right] {$\cdots$};
\begin{scope}[shift={(0,-2)}]
\node at (-2,0) {(b)};
\node (0) at (0,0) [circle, draw] {$0$};
\foreach \x [remember=\x as \prevx (initially 0)] in {1,...,5}
{
    \node (\x) at (\x*2,0) [circle, draw] {$\x$};
    \pgfmathparse{mod(\x,2) ? 1 : int(-1)}
    \draw (\prevx) to node [below, gray] {$\ket|\prevx>$} (\x);
}
\draw (0) to node[above, blue] {$1$} (1);
\draw (1) to node[above, blue] {$-\wpsi$} (2);
\draw (2) to node[above, blue] {$\wpsi^{\otimes 2}$} (3);
\draw (3) to node[above, blue] {$-\wpsi^{\otimes 3}$} (4);
\draw (4) to node[above, blue] {$\wpsi^{\otimes 4}$} (5);
\draw (5) to (11,0) node[right] {$\cdots$};
\end{scope}
\end{tikzpicture}
\]
}

We will define a vector $\wpsi$ that depends on $\psi$ and satisfies $\|\wpsi\| = 1-\Omega(1)$.
The initial coupling is given by
\[
\sum_{i=0}^\infty (-1)^{i\cdot f(\psi)} \ket |i> \ketO |\wpsi>^{\otimes i},
\]
and the transduction complexity is bounded by
\[
\sum_{i=1}^\infty \|\wpsi\|^{2i} = \OO(1).
\]

Let us now explain how to implement the local reflections~\rf{eqn:toyOperations} for this modified transducer.
Up to a sign, the content of the space incident to a vertex $i>0$ is given by
\begin{equation}
\label{eqn:introPurifier1}
\ket |i-1> \ketO |\wpsi>^{\otimes i-1} + (-1)^{f(\psi)} \ket |i> \ketO |\wpsi>^{\otimes i}.
\end{equation}
We use the input oracle to obtain the state
\begin{equation}
\label{eqn:introPurifier2}
\ket |i-1> \ketO |\wpsi>^{\otimes i-1}\ket |\psi> + (-1)^{f(\psi)} \ket |i> \ketO |\wpsi>^{\otimes i}.
\end{equation}
Later we can get back from~\rf{eqn:introPurifier2} to~\rf{eqn:introPurifier1} by uncomputing the $\psi$.
Therefore, bringing the term $\ket |i> \ketO |\wpsi>^{\otimes i-1}$ outside the brackets, it suffices to implement the mapping
\begin{equation}
\label{eqn:notToyOperations}
\begin{cases}
\ket |0> \ket |\psi> + \ket |1>\ketO|\wpsi> \mapsto \ket |0> \ket |\psi> + \ket |1>\ketO|\wpsi>,&\text{if $f(\psi)=0$;\; and}\\
\ket |0> \ket |\psi> - \ket |1>\ketO|\wpsi> \mapsto -\ket |0> \ket |\psi> + \ket |1>\ketO|\wpsi>,&\text{if $f(\psi)=1$.}
\end{cases}
\end{equation}

This is where the condition~\rf{eqn:introPurifierCases} comes into play.
We use the same rescaling idea as in~\cite{belovs:variations}, and define the state
\[
\wpsi = 
\begin{cases}
 \frac 1a \ket |0> \ket |\psi_0> + b \ket|1>\ket|\psi_1>,& \text{if $f(\psi)=0$};\\
a \ket |0> \ket |\psi_0> + \frac 1b \ket|1>\ket|\psi_1>,& \text{if $f(\psi)=1$}.\\
\end{cases}
\]
with
\[
a = \sqrt[4]{\frac{1-c+d}{1-c-d}}
\qqand
b = \sqrt[4]{\frac{c+d}{c-d}}.
\]
It is not hard to check that $\|\wpsi\|^2 = 1-\Omega(1)$.
Now the operation in~\rf{eqn:notToyOperations} reads as
\[
\begin{cases}
\sA[\ket |0> + \frac1a \ket |1>] \ket |0> \ket |\psi_0> +
\sA[\ket |0> + b \ket |1>]\ket |1> \ket |\psi_1>
\mapsto
\sA[\ket |0> + \frac1a \ket |1>] \ket |0> \ket |\psi_0> +
\sA[\ket |0> + b \ket |1>]\ket |1> \ket |\psi_1>,
\\
\sA[\ket |0> - a \ket |1>] \ket |0> \ket |\psi_0> +
\sA[\ket |0> - \frac1b \ket |1>]\ket |1> \ket |\psi_1>
\mapsto
\sA[-\ket |0> + a \ket |1>] \ket |0> \ket |\psi_0> +
\sA[-\ket |0> + \frac1b \ket |1>]\ket |1> \ket |\psi_1>,
\end{cases}
\]
which can be implemented as the 2-qubit reflection about the span of the states
\[
\sA[a\ket |0> + \ket |1> ]\ket |0>
\qqand
\sA[\ket |0> + b \ket |1> ]\ket |1>.
\]

Following the same logic as in the toy example, we obtain that $\ket |0> \transduce{} (-1)^{f(\psi)} \ket |0>$.
However, this time, both the transduction and the query complexities are bounded by a constant.
The problem is that this construction still requires infinite space.
We solve this by truncating the line after some vertex $D$.
This introduces an error, but since the norms of the vectors decrease exponentially with $i$, it suffices to take $D = \OO\sA[\log \frac 1\eps]$.
Finally, the transducer can be converted into a canonical form by the standard technique of \rf{prp:canoning}.
\pfend

\mycutecommand{\yy}{\vec{y}}

The main purpose of purifiers is to reduce error.
Let us give few examples.
First, recall the definition of the composed function $f\circ g$  from~\rf{eqn:introComposedFunction}:
\[
\begin{aligned}
\sS[f\circ g]&(z_{1,1}, \dots,z_{1,m},\;\; z_{2,1},\dots,z_{2,m},\;\;\dots\dots,\;\;z_{n,1},\dots,z_{n,m})\\
&= f\sA[
g(z_{1,1}, \dots,z_{1,m}), 
g(z_{2,1}, \dots,z_{2,m}),
\dots,
g(z_{n,1}, \dots,z_{n,m})].
\end{aligned}
\]
We use notation
\begin{equation*}
\yy_i = (z_{i,1}, \dots,z_{i,m})
\qqand
x = \sA[g(\yy_1), g(\yy_2),\dots g(\yy_n)]
\end{equation*}
so that $f(x) = \sS[f\circ g](z)$.
In the circuit model, we have the following result on the evaluation of this function.

\begin{thm}
\label{thm:introCompositionFunctionCircuit}
Let $A$ and $B$ be quantum algorithms in the circuit model that evaluate the functions $f$ and $g$, respectively, with bounded error.
Then, there exists a quantum algorithm in the circuit model that evaluates the function $f\circ g$ with bounded error in time complexity
\begin{equation}
\label{eqn:introCompositionFunctionCircuit}
\OO(L)\sA[T(A) + T(B) + s\log L]
\end{equation}
where $L$ is the worst-case Las Vegas query complexity of $A$, and $s$ is the space complexity of $B$.
\end{thm}

\pfsketch
We obtain the transducer $S_A$ in the circuit model using \rf{thm:introProg->Transducer}.
The transduction and query complexities of $S_A$ are bounded by $L$.
Let $\purifier$ be the purifier for the input oracle $B$.
Consider the transducer $S = S_A\circ\bigoplus_i \purifier$, see \rf{fig:introCompositionFunctionCircuit}.
The transducer $\purifier$ on the input oracle $B(O_{\yy_i})$ evaluates $g(y_i)$ with diminished error.
It suffices to make error somewhat smaller than $1/L$.
Then, $S$ on the input oracle $\bigoplus_i B(O_{\yy_i})$ evaluates $f\circ g$ with bounded error.

\myfigure{\label{fig:introCompositionFunctionCircuit}}
{
A composition scheme for \rf{thm:introCompositionFunctionCircuit}.
The input oracle $O_z$ is broken down as $\bigoplus_i O_{\yy_i}$.
The composed transducer contains the elements inside the blue box.
The program $B$ is executed as is, serving as an input oracle to the composed transducer.
}
{
\def\funkcijaB(#1,#2)#3{  
    \edef\x{#1}
    \edef\y{#2}
    \begin{scope}[shift={(#1,#2)}]
        \draw (0,0) rectangle (2,-1) node[pos=0.5] {$\purifier$};
        \draw (0,-1.5) rectangle (2,-2.5) node[pos=0.5] {$B$};
        \draw (0,-3) rectangle (2,-4) node[pos=0.5] {$O_{\yy_#3}$};
        \draw[purple] (1,-1) to (1,-1.5);
        \draw[purple] (1,-2.5) to (1,-3);
    \end{scope}
}
\[
\begin{tikzpicture}[every path/.append style={thick,->}]
\draw (7,0) rectangle (9,-1) node[pos=0.5] {$S_A$};
\funkcijaB(3,-2)1
\funkcijaB(6,-2)2
\node at (9,-2.5) {\Large$\cdots$};
\funkcijaB(10,-2)n
\draw[purple, out=270, in=90] (7.2, -1) to (4,-2) ;
\draw[purple, out=270, in=90] (7.5, -1) to (7,-2) ;
\draw[purple, out=270, in=90] (8.8, -1) to (11,-2) ;
\draw[rounded corners=10pt, blue] (2,0.5) rectangle (13,-3.25);
\end{tikzpicture}
\]
}

The transduction complexity of $\purifier$ on a unit admissible vector is $\OO(1)$.
Hence, by~\rf{eqn:compositionTransductionUpper}
\[
W(S, O_z, \ket|0>) \le W(S_A, O_x, \ket|0>) + \OO(1)\cdot L(S_A, O_x, \ket |0>) = \OO(L). 
\]
All the purifiers can be implemented in parallel, hence by \rf{prp:introFunctional},
\[
T(S) = T(S_A) + T(\purifier) = \OO\sA[T(A)] + \OO(s\log L).
\]
The theorem follows from \rf{thm:introImplementation} applied to $S$, where we replace execution of the input oracle by the execution of $B$.
Again, all the $B$ can be executed in parallel.
\pfend

In the QRAG model, we can easily deal with different functions $g_i$, like in the following function, which we already considered in~\rf{eqn:randomComposedFunction}:
\begin{equation}
\label{eqn:randomComposedFunctionCopy}
f\sA[
g_1(z_{1,1}, \dots,z_{1,m}), 
g_2(z_{2,1}, \dots,z_{2,m}),
\dots,
g_n(z_{n,1}, \dots,z_{n,m})].
\end{equation}
We again use notation
\[
\yy_i = (z_{i,1}, \dots,z_{i,m})
\qqand
x = \sA[g_1(\yy_1), g_2(\yy_2),\dots g_n(\yy_n)].
\]

\begin{thm}
\label{thm:introCompositionFunctionQRAG}
Consider the function as in~\rf{eqn:randomComposedFunctionCopy}.
Let $A$ and $B_1,\dots, B_n$ be quantum algorithms that evaluate the functions $f$ and $g_1,\dots, g_n$, respectively, with bounded error.
Assuming the QRAG model and QRAM access to the description of $A$ and $B_1,\dots, B_n$, there exists a quantum algorithm that evaluates the function~\rf{eqn:randomComposedFunctionCopy} with bounded error in time complexity
\begin{equation}
\label{eqn:introCompositionFunctionQRAG}
\OO(\timeR)\max_x\sB[T(A) + \sum_{i=1}^n T(B_i) L^{(i)}_x(A)].
\end{equation}
Here $L^{(i)}_x(A)$ is the $i$-th partial Las Vegas complexity $L^{(i)}\sA[A, O_x, \ket|0>]$ of the algorithm $A$ on the input oracle encoding $x$.
\end{thm}

\pfsketch
We obtain the transducer $S_A$ in the QRAG model using \rf{thm:introProg->Transducer}.
Its transduction complexity is $T(A)$.
We assume all $B_i$ have the same range and the same error.
Let $\purifier$ be the corresponding purifier.
Finally, let $S_{B_i}$ be the transducer in the QRAG model corresponding to $B_i$.
Consider the composed transducer $S$ as in \rf{fig:introCompositionFunctionQRAG}.
Similarly to \rf{thm:introCompositionFunctionCircuit}, 
the transducer $S$ on the input oracle $O_z = \bigoplus_i O_{\yy_i}$ evaluates the function $f\circ g$ with bounded error provided that the error of the purifier is sufficiently smaller than $1/L$.

\myfigure{\label{fig:introCompositionFunctionQRAG}}
{
A composition scheme for \rf{thm:introCompositionFunctionQRAG}.
The input oracle $O_z$ is again broken down as $\bigoplus_i O_{\yy_i}$.
The composed transducer $S$ contains the elements inside the blue box.
This time every $B_i$ is turned into a transducer and partakes in the composition.
}
{
\def\funkcijaB(#1,#2)#3{  
    \edef\x{#1}
    \edef\y{#2}
    \begin{scope}[shift={(#1,#2)}]
        \draw (0,0) rectangle (2,-1) node[pos=0.5] {$\purifier$};
        \draw (0,-1.5) rectangle (2,-2.5) node[pos=0.5] {$S_{B_#3}$};
        \draw (0,-3) rectangle (2,-4) node[pos=0.5] {$O_{\yy_#3}$};
        \draw[purple] (1,-1) to (1,-1.5);
        \draw[purple] (1,-2.5) to (1,-3);
    \end{scope}
}
\[
\begin{tikzpicture}[every path/.append style={thick,->}]
\draw (7,0) rectangle (9,-1) node[pos=0.5] {$S_A$};
\funkcijaB(3,-2)1
\funkcijaB(6,-2)2
\node at (9,-2.5) {\Large$\cdots$};
\funkcijaB(10,-2)n
\draw[purple, out=270, in=90] (7.2, -1) to (4,-2) ;
\draw[purple, out=270, in=90] (7.5, -1) to (7,-2) ;
\draw[purple, out=270, in=90] (8.8, -1) to (11,-2) ;
\draw[rounded corners=10pt, blue] (2,0.5) rectangle (13,-4.75);
\end{tikzpicture}
\]
}

Using that the transduction and the query complexity of $\purifier$ are $\OO(1)$, we obtain that
$
\Wmax(\purifier \circ S_{B_i}, O_{\yy_i}) = \OO\sA[T(B_i)].
$
Therefore, by~\rf{eqn:compositionTransductionMultipleUpper} and using that $A$ and $S_A$ have the same query state:
\[
W\sA[S, O_z, \ket|0>] \le W(S_A, O_x, \ket |0>) + \sum_i \OO\sA[T(B_i)] L^{(i)}\sA[A, O_x, \ket|0>].
\]

For the time complexity, $S$ is a composition of three transducers, where the last two are direct sums.  All three of them have time complexity $\OO(\timeR)$, hence, this is also the time complexity of $S$.
The theorem follows from~\rf{thm:introImplementation}.
\pfend

Comparison between Theorems~\ref{thm:introCompositionFunctionCircuit} and~\ref{thm:introCompositionFunctionQRAG} is similar to the comparison between Theorems~\ref{thm:introCompositionCircuit} and~\ref{thm:introCompositionQRAG} in \rf{sec:introComposition}.
The second theorem considers a more general case~\rf{eqn:randomComposedFunctionCopy}, and its formulation is close to~\rf{eqn:randomComposition2}.
On the other hand, in~\rf{eqn:introCompositionFunctionCircuit}, $T(B)$ will most likely dominate $s\log L$, so the latter can be removed.
Also, if $T(A)$ is smaller than $T(B)$, the whole expression is dominated by $L\cdot T(B)$, which is what we would like to have.

Finally, in the QRAG model, any bounded-error quantum algorithm $A$ can be turned into an essentially exact transducer $S_A$ such that, up to constant factors, its transduction complexity is $T(A)$ and its query complexity is the query complexity of $A$.
The latter transducers can be composed in multiple layers as in \rf{thm:introCompositionTree}.

\section{Preliminaries}
\label{sec:preliminaries}
If not said otherwise, a \emph{vector space} is a finite-dimensional complex inner product space.  They are denoted by calligraphic letters.  We assume that each vector space has a fixed orthonormal basis, and we often identify an operator with the corresponding matrix.
$I_{\cX}$ stands for the identity operator in $\cX$.
The inner product is denoted by $\ip<\cdot,\cdot>$.
$A^*$ stands for the adjoint linear operator.
All projectors are orthogonal projectors.
We use ket-notation to emphasise that a vector is a state of a quantum register, or to denote the elements of the computational basis.

We use $\OO$ for the big-Oh notation in order to distinguish from $O$, which we use for input oracles.
$\OO_A$ means that the constant may depend on $A$.
If $P$ is a predicate, we use $1_P$ to denote the corresponding indicator variable; which is equal to 1 if $P$ is true, and to 0 otherwise.

\subsection{Query Algorithms}
\label{sec:prelimQuery}

In this section, we briefly describe the model of quantum query algorithms, and define quantum Las Vegas query complexity.
The query model itself is essentially standard~\cite{buhrman:querySurvey, cleve:quantumComplexityTheory}, with the main difference that we do not restrict ourselves to the evaluation of functions, and also allow for multiple input oracles, which can be arbitrary unitaries.
The notion of quantum Las Vegas query complexity is relatively new~\cite{belovs:LasVegas}.

\mycutecommand\Ic{I^\circ}
\mycutecommand\Ib{I^\bullet}
\mycutecommand\Iw{I^\circ}
\mycutecommand\cHw{\cH^\circ}
\mycutecommand\cHq{\cH^\bullet}
\mycutecommand\cHt{\cH^\uparrow}

A quantum query algorithm $A$ works in space $\cH$, which we call the \emph{workspace} of the algorithm.
The algorithm is given an oracle $O$, which is a unitary%
\footnote{
While~\cite{belovs:LasVegas} define more general input oracles, we, for simplicity, consider only unitary input oracles in this paper.
}
 in some space $\cM$.
The interaction between the algorithm and the oracle is in the form of queries, which we are about to define.
The workspace is decomposed as $\cH = \cHw\oplus \cHq$, where the oracle is only applied to the second half.
Also, $\cHq = \cHt \otimes \cM$ for some $\cHt$, and the \emph{query} is
\begin{equation}
\label{eqn:query}
\tO = \Iw \oplus I\otimes O,
\end{equation}
where $\Iw$ and $I$ are the identities in $\cHw$ and $\cHt$, respectively.

In terms of registers, we assume the decomposition $\cH = \cHw\oplus \cHq$ is marked by a register $\cR$ so that $\ket R|0>$ corresponds to $\cHw$ and $\ket R|1>$ to $\cHq$.
Then, the query $\tO$ is an application of $O$, controlled by $\cR$, where $O$ acts on some subset of the registers (which constitute $\cM$).

The \emph{quantum query algorithm} $A=A(O)$ is a sequence of linear transformations in $\cH$:
\begin{equation}
\label{eqn:preAlgorithm}
A(O) = U_Q\, \tO\, U_{Q-1}\,\tO\,\cdots U_{1}\, \tO\, U_0,
\end{equation}
where $U_t$ are some input-independent unitaries in $\cH$.
See \rf{fig:queryAlgorithm}.
Thus, the algorithm implements a transformation $O\mapsto A(O)$: from the input oracle $O$ in $\cM$ to the linear operator~\rf{eqn:preAlgorithm} in $\cH$.

\myfigure{\label{fig:queryAlgorithm}}
{
A graphical illustration of a quantum query algorithm with $Q=3$ queries.
The algorithm interleaves input-independent unitaries $U_t$ with queries $\tO = \Iw \oplus I\otimes O$.
The intermediate state $\psi_t$ after $U_{t-1}$ and before the $t$-th query is decomposed as $\psi_t^\circ\oplus \psi_t^\bullet$, where only the second half is processed by the oracle.
}
{
\newcommand{\OneIteration}[1]{
    \edef\indxx{#1}
    \begin{scope}[shift={(4*\indxx-4,0)}]
        \draw (0.9,0.35) rectangle (2.3, 1.65) node[pos=0.5] {\Large $I\otimes O$};
        \draw (3.1,0.15) rectangle (4, 3.35) node[pos=0.5] {\Large $U_{\indxx}$};
        \draw[\witnesscolor,->] (0, 2.65) to node[above]{$\psi_{\indxx}^\circ$} (3.1,2.65);
        \draw[\nonquerycolor,->] (0, 1) to node[above]{$\psi_{\indxx}^\bullet$} (0.9,1);
        \draw[\querycolor,->] (2.3, 1) to (3.1,1);
    \end{scope}
}
\[
\begin{tikzpicture}[every node/.style={font=\scriptsize}, every path/.append style={thick,->}]
\draw[\xicolor] (-1.9,1.9) to node[above]{$\xi$} (-0.9,1.9);
\draw (-0.9,0.15) rectangle (0, 3.35) node[pos=0.5] {\Large $U_{0}$};
\OneIteration{1}
\OneIteration{2}
\OneIteration{3}
\draw[\taucolor] (12,1.9) to node[above]{$\tau$} (13,1.9);
\end{tikzpicture}
\]
}

We will generally work with the state conversion formalism.  We say that $A$ transforms $\xi$ into $\tau$ on oracle $O$, if $A(O)\xi = \tau$.
We say that $A$ does so $\eps$-approximately if $\normA|\tau - A(O)\xi|\le\eps$.
In this context, we will often call $\eps$ the \emph{error} of the algorithm.

Let us make two important remarks on the structure of thus defined query algorithms.

\begin{rem}[Alignment]
\label{rem:aligned}
Note that all the queries in~\rf{eqn:preAlgorithm} are identical, i.e., the oracle $O$ is always applied to the same registers and is always controlled by $\cR$.
To acknowledge this, we say that all the queries in $A$ are \emph{aligned}.
Usually, this is not important, but the alignment property will play a significant role in this paper, in particular in Sections~\ref{sec:properties} and~\ref{sec:programs->transducers}.
The main reason is that for the aligned program we can perform all the queries in parallel (assuming we have the intermediate states $\psi_t$ from \rf{fig:queryAlgorithm} somehow).
\end{rem}

\newcommand{\bi}[1]{\overleftrightarrow{#1}}

\begin{rem}[Unidirectionality]
The input oracle is unidirectional: the algorithm only has access to $O$.
This suffices for most of our results.
Quite often, however, bidirectional access to the input oracle is allowed, where the algorithm can query both $O$ and $O^*$.
The latter is a special case of the former, as bidirectional access to $O$ is equivalent to unidirectional access to $O\oplus O^*$.
As this situation will be common in some sections of the paper, we utilise the following piece of notation:
\begin{equation}
\label{eqn:bi}
\bi{O} = O \oplus O^*.
\end{equation}
\end{rem}

The standard complexity measure of a quantum query algorithm is $Q=Q(A)$: the total number of times the query $\tO$ is executed.
It was called Monte Carlo complexity in~\cite{belovs:LasVegas} in order to distinguish it from Las Vegas complexity defined next.

Let $\Pi^\bullet$ be the projector onto $\cHq$.
The state processed by the oracle on the $t$-th query is $ \psi_t^\bullet = \Pi^\bullet U_{t-1}\, \tO\, U_{t-2} \cdots \tO\, U_0 $, and the \emph{total query state} is 
\begin{equation}
\label{eqn:totalQueryState}
q(A, O, \xi) = \bigoplus_{t=1}^Q \psi_t^\bullet.
\end{equation}
This is the most complete way of specifying the work performed by the input oracle $O$ in the algorithm $A$ on the initial state $\xi$.
It is a member of $\cE\otimes \cM$ for some space $\cE$ (the latter actually being equal to $\bC^Q\otimes \cHt$).
The simplest way to gauge the total query state is by defining the corresponding \emph{quantum Las Vegas query complexity}:
\begin{equation}
\label{eqn:LasVegasComplexity}
L(A, O, \xi) = \|q(A, O, \xi)\|^2.
\end{equation}
As mentioned in \rf{sec:conceptualQuantum}, we extend the definitions~\rf{eqn:totalQueryState} and~\rf{eqn:LasVegasComplexity} for the case $\xi'\in \cE\otimes \cH$ for some $\cE$ using identities~\rf{eqn:queryStateExtended} and~\rf{eqn:LasVegasExtended}.
\medskip

\subsection{Multiple Input Oracles}
\label{sec:prelimMultipleOracles}

It is also possible for an algorithm to have access to several input oracles $O^{(1)},\dots,O^{(r)}$.
Las Vegas query complexity can naturally accommodate such a scenario.
Indeed, access to several input oracles is equivalent to access to the one combined oracle
\begin{equation}
\label{eqn:severalOracles}
O = O^{(1)}\oplus O^{(2)}\oplus \cdots \oplus O^{(r)}.
\end{equation}
Consequently, the space of the oracle has a similar decomposition $\cM = \cM^{(1)}\oplus\cdots\oplus \cM^{(r)}$, where $O^{(i)}$ acts in $\cM^{(i)}$.
The total query state can also be decomposed into \emph{partial query states}
\[
q(A,O,\xi) = q^{(1)}(A,O,\xi)\oplus q^{(2)}(A,O,\xi) \oplus \cdots \oplus q^{(r)}(A,O,\xi),
\]
where $q^{(i)}(A,O,\xi)$ is processed by $O^{(i)}$.
This gives partial Las Vegas query complexities
\[
L^{(i)}(A, O, \xi) = \normA|q^{(i)}(A,O,\xi)|^2.
\]

In terms of registers, it can be assumed that the input oracle is controlled by some register $\reg R$, where the value $\ket R|0>$ indicates no application of the input oracle, and $\ket R|i>$ with $i>0$ indicates the $i$-th input oracle $O^{(i)}$.
We note that we use $i$ in $\ket R|i>$ only as a label.
In particular, we do not perform any arithmetical operations on them.
Therefore, $\ket R|i>$ can have a complicated internal encoding that can facilitate the application of the oracle.

The assumption on the register $\cR$ in this section is not necessarily in contradiction with the assumptions of \rf{sec:prelimQuery}, as $\reg R$ can have a separate qubit that indicates whether $i$ in $\ket R|i>$ is non-zero.
This qubit can serve as $\reg R$ in the sense of \rf{sec:prelimQuery}.

The case of usual Monte Carlo query complexity requires additional changes, as per now it turns out that all the oracles are queried the same number of times, $Q$.
One way to allow for different number of queries is as follows.
Similarly to~\rf{eqn:query}, define the query to the $i$-th input oracle as
\[
\tO^{(i)} = \Ic \oplus I \otimes 
\s[I^{(1)} \oplus \cdots\oplus I^{(i-1)} \oplus O^{(i)} \oplus I^{(i+1)}\oplus\cdots\oplus I^{(r)} ],
\]
where the decomposition in the brackets is the same as in~\rf{eqn:severalOracles}.
In other words, $\tO^{(i)}$ is just the application of $O^{(i)}$ controlled by $\ket R|i>$.
The query algorithm is then defined as
\begin{equation}
\label{eqn:algorithmMultipleOracles}
A(O) = U_Q\, \tO^{(s_Q)}\, U_{Q-1}\,\tO^{(s_{Q-1})}\,\cdots\, U_{2}\, \tO^{(s_2)}\, U_{1}\, \tO^{(s_1)}\, U_0,
\end{equation}
where $s_1,s_2,\dots,s_Q \in [r]$.
The number of invocations of the $i$-th oracle, $Q^{(i)}$, is defined as the number of $s_j$ equal to $i$ in~\rf{eqn:algorithmMultipleOracles}.

\subsection{Evaluation of Functions}
\label{sec:prelimFunctions}
The \emph{standard} settings for quantum algorithms evaluating a (partial) function $f\colon D\to [p]$ with $D\subseteq [q]^n$ are as follows.
For an input $x\in [q]^n$, the corresponding input oracle acts in $\bC^n \otimes \bC^q$ as
\begin{equation}
\label{eqn:standardOracle}
O_x \colon \ket|i>\ket |b> \mapsto \ket |i> \ket |b \oplus x_i>
\end{equation}
for all $i\in[n]$ and $b\in [q]$.
Here $\oplus$ stands for the bitwise XOR, and we assume that $q$ is a power of 2 (we can ignore the inputs outside of the domain).

The algorithm $A$ itself has a special output register isomorphic to $\bC^p$.
After finishing the algorithm, measuring the output register should yield the value $f(x)$.
This either happens with probability 1 (for exact algorithms), or with probability at least $1/2+d$ for some constant $d>0$ (bounded error).

This definition is nice because there is one well-defined input oracle $O_x$ for each input.
Also $O_x^2 = I$, which makes uncomputing very easy.
Unfortunately, the standard definition has a drawback that the input oracle of the algorithm has a more restricted form than the one required for the algorithm itself.
This is problematic if the algorithm is expected to be used as an input oracle for another algorithm.
This issue is solved by noting that $\ket |b> \mapsto \ket |b\oplus f(x)>$ can be implemented by evaluating $f(x)$, performing the XOR operation, and uncomputing $f(x)$.
This increases the complexity by a factor of 2, which is fine if we ignore constant factors.
We call it \emph{robust} evaluation of function, as the action of the algorithm is specified for all input states, not just $\ket|0>$.

However, constant factors can be important, for instance, in iterated functions, where such factors appear as a base of the exponent, or in the settings of Las Vegas complexity in~\cite{belovs:LasVegas}, where precise complexity is sought for.
In this case, a more homogenous definition would be appreciated.

We follow one such approach from~\cite{belovs:variations}, which we call \emph{state-generating}.
We say that the input oracle $O_x$ encodes the input string $x\in[q]^n$ if it performs the transformation 
\begin{equation}
\label{eqn:stateGeneratingOracle}
O_x \colon \ket |i>\ket|0> \mapsto\ket |i>\ket|x_i>
\end{equation}
for all $i\in [n]$.
The action of this oracle on $\ket |i>\ket|0>$ is identical to that of~\rf{eqn:standardOracle}, but we do not require anything for other states.
In other words, the admissible subspace of $O_x$ consists of vectors having $\ket|0>$ in the second register.
The admissible subspace of the algorithm itself is spanned by $\ket|0>$.
On that, it has to perform the transformation $\ket |0> \mapsto \ket |f(x)>$.
The algorithm has bidirectional access to $O_x$, which we treat as unidirectional access to $\bi{O_x}$ defined in~\rf{eqn:bi}.
It is possible to assume the input oracle $O_x$ is a direct sum of $n$ unitaries acting in $\bC^q$ in order to apply the multiple input oracle settings from \rf{sec:multipleInputOracles}.

As the initial state is always $\ket |0>$, and $O_x$ is essentially determined by $x$, we will write
\begin{equation}
\label{eqn:function_QueryComplexity}
L_x (A) = L\sA[A, \bi{O_x}, \ket |0>]
\qqand
L(A) =  \max_{x\in D} L_x(A).
\end{equation}
More precisely, we define $L_x(A)$ as the supremum over all input oracles 
$O_x$ that encode the input $x$. 
We use similar notation for $L^{(i)}_x$ and $L^{(i)}$. 

This approach has a number of advantages.
First, it casts function evaluation as a special case of state conversion with state-generating input oracles~\cite{belovs:variations}.
Second, the algorithm can be directly used as a part of the input oracle for another algorithm without any uncomputation.
Finally, this definition does not involve the somewhat arbitrary XOR operation and may be, thus, regarded as being more pure.
In particular, it makes sense to ask for the precise value of its quantum Las Vegas query complexity (and not just only up to constant factors).

This approach has a number of disadvantages.
First, we have to explicitly allow bidirectional access to $O_x$ in order to allow uncomputing, as it is no longer the case that $O_x$ is its own inverse.
More importantly, though, neither the action nor the Las Vegas complexity of the algorithm is specified for the initial states orthogonal to $\ket |0>$.
For our own algorithms, we can design them so that $O_x$ is only executed with $\ket|0>$ in the second register (maybe after some perturbation, see \rf{sec:perturbed}).
But, if we are dealing with an arbitrary algorithm, we have no such guarantees.

\subsection{Circuit Model}
\label{sec:prelimCircuit}

\mycutecommand{\TC}{T_{\mathrm C}}

We assume that the space of the algorithm is embedded into a product of qubits, $(\bC^2)^{\otimes s}$, for some $s$ called the space complexity of the algorithm.
A quantum program is a product of elementary operations called gates:
\begin{equation}
\label{eqn:program}
A = G_TG_{T-1}\cdots G_1.
\end{equation}
In the circuit model, each gate $G_i$ is usually a 1- or a 2-qubit operation that can be applied to any qubit or a pair of qubits.
The number of elementary operations, $T$, is called the time complexity of $A$, and is denoted by $T(A)$.
We assume a universal gate set, so that every unitary can be written as a quantum program.
We do not focus too much on a particular model, as they are all essentially equivalent.

In a query algorithm like in \rf{sec:prelimQuery}, each execution of the input oracle $\tO$ also traditionally counts as one elementary operation.
In other words, each unitary in~\rf{eqn:preAlgorithm} can be decomposed into elementary gates as in~\rf{eqn:program} to give a corresponding program in the circuit model.
We use $T$ to denote its time complexity, and $Q$ to denote its Monte Carlo query complexity.

We will often require an algorithm like in~\rf{eqn:program} to be executed conditionally, that is, controlled by the value of some external qubit.
In other words, we would like to perform an operation $A^{\mathrm{c}}$ of the form
\[
\ket |0>\ket |\xi> \maps{A^{\mathrm{c}}} \ket |0>\ket |\xi>
,\qquad
\ket |1>\ket |\xi> \maps{A^{\mathrm{c}}} \ket |1>\ket |A\xi>,
\]
where the first qubit is the control qubit.
We will denote time complexity of this procedure by $\TC(A)$.

Since it is possible to substitute each $G_i$ in~\rf{eqn:program} by its controlled version, we have that $\TC(A) = \OO(T(A))$.
But it is often possible to do better.
For instance, assume that $A$ is of the form $A_2A_1^{\mathrm{c}}$, i.e., a large chunk of $A$ is already conditioned.
We have that $T(A) = T(A_2) + \TC(A_1)$.
On the other hand, $\TC(A) = \TC(A_2) + \TC(A_1) + \OO(1)$, because we can calculate the AND of the two control qubits of $A_1$ into a fresh qubit, execute $A_1$ controlled by this fresh qubit, and uncompute the new qubit afterwards.
In other words, the constant factor of $\TC(A) = \OO(T(A))$ is not getting accumulated with each new conditioning, but is, in a way, paid only once.

\subsection{QRAG model}
\label{sec:QRAG}

The QRAG model extends the circuit model by allowing the following Quantum Random Access Gate:
\[
\mathrm{QRAG}\colon 
\ket |i> \ket |b> \ket |x_1,\dots,x_{i-1}, x_i, x_{i+1},\dots, x_m>
\mapsto
\ket |i> \ket |x_i> \ket |x_1,\dots,x_{i-1}, b, x_{i+1},\dots, x_m>,
\]
where the first register is an $m$-qudit, and the remaining ones are quantum words (i.e., quantum registers large enough to index all the qubits in the program).
We assume the QRAG takes time $\timeR$ as specified in~\rf{eqn:timeR}.
Note that this gate would require time $\Omega(m)$ to implement in the usual circuit model, as it depends on all $m+2$ registers.

This should not be confused with the QRAM model, which allows oracle access to an array of classical registers $x_1,x_2,\dots,x_m$:
\[
\mathrm{QRAM}\colon  
\ket |i> \ket |b> \mapsto \ket |i> \ket |b\oplus x_i>,
\]
where $\oplus$ stands for the bit-wise XOR.  The difference is that $x_i$s are being fixed during the execution of the quantum procedure (but they may be changed classically between different executions).
The QRAG model is more powerful than the QRAM model.

The main reason we need the QRAG is the following result (see, e.g.,~\cite{jeffery:kDist} for a formal statement, although the same construction has been used elsewhere including~\cite{ambainis:searchVariableTimes, cornelissen:spanProgramsTime, jeffery:subroutineComposition}):

\begin{thm}[Select Operation]
\label{thm:select}
Assume the QRAG model and that we have QRAM access to a description of a quantum program as in~\rf{eqn:program} in some space $\cH$, where each gate $G_i$ either comes from a fixed set of 1- or 2-qubit operations (which can be applied to different qubits each time), or is a QRAG.
Let $\cI$ be a $T$-qudit.  Then, the following Select operation
\begin{equation}
\label{eqn:select}
\sum_i \ket I|i>\ket H|\psi_i> \mapsto \sum_i \ket I|i> \ket H |G_i\psi_i>
\end{equation}
can be implemented in time $\OO(\timeR)$.
\end{thm}

\pfstart[Proof sketch]
We add a number of scratch registers to perform the following operations.
Conditioned on $i$, we use the QRAM to read the description of $G_i$.
We switch the arguments of $G_i$ into the scratch space using the QRAG.
We apply the operation $G_i$ on the scratch space.
We switch the arguments back into memory, and erase the description of $G_i$ from the scratch memory.
All the operations besides applying $G_i$ take time $\OO(\timeR)$.
Application of a usual gate $G_i$ takes time $\OO(1)$ as the gate set is fixed, or $\OO(T_R)$ if $G_i$ is a QRAG.
\pfend

It is also possible to not store the whole program in memory, but compute it on the fly, in which case the complexity of this computation should be added to the complexity stated in \rf{thm:select}.

\begin{cor}[Parallel execution of programs]
\label{cor:selectProgram}
Assume that in the settings of \rf{thm:select} we have QRAM access to an array storing descriptions of $m$ quantum programs $A^{(1)},\dots, A^{(m)}$.
Let $\cI$ be a $m$-qudit.  Then, the following operation
\[
\sum_i \ket I|i>\ket H|\psi_i> \mapsto \sum_i \ket I|i> \ketA H |A^{(i)}\psi_i>,
\]
can be implemented in time $\OO\sA[\timeR\cdot \max_i T(A^{(i)})]$.
\end{cor}

\pfstart[Proof sketch]
Use \rf{thm:select} to implement the first gate in all of $A^{(i)}$ in parallel, then the second one, and so on until the time mark $\max_i T(A^{(i)})$.
\pfend

Another important primitive is the random access (RA) input oracle.
If $O\colon \cM\to\cM$ is an input oracle, then its RA version acts on $\cJ\otimes \cM^{\otimes K}$, where $\cJ$ is a $K$-qudit.
If the register $\cJ$ contains value $i$, the input oracle is applied to the $i$-th copy of $\cM$ in $\cM^{\otimes K}$.

The idea behind this is that if $O$ is implemented as a subroutine, then the RA input oracle is a special case of \rf{cor:selectProgram}, where all $A^{(i)}$ are the same, but act on different sets of registers (which is easy to define using $\ket J|i>$ as an offset).
Therefore, it makes sense to define the RA input oracle as an elementary operation in the QRAG model.

\subsection{Perturbed Algorithms}
\label{sec:perturbed}

The following lemma is extremely useful in quantum algorithms, but for some reason has seldom experienced the honour of being explicitly stated.

\begin{lem}
\label{lem:surgery}
Assume we have a collection of unitaries $U_1,\dots, U_m$ all acting in the same vector space $\cH$.
Let $\psi_0',\dots,\psi_m'$ be a collection of vectors in $\cH$ such that
\[
\psi_t' = U_t \psi_{t-1}'
\]
for all $t$.
Let $\psi_0,\dots, \psi_m$ be another collection of vectors in $\cH$ such that $\psi_0 = \psi'_0$ and
\begin{equation}
\label{eqn:surgeryPerturbation}
\normA | \psi_t - U_i \psi_{t-1} | \le \eps_i
\end{equation}
for all $i$.
Then,
\begin{equation}
\label{eqn:surgeryResult}
\normA | \psi_m - \psi'_m | \le \sum_{t=1}^m \eps_t.
\end{equation}
\end{lem}

\pfstart
Denote by $V_t$ the product $U_m U_{m-1} \cdots U_{t+1}$.
In particular, $V_m = I$.
Then,
\[
\psi_m - \psi'_m = V_m \psi_{m} - V_0 \psi_0 
=\sum_{t=1}^m \sA[ V_t \psi_{t} - V_{t-1} \psi_{t-1}]
= \sum_{t=1}^m V_t \sA[ \psi_{t} - U_t \psi_{t-1}].
\]
We obtain~\rf{eqn:surgeryResult} from the triangle inequality using~\rf{eqn:surgeryPerturbation} and the unitarity of $V_t$.
\pfend

In the application of this lemma, $U_t$ stands for sequential sections of a quantum algorithm.
The vectors $\psi'_t$ form the sequence of states the algorithm goes through during its execution.
The vectors $\psi_t$ is an idealised sequence, which is used instead of $\psi_t'$ in the analysis.

We call the difference between $\psi_t$ and $U_t \psi_{t-1}$ a (conceptual) perturbation.
The expression in~\rf{eqn:surgeryPerturbation} is the size of the perturbation.
Therefore, the Eq.~\rf{eqn:surgeryResult} can be stated as the total perturbation of the algorithm does not exceed the sum of the perturbations of individual steps.
If this sum is small, the final state $\psi_m$ of the analysis is not too far away from the actual final state $\psi'_m$ of the algorithm.

This lemma is implicitly used every time an approximate version of a quantum subroutine is used, which happens in pretty much every non-trivial quantum algorithm.

\subsection{Efficient Implementation of Direct-Sum Finite Automata}
\label{sec:automaton}
We will repeatedly use the following construction in this paper, for which we describe a time-efficient implementation.
We call it direct-sum quantum finite automata due to its superficial similarity to quantum finite automata.

The space of the algorithm is $\cK\otimes \cP\otimes \cH$, where $\cK$ is a $K$-qudit, $\cP$ is a qubit, and $\cH$ is an arbitrary space.
Additionally, we have black-box access to $K$ unitaries $S_0,\dots,S_{K-1}$ in $\cP\otimes \cH$.
The algorithm is promised to start in the state of the form
\begin{equation}
\label{eqn:automatonInitialState}
\ket K|0> \ket P|0> \ket H|\phi> + \sum_{t=0}^{K-1} \ket K|t>\ket P|1> \ket H|\psi_t>,
\end{equation}
and it has to perform the following transformation
\begin{equation}
\label{eqn:automaton}
\begin{minipage}{.9\linewidth}
\begin{itemize}
\item For $t=0,1,\dots,K-1$:\negmedskip
\begin{itemize}\itemsep=0pt
\item[(a)] Execute $S_t$ on $\cP\otimes \cH$ conditioned on $\ket K|t>$.
\item[(b)] Conditioned on $\ket P|0>$, replace $\ket K|t>$ by $\ket K|t+1>$.
\end{itemize}
\end{itemize}
\end{minipage}
\end{equation}
Let us elaborate on the ``replace'' in point (b).
It is not hard to show by induction that the $\ket P|0>$-part of the state contains a vector of the form $\ket K|t>\ket H|\phi_{t+1}>$.
This vector has to be replaced by $\ket K|t+1>\ket H|\phi_{t+1}>$.
We also assume that $\ket K|K>$ is identical with $\ket K|0>$.
For a graphical illustration refer to \rf{fig:automaton}.

\myfigure{\label{fig:automaton}}
{
A graphical illustration of the action of a direct-sum finite automaton.
States $\phi_t$ represent the internal state of the automaton, as it processes a ``string'' of quantum vectors $\psi_0,\dots,\psi_{K-1}$ into $\psi'_0,\dots,\psi'_{K-1}$.
Unlike the usual quantum automata, the internal state of the automaton and the current ``letter'' of the ``string'' are joined via the direct sum.
Similarity with \rf{fig:pumping} is apparent.
It is an interesting question, whether such finite automata can find other applications.
}
{\newcommand{\OneIteration}[1]{
    \edef\indxx{#1}
    \begin{scope}[shift={(2.5*\indxx,0)}]
        \draw (0.5,3.5) rectangle (1.5, 4.5) node[pos=0.5] {\Large $S_{\indxx}$};
        \draw[\taucolor] (-1,4) to node[above]{$\phi_{\indxx}$} (0.5,4);
        \draw[red] (1, 5.3) node[above]{$\psi_{\indxx}$} to  (1, 4.5);
        \draw[blue] (1, 3.5) to (1, 2.7) node[below]{$\psi'_{\indxx}$};
    \end{scope}
}
\negbigskip
\[
\begin{tikzpicture}[every node/.style={font=\scriptsize}, every path/.append style={thick,->}]
\OneIteration{0}
\OneIteration{1}
\OneIteration{2}
\OneIteration{3}
\draw[\taucolor] (9,4) to node[above]{$\phi_{4}$} (10.5,4);
\end{tikzpicture}
\]
\negbigskip
}

It is trivial to implement the transformation in~\rf{eqn:automaton} in $\OO(K\log K)$ elementary operations besides the executions of~$S_t$.
It is a technical observation that we can remove the logarithmic factor.

\begin{lem}
\label{lem:automaton}
The transformation in~\rf{eqn:automaton} can be implemented in time $\OO(K) + \sum_t \TC(S_t)$, where $\TC$ is defined in \rf{sec:prelimCircuit}.
\end{lem}

\pfstart
The register $\reg K$ uses $\ell = \log K$ qubits.
We introduce an additional register $\cC$ that also consists of $\ell$ qubits.
For $i, t\in [K]$, we denote
\[
\ket C|i\curlyvee t> = \ket |1>^{\otimes c} \ket |0>^{\otimes \ell-c}
\]
if the binary representations of $i$ and $t$, considered as elements of $\bool^\ell$, agree on the first $c$ most significant bits, and disagree on the $(c+1)$-st one (or $c=\ell$).

We modify the algorithm so that before the $t$-th iteration of the loop, the algorithm is in a state of the form
\begin{equation}
\label{eqn:automatonModifiedState}
\ket K|t> \ket P|0> \ket C|t\curlyvee t> \ket H|\phi_t> 
+
\sum_{i=0}^{t-1}\ket K|i> \ket P|1> \ket C|i\curlyvee t> \ket H|\psi'_t>
+
\sum_{i=t}^{K-1}\ket K|i>\ket P|1>\ket C|i\curlyvee t> \ket H|\psi_t> 
\end{equation}
In particular, in the $\ket P|0>$-part, the register $\reg C$ contains $\ell$ ones.

The state~\rf{eqn:automatonModifiedState} with $t=0$ can be obtained from~\rf{eqn:automatonInitialState} in $\OO(\ell)$ elementary operations by computing $\cC$ starting from the highest qubit.
Execution of $S_t$ on Step~(a) can be conditioned on the lowest qubit of $\reg C$, which is equal to 1 if and only if $\cK$ contains $t$.

It remains to consider Step~(b) and the update of the register $\cC$ during the increment from $t$ to $t+1$.
Assume that $t+1$ is divisible by $2^d$, but not by $2^{d+1}$.
Then, Step~(b) takes $d+1$ controlled 1-qubit operations.
Similarly, the change from $t$ to $t+1$ in $\ket C|i\curlyvee t>$ in~\rf{eqn:automatonModifiedState} takes time $\OO(d)$ by first uncomputing the $d+1$ lowest qubits for $t$, and then computing them for $t+1$.
Finally, after the loop, the register $\reg C$ can be uncomputed in $\OO(\ell)$ operations.
Therefore, the total number of elementary operations is $\OO(K)$.
\pfend

\section{Transducers}

In this section, we define transducers and give their basic properties.
This is an initial treatment and will be extended in \rf{sec:canonical} to include query complexity. 

\subsection{Definition}
\label{sec:transducerDefinition}

Mathematically, our approach is based on the following construction.
\begin{thm}
\label{thm:transduce}
Let $\cH\oplus \cL$ be a direct sum of two vector spaces, and $S$ be a unitary on $\cH\oplus \cL$.
Then, for every $\xi\in\cH$, there exist $\tau\in\cH$ and $v\in\cL$ such that
\begin{equation}
\label{eqn:transduce}
S\colon \xi \oplus v \mapsto \tau \oplus v.
\end{equation}
Moreover,
\begin{enumerate}[(a)]\itemsep=0pt
\item The vector $\tau = \tau(S,\xi) = \tau_\cH(S,\xi)$ is uniquely defined by $\xi$ and $S$.
\item The vector $v = v(S,\xi) = v_\cH(S,\xi)$ is also uniquely defined if we require that it is orthogonal to the 1-eigenspace of $\Pi S\Pi$, where $\Pi$ is the projection on $\cL$.
\item The mapping $\xi\mapsto\tau$ is unitary and $\xi\mapsto v$ is linear.
\end{enumerate}
\end{thm}

We will prove the theorem at the end of this section.

In the setting of \rf{thm:transduce}, we will say that $S$ \emph{transduces}
$\xi$ into $\tau$, and write $\xi\transduce{S}\tau$.
The mapping $\xi\mapsto\tau$ on $\cH$ will be called the \emph{transduction action} of $S$ on $\cH$ and denoted by $S\DownTransduce_\cH$.

We call any $v$ satisfying $S(\xi\oplus v) = \tau\oplus v$ a \emph{catalyst} for $\xi\transduce{S} \tau$.
The condition in Point~(b) of $v$ to be orthogonal to the 1-eigenspace of $\Pi S\Pi$ is crucial for uniqueness, as adding such a vector to $v$ does not affect~\rf{eqn:transduce}.
Clearly, the vector $v$ as defined in Point~(b) has the smallest possible norm.
Therefore, we can define transduction complexity of $S$ on $\xi$ as $W(S,\xi) = \|v(S,\xi)\|^2$ for the latter $v$.
We write $W_\cH(S,\xi)$ if the space $\cH$ might not be clear from the context.
The above discussion can be formulated as the following claim.

\begin{clm}
\label{clm:anyV}
For any catalyst $v$ of the transduction $\xi\transduce{S} \tau$, we have $W(S,\xi) \le \|v\|^2$.
\end{clm}

On the other hand, checking orthogonality to $\Pi S\Pi$ is complicated and unnecessary, and we avoid doing it.
We usually couple a transducer with some chosen catalyst $v$ for all the $\xi$ of interest, which need not have the smallest possible norm.
In this case, we somewhat sloppily write $W(S,\xi) = \|v\|^2$ even for this catalyst $v$.
This agreement will become especially important when we add query complexity into the picture in \rf{sec:canonical}; see in particular the note towards the end of \rf{sec:canonicalDefinition}.

Other important notions related to the transducer are its time and query complexity.
Its time complexity, $T(S)$, is defined as its time complexity as an algorithm.
For $\xi\in\cH$, its query state, $q_\cH(S,O,\xi)$, and query complexity, $L_\cH(S,O,\xi)$, are defined as those of $S$ as an algorithm on the initial state  $\xi\oplus v$.
Note that for transducers with input oracles, we will adopt a special canonical form defined in \rf{sec:canonical}, until then we mostly ignore the oracle-related notions.

\begin{exm}
\label{exm:1}
Let us a give a simple concrete example illustrating the above notions.
Assume that $\cH$ is one-dimensional and spanned by $\ket|0>$, and $\cL$ is two-dimensional and spanned by $\ket|1>$ and $\ket|2>$.
Let $S$ be the reflection of the vector $\ket|0>-\ket|1>-\ket|2>$, so that its orthogonal complement stays intact.

The transduction action of $S$ on $\cH$ is the identity, which is certified by
\[
S\colon 
\ket |0> + \frac12 \ket|1> + \frac12 \ket |2> 
\mapsto
\ket |0> + \frac12 \ket|1> + \frac12 \ket |2> .
\]
Hence, we have 
$v(S,\ket|0>) = (\ket|1> + \ket|2>)/2$, 
and
$W(S,\ket|0>) = 1/2$.
This is not the only catalyst, as one can also take $v= \ket|1>$ or $v=\ket|2>$.
However, $(\ket|1> + \ket|2>)/2$ is the only catalyst orthogonal to the 1-eigenspace of $\Pi S\Pi$, which is spanned by $\ket |1> - \ket |2>$, and also has the smallest norm.
\end{exm}

\pfstart[Proof of \rf{thm:transduce}]
The vector $v$ can be found from the equation
\[
\Pi v = \Pi\sA[\tau+v] = \Pi\sA[S(\xi+v)] = \Pi S\xi + \Pi S v = \Pi S\xi + \Pi S \Pi v.
\]
From this we would like to argue that $v$ can be expressed as
\begin{equation}
\label{eqn:transduceV}
v = (\Pi - \Pi S\Pi)^+ \Pi S\xi,
\end{equation}
where $(\cdot)^+$ stands for the Moore-Penrose pseudoinverse.
Let us show that this is indeed the case. 

Denote by $\cK$ the kernel of $\Pi - \Pi S\Pi$ in $\cL$, and by $\cK^\perp$ its orthogonal complement in $\cL$. 
The subspace $\cK$ equals the 1-eigenspace of $\Pi S\Pi$.
But since $S$ is a unitary, a 1-eigenvector of $\Pi S\Pi$ is necessarily a 1-eigenvector of $S$.
Hence, $S$ is a direct sum of the identity on $\cK$ and a unitary on $\cH\oplus \cK^\perp$.
Thus, $\Pi - \Pi S\Pi$ is a direct sum of the zero operator in $\cH\oplus\cK$ and some operator in $\cK^\perp$.
Moreover, the latter operator is invertible in $\cK^\perp$ as its kernel is empty.
Since $\xi\in \cH$ is orthogonal to $\cK$, we get that $\Pi S\xi \in \cK^\perp$.
Hence,~\rf{eqn:transduceV} indeed uniquely specifies $v$.
This proves (b) and the second half of (c).

The uniqueness of $\tau$ and the linearity of $\xi\mapsto \tau$ now follow from~\rf{eqn:transduce} and the linearity of $S$.
Finally, unitarity of $S$ implies $\|\xi\| = \|\tau\|$, hence, the map $\xi\mapsto \tau$ is also unitary.
\pfend

\begin{prp}[Transitivity of Transduction]
\label{prp:transitivityOfTransduction}
Assume $\cH\subseteq \cH_1\subseteq \cH_2$ are vector spaces, and $S$ is a unitary in $\cH_2$.
Then,
\[
S\DownTransduce_{\cH} = (S\DownTransduce_{\cH_1})\DownTransduce_{\cH} ,
\]
and, for every $\xi\in \cH$, a possible catalyst is
\begin{equation}
\label{eqn:transitivityWitness}
v_\cH(S, \xi) = v_{\cH}(S\DownTransduce_{\cH_1}, \xi) + v_{\cH_1}\sA[S, \xi\oplus v_\cH(S\DownTransduce_{\cH_1}, \xi)].
\end{equation}
\end{prp}

\pfstart
By definition,
\[
S\DownTransduce_{\cH_1}\colon 
\xi \oplus v_{\cH}(S\DownTransduce_{\cH_1},\xi) 
\mapsto 
\tau \oplus v_{\cH}(S\DownTransduce_{\cH_1},\xi) 
\]
for some $\tau\in \cH$.  The latter means that
\[
S\colon 
\xi \oplus v_{\cH}(S\DownTransduce_{\cH_1},\xi) \oplus v_{\cH_1}\sA[S, \xi \oplus v_{\cH}(S\DownTransduce_{\cH_1},\xi)]
\mapsto 
\tau \oplus v_{\cH}(S\DownTransduce_{\cH_1},\xi) \oplus v_{\cH_1}\sA[S, \xi \oplus v_{\cH}(S\DownTransduce_{\cH_1},\xi)],
\]
proving~\rf{eqn:transitivityWitness}.
\pfend

\subsection{Implementation}
\label{sec:implementation}
The key point we will now make 
is that given a transducer $S$, there exists a very simple quantum algorithm that approximately implements its transduction action on $\cH$. 
This algorithm is a generalisation of the one from~\cite{belovs:LasVegas}, which was used for implementation of the adversary bound.

Before we describe this algorithm, let us establish a few conventions.
We call $\cH$ the \emph{public} and $\cL$ the \emph{private} space of $S$.
We indicate this separation of $\cH$ and $\cL$ by a privacy qubit $\cP$.
The value 0 of $\cP$ will indicate the public space $\cH$,
and the value 1 the private space $\cL$.
This means that both $\cH$ and $\cL$ are embedded into the same register during implementation.
Thus,
\begin{equation}
\label{eqn:xi+v}
\xi\oplus v = \ket P |0> \ket H |\xi> + \ket P|1> \ket L |v>
\end{equation}
explicitly specifying the public and the private spaces.
We will extend this notation in \rf{sec:canonicalDefinition}.

As it can be understood from the name, the algorithm does not have direct access to the private space $\cL$ of the transducer.  All the interaction between the transducer and its surrounding is through the public space $\cH$.

\begin{thm}
\label{thm:pumping}
Let spaces $\cH$, $\cL$, and a positive integer $K$ be fixed.
There exists a quantum algorithm that transforms $\xi$ into $\tau'$ such that
\[
\|\tau' - \tau(S,\xi)\| \le 2 \sqrt{\frac{W(S,\xi)}{K}}
\]
for every transducer $S\colon \cH\oplus \cL\to\cH\oplus \cL$ and initial state $\xi\in\cH$.
The algorithm conditionally executes $S$ as a black box $K$ times, and uses $\OO(K)$ other elementary operations.
\end{thm}

\rf{thm:introImplementation} is a direct corollary of \rf{thm:pumping}.
A sketch of the proof of \rf{thm:pumping} was already given in the same section, see, in particular, \rf{fig:pumping}.

\pfstart[Proof of \rf{thm:pumping}]
The space of the algorithm is $\cK\otimes (\cH\oplus \cL)$, where $\cK$ is a $K$-qudit.
The register $\cH\oplus \cL$ contains the privacy qubit $\cP$ as described above.	
The algorithm starts in the state $\xi = \ket K|0>\ket P |0> \ket H|\xi>$, and 
performs the following transformations:
\begin{enumerate}\itemsep=0pt
\item Map $\ket K|0>$ into the uniform superposition $\frac1{\sqrt K} \sum_{t=0}^{K-1} \ket K|t>$.
\item For $t=0,1,\dots,K-1$:
\begin{enumerate}\itemsep=0pt
\item Execute $S$ on
$\cH\oplus \cL$
conditioned on $\ket K|t>$.
\item Conditioned on $\ket P|1>$, replace $\ket K|t>$ by $\ket K|t+1>$ (where $\ket K|K>$ is equal to $\ket K|0>$).
\end{enumerate}
\item Run Step 1 in reverse.
\end{enumerate}

Clearly, the algorithm conditionally executes $S$ exactly $K$ times.
As described now, the algorithm takes time $\OO(K\log K)$, but it is implementable in time $\OO(K)$ using \rf{lem:automaton}.

Let us prove correctness.
We write $v=v(S,\xi)$ and $\tau=\tau(S,\xi)$.
After Step~1, the algorithm is in the state
\begin{equation}
\label{eqn:pumpingOriginal}
\frac1{\sqrt K} \sum_{i=0}^{K-1} \ket K|i> \ket P |0> \ket H|\xi>.
\end{equation}
We perform a perturbation in the sense of~\rf{lem:surgery} and assume the algorithm is instead in the state
\begin{equation}
\label{eqn:pumpingModified}
\frac1{\sqrt K} \sum_{i=0}^{K-1} \ket K|i> \ket P |0> \ket H|\xi> + \frac1{\sqrt K} \ket K|0>\ket P|1>\ket L|v>.
\end{equation}
On the $t$-th iteration of the loop, the transducer $S$ on Step 2(a) transforms the part of the state
\begin{equation}
\label{eqn:pumpingOneStep}
\frac1{\sqrt K} \ket K|t> \ket P |0> \ket H|\xi> + \frac1{\sqrt K} \ket K|t>\ket P|1>\ket L|v>
\;\longmapsto\;
\frac1{\sqrt K} \ket K|t> \ket P |0> \ket H|\tau> + \frac1{\sqrt K} \ket K|t>\ket P|1>\ket L|v>,
\end{equation}
and on Step 2(b) the following transformation of the part of the state is performed: 
\[
\frac1{\sqrt K} \ket K|t>\ket P|1>\ket L|v> 
\;\longmapsto\;
\frac1{\sqrt K} \ket K|t+1>\ket P|1>\ket L|v>.
\]
Therefore, after the execution of the loop in Step 2, we get the state
\begin{equation}
\label{eqn:pumpingFinal}
\frac1{\sqrt K} \sum_{i=0}^{K-1} \ket K|i> \ket P |0> \ket H|\tau> + \frac1{\sqrt K} \ket K|0>\ket P|1>\ket L|v>.
\end{equation}
We perturb the state to
\begin{equation}
\label{eqn:pumpingTarget}
\frac1{\sqrt K} \sum_{i=0}^{K-1} \ket K|i> \ket P |0> \ket H|\tau>.
\end{equation}
After Step 3, we get the state $\tau = \ket K|0>\ket P |0> \ket H|\tau>$.

Note that the difference between the states in~\rf{eqn:pumpingOriginal} and~\rf{eqn:pumpingModified} has norm $\|v\|/\sqrt K$.
The same is true for the difference between the states in~\rf{eqn:pumpingFinal} and~\rf{eqn:pumpingTarget}.
Therefore, by \rf{lem:surgery}, the actual final state $\tau'$ of the algorithm satisfies
\[
\|\tau' - \tau\| \le 2 \frac{\|v\|}{\sqrt K}
\]
as required.
\pfend

\section{Example I: Quantum Walks}
\label{sec:walks}

In this section, we implement the electric quantum walk from~\cite{belovs:electicityQuantumWalks} using the construction outlined in~\rf{sec:introWalks}.
See also a subsequent paper~\cite{belovs:phaseHelps}, where similar ideas are applied to search and to the Welded Tree problem~\cite{childs:walkExponentialSeparation}.

\myfigure{\label{fig:walk}}
{
An example of the extension of a graph for a quantum walk.
The original graph contains two parts $A$ and $B$ of 4 and 3 vertices, respectively.
The initial probability distribution is supported on two vertices $\{u_1,u_2\}\subseteq A$.
The original edges $E$ of the graph are black, the new ones $E'$ are red.
One marked vertex in $B$ is coloured blue.
}
{
\negbigskip
\def\inputcolor{red}
\[
\begin{tikzpicture}[minimum size=15pt, inner sep=0pt]
    \node[blue] at (0,-0.8) {\Large $A$};
    \node[circle, draw] (A4) at (0,0) {};
    \node[circle, draw] (A3) at (0,1) {};
    \node[circle, draw, fill=\inputcolor!50] (A2) at (0,2) {$u_2$};
    \node[circle, draw, fill=\inputcolor!50] (A1) at (0,3) {$u_1$};
    \node[blue] at (2,-0.8) {\Large $B$};
    \node[circle, draw] (B3) at (2,0.5) {};
    \node[circle, draw] (B2) at (2,1.5) {};
    \node[circle, draw] (B1) at (2,2.5) {};
    \graph{
        (A1)--{(B1),(B3)};
        (A2)--{(B2),(B3)};
        (A3)--{(B1),(B2),(B3)};
        (A4)--{(B1),(B2)};
    };
    \draw (A1) to node[above]{$w_e$} (B1);
    \node[circle, draw, double, fill=blue!50] at (B2) {};
    \node[circle, draw, fill=\inputcolor!50] (S2) at (-1.5,2) {$u_2'$};
    \node[circle, draw, fill=\inputcolor!50] (S1) at (-1.5,3) {$u_1'$};
    \draw[-,\inputcolor] (S1) to node[above] {$\sigma_1$} (A1);
    \draw[-,\inputcolor] (S2) to node[above] {$\sigma_2$} (A2);
\end{tikzpicture}
\]
\negbigskip
}

A quantum walk is described by a bipartite graph $G$, whose parts we denote by $A$ and $B$.
Let $E$ be the set of edges of $G$.
Each edge $e$ of the graph is given a non-negative real \emph{weight} $w_{e}$.
We have some set $M\subseteq A\cup B$ of \emph{marked} vertices.
There is a subroutine \textsf{Check} that, for every vertex $u$, says whether it is marked.
The goal of the quantum walk is to detect whether $M$ is empty or not.

In the framework of electric quantum walks, the graph is extended as follows, see \rf{fig:walk}.
The quantum walk is tailored towards a specific initial probability distribution $\sigma$ on $A$.
Let $A_\sigma\subseteq A$ be the support of $\sigma$.
For each $u\in A_\sigma$, we add a new vertex $u'$ and a new dangling edge $u'u$ to the graph.
The newly added vertices are \emph{not} contained in $B$.
Let $E'$ be the set of newly added edges.
For edges if $E'$ we assume the weight $w_{u'u} = \sigma_u$.

We treat this construction as a transducer.
The private space $\cL$ of the quantum walk is $\bC^E$.
The public space $\cH$ is $\mathbb{C}^{E'}$.
The initial state $\xi$ is given by 
\[
\xi = \sum_{u\in A_\sigma} \sqrt{\sigma_u} \ket|u'u>\in \cH.
\]
For a vertex $u$, let $\cL_u$ denote the space spanned by all the edges incident to $u$ (including the ones in $E'$).
Define
\begin{equation}
\label{eqn:walkPsi}
\psi_u = \sum_{e: e\sim u} \sqrt{w_e}\ket|e>\in \cL_u
\end{equation}
where the sum is over all the edges incident to $u$.

For $U\subseteq A$ or $U\subseteq B$, let $R_U$ denote the reflection of all $\psi_u$ for $u\in U$, i.e., $R_U$ acts as negation on the span of all these $\psi_u$ and as identity on its orthogonal complement.
We define the transducer, which depends on the set of marked vertices $M$, as 
\[
S_M = R_{B\setminus M}R_{A\setminus M}.
\]
Each of $R_{A\setminus M}$ and $R_{B\setminus M}$ is decomposable into products of local reflections in $\cL_u$ as $u$ ranges over $A$ and $B$, respectively.
The corresponding local reflections are either identities for $u\in M$, or reflections of $\psi_u$ for $u\notin M$.
The implementation of $S_M$ can be done using local reflections controlled by the \textsf{Check} subroutine.

Let
\[
W = \sum_{e\in E} w_e
\]
be the total weight of the graph.
For $M\ne\emptyset$, let $R_{\sigma, M}$ be the minimum of
\begin{equation}
\label{eqn:resistance}
\sum_{e\in E} \frac{p_e^2}{w_e},
\end{equation}
over all flows $p=(p_e)$ on the graph where $\sigma_u$ units of flow are injected in $u\in A_\sigma$, and the flow is collected at the vertices in $M$.
The minimum is attained by the electrical flow, and $R_{\sigma, M}$ is the corresponding effective resistance.

\begin{thm}
\label{thm:walk}
The transducer $S_M$ defined above transduces $\xi\transduce{} - \xi$ if $M=\emptyset$, and $\xi\transduce{} \xi$ otherwise.
Its transduction complexity is
\begin{equation}
\label{eqn:walkComplexity}
W\sA[S_\emptyset, \xi] = W
\qqand
W\sA[S_M, \xi] = R_{\sigma, M}
\end{equation}
respectively.
\end{thm}

Thus, the transduction action of the quantum walk encodes the answer to the detection problem in the phase, and this is done exactly.
Let $R$ denote the maximal effective resistance over all possible choices of $M\ne\emptyset$.
We can rescale all $w_i$ by the same factor so that the maximal transduction complexity in~\rf{eqn:walkComplexity} becomes equal to $\sqrt{RW}$.
By \rf{thm:introImplementation}, the presence of marked elements can be detected, with bounded error, in $O(\sqrt{RW})$ executions of $S_M$.
This coincides with the complexity estimate established in~\cite{belovs:electicityQuantumWalks}.

\pfstart[Proof of \rf{thm:walk}]
Let us start with the case $M=\emptyset$.
The initial coupling is
\[
\xi\oplus v_\emptyset = \sum_{e\in E\cup E'} \sqrt{w_e} \ket |e>.
\]
We have $\xi\oplus v_\emptyset = \sum_{u\in A} \psi_u$ and $v_\emptyset = \sum_{u\in B} \psi_u$.
Hence, $R_A$ reflects all $\xi\oplus v_\emptyset$, and $R_B$ only reflects $v_\emptyset$.
This gives us the following chain of transformations
\[
\xi \oplus v_\emptyset \maps{R_A} -\xi \oplus -v_\emptyset \maps{R_B} -\xi \oplus v_\emptyset,
\]
giving $\xi \transduce{S_\emptyset} -\xi$.
Note that this chain of transformation adheres to~\rf{eqn:QW_sequence_positive}.

Now assume $M\ne\emptyset$.
Let $p_e$ be a flow on the graph where $\sigma_u$ units of flow are injected in $u'$, and the flow is collected at $M$.
To solve the sign ambiguity, we assume that all the edges are oriented towards $A$.
This time, we define the catalyst $v_M$ so that
\begin{equation}
\label{eqn:walkPositiveWitness}
\xi\oplus v_M = \sum_{e\in E\cup E'} \frac {p_e}{\sqrt{w_e}} \ket |e>.
\end{equation}
Recall that $p_{u'u} = w_{u'u} = \sigma_u$, hence the above equation is satisfied in $\cH$.

The projection of~\rf{eqn:walkPositiveWitness} onto $\cL_u$ is given by
\begin{equation}
\label{eqn:walkPu}
\sum_{e: e\sim u} \frac {p_e}{\sqrt{w_e}}\ket|e>.
\end{equation}
This state is not changed by the corresponding local reflection in $\cL_u$.
Indeed if $u\in M$, the corresponding local reflection is the identity.
If $u\notin M$, then, by the flow condition, the state in~\rf{eqn:walkPu} is orthogonal to $\psi_u$ in~\rf{eqn:walkPsi}.
Thus, neither $R_{A\setminus M}$ nor $R_{B\setminus M}$ change $\xi\oplus v_M$, hence, $\xi\transduce{S_M} \xi$.
Note that in this case we adhere to~\rf{eqn:QW_sequence_negative}.

The corresponding transduction complexities are as given in~\rf{eqn:walkComplexity}.
The catalyst in the second case uses the flow $p_e$ through the graph, which is not unique.
We choose the minimal one as per \rf{clm:anyV}.
\pfend

\section{Canonical Transducers}
\label{sec:canonical}
In this section, we define a specific form of transducers we will be using in this paper.
The main point is in the application of the input oracle.
Inspired by the construction in~\cite{belovs:LasVegas}, we assume that the transducer first executes the input oracle, and then performs some input-independent unitary.
Moreover, the input oracle is always applied to the private space of the transducer.
These assumptions simplify many constructions, and any transducer can be transformed into the canonical form with a small overhead as shown later in \rf{prp:canoning}.

\subsection{Definition}
\label{sec:canonicalDefinition}

Concerning the input oracle, the assumptions are similar to those in \rf{sec:prelimQuery}.
The input oracle is a unitary $O$ in some space $\cM$, and we have unidirectional access to $O$.
The oracle only acts on the local space $\cL$ of the transducer.
Let us decompose the latter in two parts $\cL = \cLw\oplus \cLq$, which stand for the work (non-query) and query parts.
We also have $\cLq = \cLt \otimes\cM$ for some space $\cLt$.
We denote the identity on $\cLt$ simply by $I$.

A canonical transducer $S = S(O)$ performs the following transformations, see also \rf{fig:canonical}:
\begin{itemize}\itemsep=0pt
\item It executes the input oracle $I\otimes O$ on $\cLq$.
Similarly to~\rf{eqn:query}, we call it a query and denote it by 
\begin{equation}
\label{eqn:canonicalQuery}
\tO = I_\cH \oplus \Iw \oplus I\otimes O,
\end{equation}
where $I_\cH$ and $\Iw$ are identities on $\cH$ and $\cLw$, respectively.
\item It performs an input-independent work unitary $\Sw$ on $\cH\oplus \cL$.
\end{itemize}

The decomposition $\cL = \cLw\oplus\cLq$ yields the decomposition $v = \vw\oplus\vq$ of the catalyst.
Thus, the action $S(O)$ of the transducer $S$ on the input oracle $O$ is given by the following chain of transformations:
\begin{equation}
\label{eqn:canonicalForm}
S(O)\colon \xi \oplus \vw \oplus \vq
\maps{\tO}
\xi \oplus \vw \oplus(I\otimes O)\vq 
\maps{\Sw}
\tau \oplus \vw \oplus \vq .
\end{equation}

Now we can make the following complexity-related definitions.
The catalyst is $v = v(S, O,\xi)$.%
\footnote{It is the same as $v(S(O),\xi)$ in the previous notation, however, we opted to $v(S,O,\xi)$ to reduce the number of brackets and to keep notation synchronised with~\cite{belovs:LasVegas}.}
The transduction complexity is
\begin{equation}
\label{eqn:transductionComplexity}
W(S, O, \xi) = \normA|v(S, O, \xi)|^2.
\end{equation}
The query state is $q(S, O, \xi) = \vq =\Pi^{\bullet}v(S,O,\xi)$, where $\Pi^{\bullet}$ denotes the orthogonal projector onto $\cLq$.
The (Las Vegas) query complexity of the transducer is
\begin{equation}
\label{eqn:transducerQueryComplexity}
L(S, O, \xi) = \normA|q(S, O, \xi)|^2.
\end{equation}
Note that formally the definitions $q(S, O,\xi)$ and $L(S,O,\xi)$ are in conflict with the same definitions~\rf{eqn:totalQueryState} and~\rf{eqn:LasVegasComplexity} if $S$ is considered as a \emph{program} and not as a transducer.
However, this should not cause a confusion.
If the space $\cH$ is not clear from the context, we will add it as a subscript as in \rf{sec:transducerDefinition}.

Time complexity $T(S)$ of the transducer is the number of elementary operations required to implement the unitary $\Sw$.
Note that we do not count the query towards time complexity of the transducer.

Finally, the definitions~\rf{eqn:transductionComplexity} and~\rf{eqn:transducerQueryComplexity} can be extended to $\xi'\in \cE\otimes \cH$ using~\rf{eqn:transductionExtended}.

\paragraph{Note on Non-Uniqueness of Catalyst}
It is important to note that in this setting the non-uniqueness of the  catalyst $v$ discussed in \rf{sec:transducerDefinition} becomes very important.
To understand why, consider again \rf{exm:1} from that section.
This time, assume that $\cLw$ is spanned by $\ket |1>$ and $\cLq$ by $\ket |2>$, the input oracle is $O=I$, and $\Sw = S$ as defined previously.

The ``right''  catalyst $v = (\ket|1> + \ket |2>)/2$ for the initial state $\xi=\ket|0>$ suggests that $W(S, O, \ket|0>) = 1/2$ and $L(S, O, \ket|0>) = 1/4$.
However, if we take $v=\ket|1>$, we get that $W(S,O,\ket |0>) = 1$ and $L(S, O, \ket |0>) = 0$.
Therefore, there is no longer a single catalyst that minimises both the transduction and the query complexity.
This is similar to usual algorithms, where time and query complexity can be minimised by different algorithms.

We solve this complication by implicitly assigning a specific catalyst $v(S, O, \xi)$ for every $O$ and $\xi$ of interest, that gives \emph{both} $W(S,O,\xi)$ and $L(S, O,\xi)$ simultaneously.
Of course, neither of the two are guaranteed to be minimal.
It is possible to study the trade-off between the transduction and the query complexity for a fixed $O$ and $\xi$, but we will not explicitly pursue that in this paper.

\mycutecommand{\xiw}{\xi^\circ}
\mycutecommand{\xiq}{\xi^\bullet}

\paragraph{Implementation Details}
In terms of registers, as in \rf{sec:prelimQuery}, the separation $\cL = \cLw\oplus \cLq$ is indicated by the qubit $\cR$.
Now it makes sense to assume that the registers $\cH$, $\cLw$ and $\cLq$ are the same, the distinction being given by the values of the registers $\cP$ and $\cR$.
We will usually place this common register as the unnamed last register in our expressions.
In particular,
\begin{equation}
\label{eqn:canonicalDecomposition}
\xi \oplus v
=
\ket P|0> \ket R |0> \ket |\xi> + \ket P|1> \ket R|0> \ket |\vw> + \ket P|1> \ket R|1> \ket |\vq>.
\end{equation}
Note that the space $\cH$ is indicated by $\ket P|0> \ket R |0>$ meaning that it is not acted on by the oracle.
This allows us to implement the query as an application of $O$ controlled by $\ket R|1>$,
which is in accord with the convention established in \rf{sec:prelimQuery}.

We will also use registers $\cH$ and $\cL$ in the sense of \rf{sec:implementation}, that is, containing $\cR$.
In particular, we can write the action of a canonical transducer~\rf{eqn:canonicalForm} in registers like
\begin{equation}
\label{eqn:canonicalSequence}
S(O)\colon \ket P |0> \ket H |\xi> + \ket P|1> \ket L |v>
\maps{\tO}
\ket P |0> \ket H |\xi> + \ket P|1> \ketA L |\tO v>
\maps{\Sw}
\ket P |0> \ket H |\tau> + \ket P|1> \ket L |v>,
\end{equation}
where we used shorthand
\begin{equation}
\label{eqn:Ov}
\tO v = \ket R|0> \ket |\vw> + \ket R|1> \ketA |(I\otimes O)\vq>.
\end{equation}

\subsection{Multiple Input Oracles}
\label{sec:multipleInputOracles}

Following~\cite{belovs:LasVegas}, we can also allow multiple input oracles joined by direct sum:
\begin{equation}
\label{eqn:multipleOracles}
O = O^{(1)}\oplus O^{(2)}\oplus \cdots \oplus O^{(r)},
\end{equation}
where the $i$-th input oracle $O^{(i)}$ acts in space $\cM^{(i)}$ and $\cM = \cM^{(1)}\oplus\cdots\oplus \cM^{(r)}$.
We get the corresponding decomposition of the query state:
\begin{equation}
\label{eqn:multipleOraclesQeuryState}
q(S, O, \xi) = q^{(1)}(S, O, \xi) \oplus \cdots \oplus q^{(r)}(S, O, \xi),
\end{equation}
where $q^{(i)}(S,O,\xi)$ is the \emph{partial query state} of the $i$-th input oracle.
This also gives query complexities of the individual oracles:
\[
L^{(i)}(S, O, \xi) = \normA|q^{(i)}(S, O, \xi)|^2.
\]

It makes sense to tweak the notation assumed earlier in \rf{sec:canonicalDefinition} to make it in line with \rf{sec:prelimMultipleOracles}.
We assume the register $\cR$ can hold an integer from 0 to $r$, where $\ket R|i>$ with $i>0$ indicates the space of the $i$-th input oracle.
Thus, in place of~\rf{eqn:canonicalDecomposition}, we have
\begin{equation}
\label{eqn:multipleOraclesWitness}
\xi \oplus v = 
\ket P|0> \ket R|0> \ket |\xi> +
\ket P|1> \ket R|0> \ket |\vw> +
\ket P|1>\sum_{i=1}^r \ket R|i> \ketA |v^{(i)}>
\end{equation}
with $v^{(i)} = q^{(i)}(S, O, \xi)$.
To apply the $i$-th input oracle $O^{(i)}$, it suffices to condition it on $\ket R|i>$.
The action of a canonical transducer stays given by~\rf{eqn:canonicalSequence}, where, this time,
\begin{equation}
\label{eqn:OvMultipleOracles}
\tO v = \ket R|0> \ket |\vw> + \sum_{i=1}^r \ket R|i> \ket |\sA[I\otimes O^{(i)}]v^{(i)}>.
\end{equation}

For notational convenience, we will assume a single input oracle in most of the paper.
The case of multiple oracles can be obtained using the decomposition of the query state in~\rf{eqn:multipleOraclesQeuryState}, which contains all the necessary information.

\subsection{Reducing the Number of Oracle Calls}
\label{sec:reducingOracle}
One problem with the algorithm in \rf{thm:pumping} is that, when applied to a canonical transducer, the input oracle $O$ is executed the same number of times as the work unitary $\Sw$.
This is suboptimal as the transduction complexity can be much larger than the query complexity.
In this section, we describe a query-efficient implementation, which can also handle multiple input oracles.

Let $S$ be a canonical transducer with $r$ input oracles joined into one oracle $O$ via direct sum as in~\rf{eqn:multipleOracles}.
In the following theorem, we assume the spaces $\cH$, $\cLw$, $\cLt$, $\cM^{(1)},\dots,\cM^{(r)}$ are fixed, while the operators $\Sw$ and $O^{(1)},\dots,O^{(r)}$ can vary.

\begin{thm}[Query-Optimal Implementation of Transducers]
\label{thm:optimalImplementation}
Let $K \ge K^{(1)},\dots, K^{(r)}$ be positive integers, which we assume to be powers of 2 for simplicity.
There exists a quantum algorithm that conditionally executes $\Sw$ as a black box $K$ times, makes $K^{(i)}$ queries to the $i$-th input oracle $O^{(i)}$, and uses $\OO(K+K^{(1)}+\cdots+K^{(r)})\log r$ other elementary operations.
For each $\Sw$, $O^{(i)}$, and initial state $\xi$, the algorithm transforms $\xi$ into $\tau'$ such that
\begin{equation}
\label{eqn:optimalImplementationEstimate}
\normA|\tau' - \tau(S, O, \xi)| \le \frac 2{\sqrt K} \sqrt{W(S, O, \xi) + \sum_{i=1}^r\s[\frac{K}{K^{(i)}} -1] L^{(i)}(S, O, \xi) }.
\end{equation}
\end{thm}

\rf{thm:pumping} is a special case of this theorem with all $K^{(i)}$ equal to $K$.
Observe that the number of elementary operations is equal to the total number of invocations of $\Sw$ and $O^{(i)}$ times $\log r$.  It is highly unlikely that this part of the algorithm would dominate its time complexity.

\pfstart[Proof of \rf{thm:optimalImplementation}]
The proof is an extension of that of \rf{thm:pumping}.
Its outline was already given in \rf{sec:introCanonical}; see in particular \rf{fig:implementationBetter}.

Let us define $D^{(i)} = K/K^{(i)}$, which is a power of 2.
The space of the algorithm is $\cK\otimes (\cH\oplus \cL)$, where $\cK$ is a $K$-qudit.
The register $\cH\oplus \cL$ contains the registers $\cP$ and $\cR$ as described above.
The algorithm starts its work in the state $\xi = \ket K|0>\ket P |0> \ket R|0> \ket|\xi>$.
Its steps are as follows.

\begin{enumerate}\itemsep=0pt
\item Map $\ket K|0>$ into the uniform superposition $\frac1{\sqrt K} \sum_{t=0}^{K-1} \ket K|t>$.
\item For $t=0,1,\dots, K-1$:
\begin{enumerate}\itemsep=0pt
\item For each $i=1,\dots, r$:
\begin{itemize}
\item if $t$ is divisible by $D^{(i)}$, execute the input oracle $O^{(i)}$ conditioned on $\ket R|i>$.
\end{itemize}
\item Execute the work unitary $\Sw$ on $\cH\oplus \cL$ conditioned on $\ket K|t>$.
\item Conditioned on $\ket P|1>\ket R|0>$, replace $\ket K|t>$ by $\ket K|t+1>$.
\item For each $i=1,\dots, r$:
\begin{itemize}
\item if $t+1$ is divisible by $D^{(i)}$, add $D^{(i)}$ to $\reg K$ conditioned on $\ket R|i>$.  The operation is performed modulo $K$.
\end{itemize}
\end{enumerate}
\item Run Step 1 in reverse.
\end{enumerate}

The analysis is similar to that in the proof of \rf{thm:pumping}.
Now we assume that after Step~1, instead of the state 
$\frac1{\sqrt K} \sum_{t=0}^{K-1} \ket K|t> \ket P |0> \ket R|0> \ket|\xi>$,
we are in the state
\begin{equation}
\label{eqn:optimalInitial}
\frac1{\sqrt K} \sum_{t=0}^{K-1} \ket K|t> \ket P |0> \ket R|0>\ket |\xi> + 
\frac1{\sqrt K} \ket P|1> \skB[\ket K|0>\ket R|0> \ket |\vw> + \sum_{i=1}^r \sum_{t=0}^{D^{(i)}-1} \ket K|t>\ket R|i>\ketA |v^{(i)}>].
\end{equation}
Since $t=0$ is divisible by all $D^{(i)}$, after Step~2(a) of the first iteration of the loop, we have the state
\[
\frac1{\sqrt K} \sum_{t=0}^{K-1} \ket K|t> \ket P |0> \ket R|0>\ket |\xi> + \frac1{\sqrt K} \ket P|1> \skB[\ket K|0>\ket R|0> \ket |\vw> + \sum_{i=1}^r \sum_{t=0}^{D^{(i)}-1} \ket K|t>\ket R|i>\ketA |(I\otimes O^{(i)})v^{(i)}>].
\]
The crucial observation is that on each iteration of the loop on step~2(b), the following transformation is performed.
The part of the state
\begin{equation}
\label{eqn:optimal1}
\frac1{\sqrt K} \ket K|t> \skB[\ket P |0> \ket R|0>\ket |\xi> + \ket P|1>\ket R|0> \ket |\vw> + \ket P|1>\sum_{i=1}^r \ket R|i>\ketA |(I\otimes O^{(i)})v^{(i)}>]
\end{equation}
gets mapped by $\Sw$ into
\begin{equation}
\label{eqn:optimal2}
\frac1{\sqrt K} \ket K|t> \skB[\ket P |0> \ket R|0>\ket |\tau> + \ket P|1>\ket R|0> \ket |\vw> + \ket P|1>\sum_{i=1}^r \ket R|i>\ketA |v^{(i)}>],
\end{equation}
where we used~\rf{eqn:canonicalSequence} with~\rf{eqn:OvMultipleOracles}.

If $t=cD^{(i)}-1$ for some integer $c$, then on Step~2(d), we perform the transformation of the part of the state
\begin{equation}
\label{eqn:improvedPumping1}
\frac1{\sqrt K} \ket P|1> \ket R|i> \sum_{t=0}^{D^{(i)}-1} \ket K|(c-1)D^{(i)}+t>\ketA |v^{(i)}>
\maps{}
\frac1{\sqrt K} \ket P|1> \ket R|i>\sum_{t=0}^{D^{(i)}-1} \ket K|cD^{(i)} + t>\ketA|v^{(i)}>,
\end{equation}
which is then mapped on Step~2(a) of the next iteration into
\[
\frac1{\sqrt K} \ket P|1> \ket R|i>\sum_{t=0}^{D^{(i)}-1} \ket K|cD^{(i)} + t>\ket|(I\otimes O^{(i)})v^{(i)}>.
\]
Therefore, after all the $K$ iterations of the loop in Step 2, we result in the state
\[
\frac1{\sqrt K} \sum_{t=0}^{K-1} \ket K|t> \ket P |0> \ket R|0>\ket|\tau> + 
\frac1{\sqrt K} \ket P|1> \skB[\ket K|0>\ket R|0> \ket |\vw> + \sum_{i=1}^r \sum_{t=0}^{D^{(i)}-1} \ket K|t>\ket R|i>\ketA |v^{(i)}>].
\]
After that, we finish as in \rf{thm:pumping}, by assuming we are in the state 
$\frac1{\sqrt K} \sum_{t=0}^{K-1} \ket K|t> \ket P |0> \ket R|0> \ket |\tau>$
instead, and applying Step 3.

The total perturbation of the algorithm is
\begin{equation}
\label{eqn:improvedPumpingTotalPerturbation}
\frac2{\sqrt K} 
\normB|{\ket K|0>\ket R|0> \ket |\vw> + \sum_{i=1}^r \sum_{t=0}^{D^{(i)}-1} \ket K|t>\ket R|i>\ketA |v^{(i)}>}|,
\end{equation}
which is equal to the right-hand side of~\rf{eqn:optimalImplementationEstimate}.

The claim on the number of executions of $\Sw$ and $O^{(i)}$ in the algorithm is obvious.
Besides that, it is trivial to implement the algorithm in 
$\OO(K+K^{(1)}+\cdots+K^{(r)})\log K \log r$ elementary operations,
where the $\log r$ factors comes from the necessity to index one of the $r$ input oracles.
Using a slight modification of \rf{lem:automaton}, the algorithm can be implemented in $\OO(K+K^{(1)}+\cdots+K^{(r)})\log r$ elementary operations.
One crucial point here is that when $t = cD^{(i)}-1$, addition of $D^{(i)}$ on Step~2(d) is equivalent to the replacement of $c-1$ by $c$ in the highest qubits of the register $\cK$ as indicated by~\rf{eqn:improvedPumping1}.
We omit the details.
\pfend

The following corollary, which is equivalent to \rf{thm:introImplementationBetter}, is easier to apply.

\begin{cor}
\label{cor:optimalImplementation}
Assume $r=\OO(1)$, and let $\eps, W, L^{(1)},\dots,L^{(r)}>0$ be parameters.
There exists a quantum algorithm that conditionally executes $\Sw$ as a black box $K=\OO(1+W/\eps^2)$ times, 
makes $\OO(L^{(i)}/\eps^2)$ queries to the $i$-th input oracle $O^{(i)}$, and uses $\OO(K)$ other elementary operations.
The algorithm $\eps$-approximately transforms $\xi$ into $\tau(S, O, \xi)$ for all $S$, $O^{(i)}$, and $\xi$ such that $W(S, O, \xi)\le W$ and $L^{(i)} (S, O, \xi) \le L^{(i)}$ for all $i$.
\end{cor}

\begin{proof}[Proof of \rf{cor:optimalImplementation}]
We first prove a relaxed version of the corollary, where we allow each input oracle to be called $\OO\sA[1 + L^{(i)}/\eps^2]$ times.
Afterwards, we show how to remove this assumption.
We allow error $\eps/2$ in this relaxed version.

If $W<\eps^2/16$, the (relaxed) corollary follows from \rf{thm:pumping}, as the algorithm then executes $\Sw$ and each input oracle $K = 1$ times.
Therefore, we will assume $W\ge \eps^2/16$.
Also, by definition we have that $W(S,O,\xi)\ge L^{(i)}(S,O,\xi)$.
Therefore, we may assume that $W\ge L^{(i)}$, reducing $L^{(i)}$ otherwise.

We intend to use \rf{thm:optimalImplementation}.
We take $K$ as the smallest power of 2 exceeding $32(r+1)W/\eps^2$.
In particular, $K = \Theta(W/\eps^2)$ by our assumption on $W$ and $r=\OO(1)$.
We take $K^{(i)}$ as the largest power of 2 that does not exceed $\max\sfigA{1,\, K L^{(i)}/W}$.
It satisfies $K^{(i)} \le K$ as required.
On the other hand, $K^{(i)}\ge KL^{(i)}/(2W)$ implying $K/K^{(i)} \le 2W/L^{(i)}$, which gives the error estimate~\rf{eqn:optimalImplementationEstimate} at most
\[
\frac2{\sqrt K} \sqrt{2(r+1)W} 
\le \frac{2\sqrt{2(r+1)W}}{\sqrt{32(r+1)W/\eps^2}}
\le \eps/2.
\]
If $K^{(i)}>1$, we get $K^{(i)} \le K L^{(i)}/W = \OO(L^{(i)}/\eps^2)$, which finishes the proof of the relaxed statement of the corollary.

In the following, we assume we used \rf{thm:optimalImplementation} in the proof, the case of \rf{thm:pumping} being similar.
Consider all the values of $i$ such that $K^{(i)}=1$.
By the proof of \rf{thm:optimalImplementation}, the input oracles are applied to the perturbation added in~\rf{eqn:optimalInitial}, also the norm of the perturbation is at most $\eps/4$ (\emph{cf.}~\rf{eqn:improvedPumpingTotalPerturbation}).

To get the original statement of the corollary, we do not apply all the input oracles $O^{(i)}$ with $K^{(i)}=1$.
As they are applied once in the algorithm, this gives additional perturbation of size at most $\eps/2$.
Combining with the perturbation $\eps/2$ of the relaxed algorithm itself, we get an estimate of the total error of at most $\eps$.
\end{proof}

\section{Example II: Adversary Bound}
\label{sec:adv}
As mentioned in the introduction, the construction of transducers is based on the implementation of the adversary bound in~\cite{belovs:LasVegas}.
The quantum adversary bound was first developed as a powerful tool for proving quantum query lower bounds.
However, it was later extended to include upper bounds as well, and we consider the latter in this paper.
For more detail on the adversary bound, refer to the introduction of~\cite{belovs:LasVegas} and the references therein.
We first consider the general case of state conversion with unidirectional unitary input oracles, and then move on to more usual function evaluation problems.

\subsection{State Conversion}
\label{sec:stateConversion}
We consider the adversary bound for state conversion from~\cite{belovs:LasVegas}.
In the state conversion problem, we have a collection of pairs $\xi_x\mapsto\tau_x$ of states in $\cH$ and input oracles $O_x\colon \cM\to\cM$, where $x$ ranges over some finite set $D$.
The task is to develop an algorithm $A$ such that $A(O_x)\xi_x = \tau_x$ for all $x$.
The goal is to minimise $L(A, O_x, \xi_x)$.
The corresponding adversary bound is the following multi-objective optimisation problem:
\begin{subequations}
\label{eqn:advExplicit}
\begin{alignat}{3}
&\mbox{\rm minimise} &\quad& \sA[\norm|v_x|^2]_{x\in D} &\quad&\\
& \mbox{\rm subject to}&&  
\ip<\xi_x, \xi_y> - \ip<\tau_x,\tau_y> = \ipA<v_x,\;  (I_{\cW}\otimes(I_\cM-O^*_xO_y)) v_y> && \text{\rm for all $x, y\in D$;}  \label{eqn:advExplicitCondition}\\
&&& \text{$\cW$ is a vector space}, \qquad
v_x \in \cW\otimes\cM.
\end{alignat}
\end{subequations}

A canonical transducer $S_v$ can be obtained from any feasible solution $v = (v_x)$ to this problem.
It works as follows (with $I = I_\cW$):
\[
\xi_x \oplus v_x \maps{I\otimes O_x} \xi_x \oplus (I \otimes O_x) v_x \maps{\Sw_v} \tau_x \oplus v_x,
\]
where $\Sw_v$ is an input-independent unitary whose existence is assured by~\rf{eqn:advExplicitCondition}, as the latter can be rewritten as
\[
\ip<\xi_x, \xi_y> +  \ipA<(I\otimes O_x) v_x,\;  (I\otimes O_y) v_y>
=
\ip<\tau_x, \tau_y> +  \ip<v_x, v_y>,
\]
and two state collections with the same combination of inner products always admit such a state-independent transforming unitary. 
Thus we have a transduction $\xi_x \transduce{S_v(O_x)} \tau_x$ with 
the transduction and the query complexities satisfying
\[
W(S_v, O_x, \xi_x) = L(S_v, O_x, \xi_x) = \|v_x\|^2,
\qquad
q(S_v, O_x, \xi_x) = v_x.
\]

As shown in~\cite{belovs:LasVegas}, this perfectly captures Las Vegas query complexity of state conversion.
Note that canonical transducers with empty non-query space $\cLw$ are essentially equivalent to this construction.

\subsection{Function Evaluation}
\label{sec:function}
Now we describe the usual case of function evaluation.
These results can be derived from~\cite{belovs:LasVegas} and~\cite{belovs:variations}.
First, we define the formalism behind function-evaluating transducers, and then move on to the adversary bound.

Let $f\colon D \to [p]$ be a function with domain $D \subseteq [q]^n$.
We want to construct a transducer $S_f$ that evaluates $f$.
We assume the state-generating settings from \rf{sec:prelimFunctions}, which means
that, for every $x\in D$, with bidirectional access to the input oracle 
$O_x\colon \ket|i>\ket|0>\mapsto \ket |i>\ket|x_i>$,
the transducer $S_f$ has to perform the transduction $\ket |0> \transduce{} \ket |f(x)>$.
Again, we cast this as unidirectional access to $\bi{O_x}$ from~\rf{eqn:bi}.

Similarly to~\rf{eqn:function_QueryComplexity}, we write
\begin{equation}
\label{eqn:function_Wx}
W_x (S) = W\sA[S, \bi{O_x}, \ket |0>]
\qqand
W(S) =  \max_{x\in D} W_x(S).
\end{equation}
We use similar notation for $L$ and $L^{(i)}$.

There is a slight discrepancy between various existing definitions of the adversary bound $\Adv(f)$ for non-Boolean functions, in the sense that they differ by a factor of at most 2 (see, e.g., Section 3 of~\cite{lee:stateConversion}.)
We adopt the formulation from~\cite{belovs:variations}, which reads in notation of that paper as
\[
\Adv(f) = \gamma_2\sB[1_{f(x)\ne f(y)} \midA \bigoplus\nolimits_{i\in[n]} 1_{x_i\ne y_i}]_{x,y\in D},
\]
and which is equivalent to $\gamma_2(J-F | \Delta)$ in notations of~\cite{lee:stateConversion}.
An explicit definition of $\Adv(f)$ is:%
\footnote{
The usual definition has $\max\sfig{\sum_{i=1}^n \|u_{x,i}\|^2, \sum_{i=1}^n \|v_{x,i}\|^2 }$ in the objective instead of $\frac12 \sB[\sum_{i=1}^n \|u_{x,i}\|^2 + \|v_{x,i}\|^2 ]$.
The two formulations are equivalent~\cite{belovs:LasVegas}.
But even a priori, our formulation does not exceed the usual formulation, and, since we are interested in upper bounds, supersedes the latter.
}
\begin{subequations}
\label{eqn:advFunction}
\begin{alignat}{3}
&\mbox{\rm minimise} &\quad& \max_{x\in D} \frac12 \sum_{i=1}^n \sC[\|u_{x,i}\|^2 + \|v_{x,i}\|^2 ] &\quad& \label{eqn:advFunctionObjective}\\
& \mbox{\rm subject to}&&  
1_{f(x)\ne f(y)} = \sum_{i: x_i\ne y_i} \ip<u_{x,i}, v_{y,i}> && \text{\rm for all $x, y\in D$;}  \label{eqn:advFunctionCondition}\\
&&& \text{$\cW$ is a vector space}, \qquad
u_{x,i}, v_{x,i} \in \cW.
\end{alignat}
\end{subequations}

\mycutecommand{\vup}{v^{\uparrow}}
\mycutecommand{\vdown}{v^{\downarrow}}

Let us now describe the canonical transducer $S_{u,v}$ corresponding to a feasible solution $u_{x,i}, v_{x,i}$ of~\rf{eqn:advFunction}.
Its local space $\cL = \cLq$ is of the form $\cW \otimes \cR\otimes \cB\otimes \cQ$. 
Here $\cW$ acts as  $\cLt$ in notation of \rf{sec:canonicalDefinition},
$\cR$ is an $n$-qudit indicating the index of the input variable,
$\cB$ is a qubit indicating direction of the query,
and $\cQ$ is a $q$-qudit storing the output of the query.
We use $\uparrow$ and $\downarrow$ to denote the basis states of $\cB$.
The first one stands for the direct, and the second one for the inverse query.
The input oracle acts on $\cR\otimes \cB\otimes \cQ$ as
\begin{equation}
\label{eqn:function_inputOracleDecomposition}
\bi{O_x} = \bigoplus_{i\in [n]} \bi{O_{x,i}},
\end{equation}
where
\begin{equation}
\label{eqn:function_inputConstituent}
O_{x,i}\colon\quad \cQ\to\cQ,\quad \ket|0>\mapsto \ket |x_i>
\end{equation}
is the $i$th constituent of the input oracle.

Define the following vectors in $\cW$:
\[
\vup_{x,i} = \frac{u_{x,i} + v_{x,i}}2
\qqand
\vdown_{x,i} = \frac{u_{x,i} - v_{x,i}}2.
\]
They possess the following important property:
\begin{equation}
\label{eqn:function_vinner}
\ipA<\vup_{x,i}, \vup_{y,i}> - \ipA<\vdown_{x,i}, \vdown_{y,i}> 
= 
\frac{\ip<u_{x,i}, v_{y,i}> + \ip<v_{x,i}, u_{y,i}>}2.
\end{equation}

The catalyst for the input $x$ is 
\begin{equation}
\label{eqn:functionWitness}
v_x = \sum_{i\in [n]} \ket R|i> \skB[ \ket B|\uparrow>\ket Q|0> \ketA W|v^\uparrow_{x,i}> + \ket B|\downarrow>\ket Q|x_i>\ketA W|v^\downarrow_{x,i}> ].
\end{equation}
The transducer starts in $\xi_x \oplus v_x = \ket |0> \oplus \ket|v_x>$.
It applies the input oracle~\rf{eqn:function_inputOracleDecomposition}, which gives the state
\begin{equation}
\label{eqn:function1}
\psi_x = \ket |0> \oplus 
\sum_{i\in [n]} \ket R|i> \skB[ \ket B|\uparrow>\ket Q|x_i> \ketA W|v^\uparrow_{x,i}> + \ket B|\downarrow>\ket Q|0>\ketA W|v^\downarrow_{x,i}> ].
\end{equation}
The construction then follows from the following claim.

\begin{clm}
There exists an input-independent unitary $\Sw_{u,v}$ that maps the state $\psi_x$ from~\rf{eqn:function1} into $\ket |f(x)> \oplus \ket |v_x>$ for all $x$.
\end{clm}

\pfstart
Indeed, for a pair of $x,y\in D$, we have
\begin{equation}
\label{eqn:functioninner1}
\ip<\psi_x,\psi_y> = 1 + \sum_{i\in [n]} \skA[
1_{x_i = y_i} \ipA<\vup_{x,i}, \vup_{y,i}> + \ipA<\vdown_{x,i}, \vdown_{y,i}>].
\end{equation}
On the other hand, the inner product between 
$\ket |f(x)> \oplus \ket |v_x>$ and $\ket |f(y)> \oplus \ket |v_y>$ is
\begin{equation}
\label{eqn:functioninner2}
1_{f(x)=f(y)} + \sum_{i\in [n]} \skA[
 \ipA<\vup_{x,i}, \vup_{y,i}> + 1_{x_i = y_i}\ipA<\vdown_{x,i}, \vdown_{y,i}>].
\end{equation}
To establish the existence of the unitary $\Sw_{u,v}$ it suffices to show that~\rf{eqn:functioninner1} and~\rf{eqn:functioninner2} are equal for all $x,y\in D$.
Subtracting the latter from the former gives us
\[
1_{f(x)\ne f(y)} - \sum_{i: x_i\ne y_i} \skA[\ipA<\vup_{x,i}, \vup_{y,i}> - \ipA<\vdown_{x,i}, \vdown_{y,i}>] 
=
1_{f(x)\ne f(y)} - \frac12 \sum_{i: x_i\ne y_i} \skA[\ip<u_{x,i}, v_{y,i}> + \ip<v_{x,i}, u_{y,i}>]
= 0
\]
using~\rf{eqn:function_vinner} and~\rf{eqn:advFunctionCondition}.
\pfend

Thus, we have that $S_{u,v}$ on the input oracle $\bi{O_x}$ transduces $\ket|0>$ into $\ket|f(x)>$.
If we consider $\bi{O_x}$ as a direct sum of $n$ input oracles as in~\rf{eqn:function_inputOracleDecomposition}, we get the partial query states
\[
q^{(i)}_x (S_{u,v}) = \vup_{x,i} \oplus \vdown_{x,i},
\]
with
\[
L^{(i)}_x (S_{u,v}) = \|\vup_{x,i}\|^2 +  \|\vdown_{x,i}\|^2 = 
\frac{\|u_{x,i}\|^2 + \|v_{x,i}\|^2}2
\]
by the parallelogram identity.
Finally, the transduction complexity and the total query complexity is
\[
W_x (S_{u,v}) =
L_x (S_{u,v}) = 
\frac12 \sC[\sum_{i=1}^n \|u_{x,i}\|^2 + \|v_{x,i}\|^2 ]
\]
which is in the objective of~\rf{eqn:advFunctionObjective}.
This can be summarised as
\begin{thm}
\label{thm:advTransducer}
For every function $f\colon D\to [p]$ with $D\subseteq[q]^n$, there exists a canonical transducer $S_f$ evaluating the function $f$ and whose transduction complexity is bounded by $\Adv(f)$.
In more detail, 
the admissible subspace of $S_f$ is $\ket|0>$;
for every $x\in D$, $S_f$ transduces $\ket |0>\transduce{} \ket |f(x)>$ with bidirectional access to the input oracle $O_x$ encoding the input string $x$; and 
$
W_x (S_f) =
L_x (S_f) \le \Adv(f).
$
Moreover, the catalyst of the transducer $S_f$ is as in~\rf{eqn:functionWitness}.
In particular, it executes the input oracle only on its admissible subspace.
\end{thm}

\section{Composition of Transducers}
\label{sec:properties}

In this section, we describe basic properties of transducers.
In particular, we show how to combine simple transducers in order to obtain more complex ones.  This is akin to quantum algorithms being build out of elementary operations and subroutines.
In this section, we mostly focus on the circuit model of computation.

\subsection{Basic Properties}
From \rf{thm:transduce}, it follows that, for a fixed $S$ and $O$, the mappings $\xi\mapsto v(S, O, \xi)$ and $\xi\mapsto q(S, O, \xi)$ are linear.
In particular, for $c\in \bC$, we have
\begin{equation}
\label{eqn:rescaling}
W(S, O, c\xi) = |c|^2 W(S, O, \xi)
\qqand
L^{(i)}(S, O, c\xi) = |c|^2 L^{(i)}(S, O, \xi).
\end{equation}
Notice, however, that it is \emph{not} necessarily true that $W(S,O,\xi_1+\xi_2) = W(S,O,\xi_1) + W(S,O,\xi_2)$ even for orthogonal $\xi_1$ and $\xi_2$.
On the other hand, it \emph{is} the case that for $\xi_1\in \cE_1\otimes \cH$ and $\xi_2\in \cE_2\otimes \cH$, we have
\[
W(S,O,\xi_1\oplus \xi_2) = W(S,O,\xi_1) + W(S,O,\xi_2),
\]
using the extended definition of~\rf{eqn:transductionExtended}.

\begin{prp}[Inverse]
\label{prp:inverse}
For a canonical transducer $S$, the inverse transducer $S^{-1}$ satisfies $\tau\transduce{S^{-1}(O^*)} \xi$ whenever $\xi\transduce{S(O)}\tau$.
Moreover, $S^{-1}$ can be implemented in the canonical form, it has the same time complexity as $S$,
\[
W(S^{-1}, O^*, \tau) = W(S, O, \xi),
\qqand
q(S^{-1}, O^*, \tau) = (I\otimes O) q (S, O, \xi).
\]
\end{prp}

\pfstart
Let $v = \vw \oplus \vq$ be the catalyst of the transduction $\xi\transduce{S} \tau$.
From~\rf{eqn:transduce}, it is clear that if $S(O)$ maps $\xi \oplus v \mapsto \tau \oplus v$, then $S(O)^*$ maps $\tau \oplus v \mapsto \xi \oplus v$, hence, transduces $\tau$ into $\xi$ with the same catalyst.
One problem is that its action, as the inverse of~\rf{eqn:canonicalForm},
\[
\tau \oplus \vw \oplus \vq 
\maps{(\Sw)^*}
\xi \oplus \vw \oplus(I\otimes O)\vq 
\maps{\tO^*}
\xi \oplus \vw \oplus \vq.
\]
is not in the canonical form.
But we can take the following transducer $S^{-1}$ in its stead:
\[
\tau \oplus \vw \oplus(I\otimes O)\vq 
\maps{\tO^*}
\tau \oplus \vw \oplus \vq 
\maps{(\Sw)^*}
\xi \oplus \vw \oplus (I\otimes O)\vq.
\]
It is in the canonical form, and satisfies all the conditions.
\pfend

\subsection{Alignment}
In the remaining part of this section, we will study different ways of combining transducers $S_1,\dots,S_m$.
First, we consider parallel composition of transducers $\bigoplus_i S_i$, where individual $S_i$ act on orthogonal parts of the workspace.
Then we move onto sequential composition $S_m * S_{m-1} * \cdots * S_1$, where they act on the same space one after another.
Finally, we consider functional composition of two transducers, where the second transducer acts as an oracle for the first one.

We generally assume that all $S_i$ use the same oracle $O$.
This is without loss of generality since if they use different sets of input oracles, we can assume they all use the oracle $O$ which is the direct sum of the union of these sets of oracles.
Individual $S_i$ will then just ignore the input oracles they are not using.

In principle, the spaces $\cH_i$ and $\cL_i$ can differ between different $S_i$, but we assume they are all embedded into some larger register $\cH$, which also serves as $\cL$ per our convention of \rf{sec:implementation}.
What is crucial, though, is that all $S_i$ use the same privacy qubit $\cP$, the same query register $\cR$, and, most importantly, the oracle $O$ is applied to the same subset of registers in all $S_i$.
This is summarised by the following definition, \emph{cf.} \rf{rem:aligned}:

\begin{defn}[Alignment]\label{defn:alignment}
We say that canonical transducers $S_1,\dots,S_m$ are \emph{aligned in the oracle $O^{(i)}$} if the query of $O^{(i)}$ is conditioned on the same value $\ket R|i>$ of the same register and acts on the same subset of registers in all of them.
We say that $S_1,\dots,S_m$ are \emph{aligned} if they are aligned in all their oracles, and, additionally, use the same privacy qubit $\cP$. 
\end{defn}

If the alignment condition is not satisfied, the transducers have to explicitly move their registers around to meet it.
This might take time if the registers are lengthy.

\subsection{Parallel Composition}
\label{sec:parallel}

Parallel composition of two or several quantum programs is their execution as a direct sum on orthogonal parts of the space of the algorithm.
For transducers, parallel composition can be implemented in a straightforward way.

\begin{defn}[Direct Sum of Transducers]
\label{defn:parallel}
Let $S_1,\dots,S_m$ be canonical transducers.
We assume they all use the same space $\cH\oplus\cL$, the same input oracle $O\colon \cM\to\cM$, and are aligned.
In particular, the query $\tO = I_\cH \oplus \Iw \oplus I\otimes O$ is given by~\rf{eqn:canonicalQuery} and is controlled by $\ket R|1>$ in all of them.

Let the register $\cJ = \bC^m$.
The direct sum $\bigoplus_i S_i$ is a canonical transducer in the space $\cJ\otimes(\cH\oplus\cL) = (\cJ\otimes\cH)\oplus(\cJ\otimes\cL)$.
It has the same input oracle $O$ as all $S_i$.
The query is again controlled by $\ket R|1>$.
The work unitary of $\bigoplus_i S_i$ is $\bigoplus_i \Sw_i$, where $\Sw_i$ is the work unitary of $S_i$.
\end{defn}

\begin{prp}[Parallel Composition]
\label{prp:parallel}
The canonical transducer $S = \bigoplus_i S_i$ from \rf{defn:parallel} satisfies the following conditions.
Assume $\xi_i\transduce{S_i(O)}\tau_i$ for all $i$.
Then, $S(O)$ transduces $\xi = \bigoplus_i \xi_i$ into $\tau = \bigoplus_i \tau_i$.
Moreover,
\begin{equation}
\label{eqn:paralellComplexity}
W(S, O, \xi) = \sum_i W(S_i, O, \xi_i)
\qqand
q(S, O, \xi) = \bigoplus_i q(S_i, O, \xi_i).
\end{equation}
The time complexity of $S$ is equal to the time complexity of implementing $\bigoplus_i \Sw_i$.
\end{prp}

\pfstart
Recall that all $S_i$ are aligned, and, hence, use the same space, $\cP$ is their common privacy qubit, and $\cR$ their common query register.
The privacy and the query registers of $S$ will still be $\cP$ and $\cR$.

Let $v_i$ be the catalyst of transduction $\xi_i \transduce{S_i(O)} \tau_i$.
The initial coupling of $S$ is
\begin{equation}
\label{eqn:parallel_xi+v}
\xi \oplus v = \sum_{i=1}^m \ket J|i> \skA[ \ket P|0> \ket H|\xi_i> + \ket P|1> \ket L|v_i> ].
\end{equation}
As required by canonicity, we first apply the input oracle $O$ controlled by $\ket R|1>$.
This has the effect that $O$ is applied to all $S_i$ in parallel.
Then, conditioned on the value $i$ in $\reg J$, we apply the work unitary $\Sw_i$ to the last two registers.
By the assumption $\xi_i\transduce{S_i(O)}\tau_i$ , this gives
\[
\sum_{i=1}^m \ket J|i> \skA[ \ket P|0> \ket H|\tau_i> + \ket P|1> \ket L|v_i> ] = \tau \oplus v.
\]
Eq.~\rf{eqn:paralellComplexity} follows from~\rf{eqn:parallel_xi+v}.
\pfend

\begin{rem}[Different input oracles in $S_i$]
\label{rem:directSumDifferentOracle}
In \rf{defn:parallel}, we assume all $S_i$ use the same input oracle $O$.
It is, however, possible to allow each $S_i$ to use its own input oracle $O_i$ so that the total input oracle of $\bigoplus_i S_i$ is $\bigoplus_i O_i$, and $\xi_i \transduce{S_i(O_i)} \tau_i$.
This is used, for instance, in iterated and composed functions.
We account for this possibility by allowing $\cM$ to contain $\cJ$, so that the input oracle of $S_i$ has read-only access to the value of $\ket E|i>$.
This does not interfere with the proof of \rf{prp:parallel}, and we get the following identities instead of~\rf{eqn:paralellComplexity}:
\begin{equation}
\label{eqn:paralellComplexityMultipleOracles}
W(S, O, \xi) = \sum_i W(S_i, O_i, \xi_i)
\qqand
q(S, O, \xi) = \bigoplus_i q(S_i, O_i, \xi_i).
\end{equation}
\end{rem}

The time complexity of implementing $\bigoplus_i S_i$ greatly depends on the structure of $S_i$ and the computational model.
For instance, if we assume the QRAG model, implementation of $\bigoplus S_i$ can be done in time essentially $\max_i T(S_i)$ as per~\rf{cor:selectProgram}.
The important special case when all $S_i$ are the same can be efficiently handled in the circuit  model as well. 
We write it out explicitly for ease of referencing.

\begin{defn}[Multiplication by Identity]
\label{defn:byIdentity}
Let $S$ be a canonical transducer with space $\cH\oplus \cL$, input oracle $O\colon \cM\to \cM$, and query given by~\rf{eqn:canonicalQuery}.
Let also $\cE$ be a space, and $I_\cE$ be the identity on $\cE$.
We define $I_\cE \otimes S$ as a canonical transducer in the space $(\cE\otimes \cH)\oplus (\cE\otimes \cL)$, with the work unitary $I_\cE\otimes \Sw$.
In other words, $I_\cE \otimes S = \bigoplus_i S$ where the summation is over the basis of $\cE$.
\end{defn}

\begin{cor}
\label{cor:byIdentity}
In the settings of \rf{defn:byIdentity}, assume $\xi_i\transduce{S(O)}\tau_i$ as $i$ ranges over the basis of $\cE$.
Then, $I_\cE\otimes S(O)$ transduces $\xi = \bigoplus_i \xi_i$ into $\tau = \bigoplus_i \tau_i$.
The complexities are given by~\rf{eqn:paralellComplexity} or~\rf{eqn:paralellComplexityMultipleOracles} with all $S_i=S$.
The time complexity is $T(I_\cE\otimes S) = T(S)$.
\end{cor}

Note that complexities in~\rf{eqn:transductionExtended} are the same as in this corollary with Eq.~\rf{eqn:paralellComplexity} used.

\subsection{Sequential Composition}
\label{sec:composition}

A quantum program~\rf{eqn:program} is a sequence of gates applied one after the other.
Therefore, sequential composition is of prime importance in quantum algorithms.
Let us formally define it for transducers.

\begin{defn}[Sequential Composition]
\label{defn:sequential}
Let $S_1,\dots, S_m$ be an aligned family of canonical transducers, all on the same public space $\cH$ and the same input oracle $O$.
Its sequential composition is a canonical transducer $S = S_m * S_{m-1} * \cdots * S_1$ on the same public space with the following property.
For every sequence of transductions
\begin{equation}
\label{eqn:sequenceOfTransductions}
\xi = \psi_1 \transduce{S_1(O)} \psi_2 \transduce{S_2(O)} \psi_3 \transduce{S_3(O)} \cdots \transduce{S_{m-1}(O)} \psi_{m} \transduce{S_m(O)} \psi_{m+1} = \tau,
\end{equation}
$S(O)$ transduces the initial state $\xi$ into the final state $\tau$.
\end{defn}

The above definition does not specify the implementation as we give two different implementations in this section.
The first one is tailored towards the circuit model, and the second one towards the QRAG model.
This division is not strict though.

\begin{prp}[Sequential Composition, Sequential Implementation]
\label{prp:sequentialsequential}
A sequential composition $S = S_m * S_{m-1} * \cdots * S_1$ as in \rf{defn:sequential} can be implemented by a canonical transducer with the following parameters:
\begin{equation}
\label{eqn:sequentialWitnessComplexity}
W(S, O, \xi) = \sum_{t=1}^m W(S_t, O, \psi_t),
\qquad
q(S, O, \xi) = \bigoplus_{t=1}^m q(S_t, O, \psi_t),
\end{equation}
and
\begin{equation}
\label{eqn:sequentialTimeComplexity}
T(S) = \OO(m)+ \sum_{t=1}^m \TC(S_t),
\end{equation}
where $\TC$ is defined in \rf{sec:prelimCircuit}.
\end{prp}

\pfstart
The general idea is simple.  
We first apply the input oracle to make the queries in all $S_t$ in parallel.
After the query, we execute all $\Sw_t$ one after the other, see \rf{fig:sequentialsequential}.
Some care must be taken, however, so that transducers act inside their respective private spaces and do not interfere with private spaces of other transducers.

\myfigure{\label{fig:sequentialsequential}}
{
A graphical illustration to the construction of \rf{prp:sequentialsequential} with $m=4$.
For convenience of representation, we draw the catalyst states by vertical lines, not horizontal ones like in \rf{fig:canonical}.
We also write $O\vq_t$ instead of $(I\otimes O)\vq_t$ to save space.
\\
We first apply the query $I_m\otimes I\otimes O$, which implements the queries in all $\Sw_t$.
After that, we apply all of $\Sw_t$ one after the other.
}
{
\newcommand{\OneIteration}[1]{
    \edef\indxx{#1}
    \begin{scope}[shift={(3.3*\indxx,0)}]
        \draw (0.5,3.5) rectangle (2.5, 4.5) node[pos=0.5] {\Large $\Sw_{\indxx}$};
        \draw[\witnesscolor] (1, 5.3) node[above] {$\vw_{\indxx}$} to (1, 4.5);
        \draw[\witnesscolor] (1, 3.5) to (1, 2.7) node[below] {$\vw_{\indxx}$};
        \draw[\nonquerycolor] (2, 7.8) node[above]{$\vq_{\indxx}$} to (2, 7);
        \draw[\querycolor] (2, 6) to node[right]{$O\vq_{\indxx}$} (2, 4.5);
        \draw[\nonquerycolor] (2, 3.5) to (2, 2.7) node[below]{$\vq_{\indxx}$};
    \end{scope}
}
\newcommand{\Mezhdu}[1]{
    \edef\indxx{#1}
    \begin{scope}[shift={(3.3*\indxx,0)}]
        \draw[thick, black,->] (-0.8,4) to node[above]{$\psi_{\indxx}$} (0.5,4);      
    \end{scope}
}
\negbigskip
\negbigskip
\[
\begin{tikzpicture}[every node/.style={font=\scriptsize}, every path/.append style={thick,->}]
\OneIteration{1}
\Mezhdu{2}
\OneIteration{2}
\Mezhdu{3}
\OneIteration{3}
\Mezhdu{4}
\OneIteration{4}
\draw (4,6) rectangle (15.7, 7) node[pos=0.5] {\Large $I_m\otimes I\otimes O$};
\draw[\xicolor] (2.5,4) node[left] {\normalsize $\xi$} to node[above]{$\psi_{1}$} (3.8,4);
\draw[\taucolor] (15.7,4) to node[above]{$\psi_{5}$} (16.8,4) node[right] {\normalsize $\tau$};
\end{tikzpicture}
\]
\negbigskip
}

The transducer $S$ uses the same privacy and query registers $\cP$ and $\cR$ as the family $S_1,\dots,S_m$.
We additionally use an $m$-qudit $\cK$, so that the local space of the transducer $S_t$ is marked by the value $\ket K|t>$.
Therefore, the space of $S$ is of the form $\cK\otimes (\cH\oplus \cL)$.
Its public space is still $\cH$ embedded as $\ket K|0> \otimes \cH$.

Let $\xi$ and $\psi_t$ be as in~\rf{eqn:sequenceOfTransductions}.
For each $t$, let $v_t$ be the catalyst of the transduction $\psi_t \transduce{S_t(O)} \psi_{t+1}$ from~\rf{eqn:sequenceOfTransductions}.
For the transduction $\xi \transduce{S(O)} \tau$, we have the following initial coupling:
\begin{equation}
\label{eqn:sequentialWitness1}
\xi\oplus v = \ket P|0>\ket K|0>\ket H|\xi> + \ket P|1> \sum_{t=1}^m \ket K |t>\ket L|v_t>.
\end{equation}
This already implies~\rf{eqn:sequentialWitnessComplexity}.

As required by the definition, the first operation is the application of the input oracle conditioned on $\ket R|1>$.
This gives the state
\[
\ket P|0>\ket K|0>\ket H|\psi_1> + \ket P|1> \sum_{t=1}^m \ket K |t>\ketA L|\tO v_t>,
\]
where $\tO v_t$ is defined as in~\rf{eqn:Ov}, and we used that $\xi = \psi_1$.

The work part of the transducer $S$ is as follows.
We assume the indexing of $\reg K$ is done modulo $m$.
\begin{itemize}
\item For $t=1,\dots,m$:
\begin{enumerate}
\item Controlled by $\ket P|0>$, replace the value $\ket K|t-1>$ by $\ket K|t>$.
\item Apply the work unitary $\Sw_t$ controlled by $\ket K|t>$.
\end{enumerate}
\end{itemize}

By induction and using~\rf{eqn:canonicalSequence}, after $\ell$ iterations of the loop, we have the state
\[
\ket P|0>\ket K|\ell>\ket H|\psi_{\ell+1}> 
+ \ket P|1> \sum_{t=1}^\ell \ket K |t>\ket L|v_t>
+ \ket P|1> \sum_{t=\ell+1}^m \ket K |t>\ketA L|\tO v_t>.
\]
And after execution of the whole transducer $S$, we have the state
\[
\ket P|0>\ket K|0>\ket H|\tau> + \ket P|1> \sum_{t=1}^m \ket K |t>\ket L|v_t> = \tau \oplus v,
\]
where we used that $\psi_{m+1}=\tau$.
The time complexity in~\rf{eqn:sequentialTimeComplexity} can be achieved using the direct-sum finite automaton of \rf{lem:automaton}.
\pfend

The second implementation of the sequential composition is slightly less intuitive.
We move all the intermediate states in~\rf{eqn:sequenceOfTransductions} to the catalyst.
This increases the catalyst size, but now all the operations $\Sw_t$ can be implemented in parallel.

\begin{prp}[Sequential Composition, Parallel Implementation]
\label{prp:sequentialParallel}
A sequential composition $S = S_m * S_{m-1} * \cdots * S_1$ as in \rf{defn:sequential} can be implemented by a transducer with the following parameters:
\begin{equation}
\label{eqn:sequentialComplexity}
W(S, O, \xi) = \sum_{t=2}^m \|\psi_t\|^2 + \sum_{t=1}^m W(S_t, O, \psi_t)
,\qquad
q(S, O, \xi) = \bigoplus_{t=1}^m q(S_t, O, \psi_t)
\end{equation}
and
\begin{equation}
\label{eqn:sequentialParallelTime}
T(S) = \OO(\log m) + T \sC[\bigoplus_t \Sw_t].
\end{equation}
\end{prp}

In the QRAG model, we usually replace $\OO(\log m)$ in~\rf{eqn:sequentialParallelTime} by $\OO(\timeR)$.

\pfstart
The main new idea compared to \rf{prp:sequentialsequential} is that we do not compute the intermediate $\psi_t$ from~\rf{eqn:sequenceOfTransductions}, but store the sequence $\psi_2,\dots,\psi_m$ in the catalyst.  
We apply all of $\Sw_t$ in parallel, thus moving this sequence forward by one position,
see \rf{fig:sequentialParallel}.

\myfigure{\label{fig:sequentialParallel}}
{
A graphical illustration to the construction of \rf{prp:sequentialParallel} with $m=4$.
For convenience of representation, we draw the catalyst states by vertical lines, not horizontal ones like in \rf{fig:canonical}.
We also write $O\vq_t$ instead of $(I\otimes O)\vq_t$ to save space.
\\
We first apply the query $I_m\otimes I\otimes O$, which implements the queries in all $\Sw_t$.
After that, we apply all of $\Sw_t$ in parallel.
The initial state $\xi$ becomes the input for $\Sw_1$.
The input of $\Sw_t$ for $t>1$ is taken from the catalyst, and the output becomes the catalyst for $t+1$, except for $t=m$, which yields the terminal state $\tau$.
}
{
\newcommand{\OneIteration}[1]{
    \edef\indxx{#1}
    \begin{scope}[shift={(3.3*\indxx,0)}]
        \draw (0.5,3.5) rectangle (2.5, 4.5) node[pos=0.5] {\Large $\Sw_{\indxx}$};
        \draw[\witnesscolor] (1, 5.3) node[above] {$\vw_{\indxx}$} to (1, 4.5);
        \draw[\witnesscolor] (1, 3.5) to (1, 2.7) node[below] {$\vw_{\indxx}$};
        \draw[\nonquerycolor] (2, 7.8) node[above]{$\vq_{\indxx}$} to (2, 7);
        \draw[\querycolor] (2, 6) to node[right]{$O\vq_{\indxx}$} (2, 4.5);
        \draw[\nonquerycolor] (2, 3.5) to (2, 2.7) node[below]{$\vq_{\indxx}$};
    \end{scope}
}
\newcommand{\Golova}[1]{
    \edef\indxx{#1}
    \begin{scope}[shift={(3.3*\indxx,0)}]
        \draw[\xicolor] (-0.1,5.3) node[above] {$\psi_{\indxx}$}  .. controls (-0.1,4) .. (0.5,4);        
    \end{scope}
}
\newcommand{\Hvost}[1]{
    \edef\indxx{#1}
    \begin{scope}[shift={(3.3*\indxx-3.3,0)}]
        \draw[\taucolor,<-] (3.2,2.8) node[below] {$\psi_{\indxx}$} .. controls (3.2,4) .. (2.5,4);        
    \end{scope}
}
\negbigskip
\negbigskip
\[
\begin{tikzpicture}[every node/.style={font=\scriptsize}, every path/.append style={thick,->}]
\OneIteration{1}
\Hvost{2}
\Golova{2}
\OneIteration{2}
\Hvost{3}
\Golova{3}
\OneIteration{3}
\Hvost{4}
\Golova{4}
\OneIteration{4}
\draw (4,6) rectangle (15.7, 7) node[pos=0.5] {\Large $I_m\otimes I\otimes O$};
\draw[\xicolor] (2.5,4) node[left] {\normalsize $\xi$} to node[above]{$\psi_{1}$} (3.8,4);
\draw[\taucolor] (15.7,4) to node[above]{$\psi_{5}$} (16.8,4) node[right] {\normalsize $\tau$};
\end{tikzpicture}
\]
\negbigskip
}

Again, we assume all $S$ have space $\cH\oplus\cL$, have privacy qubit $\cP$, and query register $\cR$.
The transducer $S$ uses the same query register $\cR$, but we introduce a new privacy qubit $\cP'$.
We also use an $m$-qudit $\cK$, so that the space of $S$ is of the form $\cK\otimes \cP'\otimes (\cH\oplus \cL)$.

Let $\xi$ and $\psi_t$ be as in~\rf{eqn:sequenceOfTransductions}.
For each $t$, let $v_t$ be the catalyst of transduction $\psi_t \transduce{S_t(O)} \psi_{t+1}$.
We start with the following initial coupling
\[
\xi \oplus v = 
\ket P'|0>\ket H|\xi> + \ket P'|1> \skC[\sum_{t=2}^m \ket K |t>\ket P|0>\ket H|\psi_{t}> + \sum_{t=1}^m \ket K|t>\ket P|1>\ket L|v_t>].
\]
This already gives~\rf{eqn:sequentialComplexity}.
As required by the canonicity assumption, we first apply the input oracle.
This gives
\[
\ket P'|0>\ket H|\xi> + \ket P'|1> \skC[\sum_{t=2}^m \ket K |t>\ket P|0>\ket H|\psi_{t}> + \sum_{t=1}^m \ket K|t>\ket P|1>\ket L|\tO v_t>],
\]
where $\tO$ is defined in~\rf{eqn:Ov}.
Here we use that $\cH$ has value $0$ in the register $\cR$ and is not affected by the input oracle.
Next, we exchange $\ket P'|0>\otimes \cH$ with $\ket P'|1> \ket K|1> \ket P|0> \otimes \cH$.
Since $\xi=\psi_1$, this gives
\[
\ket P'|1> \skC[\sum_{t=1}^m \ket K |t>\ket P|0>\ket H|\psi_{t}> + \sum_{t=1}^m \ket K|t>\ket P|1>\ket L|\tO v_t>].
\]
Now we apply $\Sw_{t}$ to the last two registers, controlled by the value in the register $\cK$.
In other words, we apply $\bigoplus_t \Sw_{t}$.
By~\rf{eqn:canonicalSequence}, this gives
\begin{equation}
\label{eqn:sequential1}
\ket P'|1> \skC[\sum_{t=1}^m \ket K |t>\ket P|0>\ket H|\psi_{t+1}> + \sum_{t=1}^m \ket K|t>\ket P|1>\ket L|v_t>].
\end{equation}
Now we increment the value in the register $\cK$ conditioned on $\ket P|0>$, and exchange $\ket P'|0>\otimes \cH$ with $\ket P'|1> \ket K|m+1> \ket P|0> \otimes \cH$.
Since $\tau=\psi_{m+1}$, we get
\begin{equation}
\label{eqn:sequential2}
\ket P'|0>\ket H|\tau> + \ket P'|1> \skC[\sum_{t=2}^m \ket K |t>\ket P|0>\ket H|\psi_{t}> + \sum_{t=1}^m \ket K|t>\ket P|1>\ket L|v_t>] = \tau \oplus v.
\end{equation}
The time complexity estimate is obvious.
\pfend

\subsection{Functional Composition}
\label{sec:functional}

We defined functional composition in \rf{sec:conceptualTypes} as an algorithm where a subroutine is used to implement an oracle call.
In the case of transducers, this falls under the transitivity of transduction,  \rf{prp:transitivityOfTransduction}, as part of the transducer (the oracle call) is implemented as a transduction action of another transducer.
In this section, we give an explicit construction for canonical transducers.

We assume the settings similar to \rf{fig:composition},
except we use transducers $S_A$ and $S_B$ instead of programs $A$ and $B$.
The outer transducer $S_A$ has two input oracles $O\oplus O'$ joined as in \rf{sec:multipleInputOracles}.
The first one, $O$, is the global input oracle, and the second one, $O'$, is the oracle realised as the transduction action of the inner transducer $S_B$.
We assume both transducers are in the canonical form and aligned in the oracle $O$.
We denote their public spaces by $\cH_A$ and $\cH_B$, respectively.
If there are several subroutines implemented by different transducers, then, as in \rf{fig:compositionMultiple}, they can be joined via parallel composition of \rf{prp:parallel}.

\begin{prp}[Functional Composition]
\label{prp:functional}
Under the above assumptions, there exists a canonical transducer $S_A\circ S_B$ with the public space $\cH_A$ and the oracle $O$, which satisfies the following properties.
\begin{itemize}\itemsep=0pt
\item Its transduction action on input oracle $O$ is equal to the transduction action of $S_A$ on oracle $O \oplus O'$, where $O' = S_B(O)\DownTransduce_{\cH_B}$.
\item Its transduction complexity satisfies
\begin{equation}
\label{eqn:functionalTransductionComplexity}
W(S_A\circ S_B, O, \xi) = W(S_A, O\oplus O', \xi) + W\sB[ S_B, O, q^{(1)}(S_A, O\oplus O', \xi)].
\end{equation}
\item Its total query state is
\begin{equation}
\label{eqn:functionalQueryComplexity}
q(S_A\circ S_B, O, \xi) = q^{(0)}(S_A, O\oplus O', \xi) \oplus q\sB[ S_B, O, q^{(1)}(S_A, O\oplus O', \xi)].
\end{equation}
\item Its time complexity is $\TC(S_A)+\TC(S_B)$.
\end{itemize}
Here $q^{(0)}$ and $q^{(1)}$ denote the partial query states of $S_A$ to the oracles $O$ and $O'$, respectively.
We also used the extended versions of $W$ and $q$ from~\rf{eqn:transductionExtended} for the transducer $S_B$.
\end{prp}

Note that both Equations~\rf{eqn:functionalTransductionComplexity} and~\rf{eqn:functionalQueryComplexity} are reminiscent of the ``gold standard''~\rf{eqn:randomComposition2}.

\pfstart
Let $\cM$ be the space of the input oracle $O$.
The space of the transducer $S_A$ is of the form 
$\cH_A\oplus \cLw_A \oplus \cLq_A \oplus \cL'_A$, where the last two spaces are for the queries to $O$ and $O'$, respectively.
In particular, $\cL'_A = \cLt_A \otimes \cH_B$ by the assumption that $S_B(O)$ implements $O'$.

Let $v = \vw \oplus \vq \oplus v'$ be the catalyst for the transduction $\xi\transduce{}\tau$ by $S_A(O\oplus O')$.
That is, $\vq = q^{(0)}(S_A, O\oplus O', \xi)$ and $v' = q^{(1)}(S_A, O\oplus O', \xi)$.
By~\rf{eqn:canonicalSequence}, the transducer $S_A$ imposes the following chain of transformations in $\cH_A\oplus \cLw_A \oplus \cLq_A \oplus \cL'_A$:
\begin{equation}
\label{eqn:functional_1}
\xi \oplus \vw \oplus  \vq \oplus v'
\maps{\widetilde{O}}
\xi \oplus \vw \oplus (I\otimes O)\vq \oplus v'
\maps{\widetilde{O'}}
\xi \oplus \vw \oplus (I\otimes O)\vq \oplus (I\otimes O') v'
\maps{\Sw_A}
\tau \oplus \vw \oplus  \vq \oplus v',
\end{equation}
where we separated the applications of $O$ and $O'$.
Recall that the identity $I$ acts on $\cLt_A$.

\mycutecommand{\ww}{w^\circ}
\mycutecommand{\wq}{w^\bullet}

The space of $S_B$ is $\cH_B \oplus \cLw_B \oplus \cLq_B$, where $\cLq_B = \cLt_B\otimes \cM$.
We obtain the transducer $I\otimes S_B$ as in \rf{defn:byIdentity}, whose space is $\cLt_A\otimes (\cH_B\oplus \cLw_B\oplus \cLq_B)$.

Let $w= \ww \oplus \wq$ be the catalyst for the transduction $v'\transduce{} (I\otimes O') v'$ by $I\otimes S_B(O)$.  In particular, $\wq =  q\sA[ S_B, O, q^{(1)}(S_A, O\oplus O', \xi)]$.
By~\rf{eqn:canonicalSequence} again, we have the chain of transformations in $\cLt_A\otimes (\cH_B\oplus \cLw_B\oplus \cLq_B)$:
\begin{equation}
\label{eqn:functional_2}
v' \oplus \ww \oplus  \wq
\maps{\tO}
v' \oplus \ww \oplus (I\otimes I_B\otimes O)\wq
\maps{I\otimes \Sw_B}
(I\otimes O')v' \oplus \ww \oplus \wq,
\end{equation}
where $I_B$ is the identity on $\cLt_B$.

\myfigure{\label{fig:functional}}
{
A graphical representation of a function composition of transducers $S_A$ and $S_B$.
We omit the tensor multipliers of $O$ and $O'$ on the arrows to save space.
The vector $v'$ is simultaneously a non-queried catalyst for $S_A\circ S_B$ and the input to $I\otimes S_B$.
}
{
\negbigskip
\[
\begin{tikzpicture}[every node/.style={font=\scriptsize}, every path/.append style={thick,->}]
    \draw[\xicolor] (0,0) node[above] {$\xi$} to (11,0) ;
    \draw[\witnesscolor] (0,-1) node[above] {$\vw$} to (11,-1) ;
    \draw[\nonquerycolor,out=0,in=180] (0,-2) node[above] {$\vw$} to (2,-4) ;
    \draw[-,purple,out=0,in=180] (0,-3) node[above] {$v'$} to (2,-2) ;
    \draw[-,\witnesscolor,out=0,in=180] (0,-4) node[above] {$\ww$} to (2,-3) ;
    \draw[\nonquerycolor] (0,-5) node[above] {$\wq$} to (2,-5) ;
    \draw[-,purple] (2,-2) to (5,-2);
    \draw[-,\witnesscolor] (2,-3) to (5,-3);
    \draw (2, -3.5) rectangle (5,-5.5) node[pos=0.5] {\normalsize$(I\oplus I\otimes I_B)\otimes O$};
    \draw[purple,out=0,in=180] (5,-2) to (7,-3) ;
    \draw[\witnesscolor,out=0,in=180] (5,-3) to (7,-4) ;
    \draw[\querycolor] (5,-5) to node[above] {$O\wq$} (7,-5) ;
    \draw[-,\querycolor, out=0, in=180] (5,-4) to (7,-2) ;
    \draw (7, -2.5) rectangle (9,-5.5) node[pos=0.5] {\normalsize$I\otimes \Sw_B$};
    \draw[\querycolor] (7,-2) to node[above,pos=0.4] {$O\vq$} (11,-2) ;
    \draw[\querycolor] (9,-3) to node[above] {$O'v'$} (11,-3);
    \draw[\witnesscolor] (9, -4) to (13,-4) node[above]{$\ww$};
    \draw[\nonquerycolor] (9, -5) to (13,-5) node[above]{$\wq$};
    \draw (11,0.5) rectangle (12, -3.5) node[pos=0.5] {\Large $\Sw_A$};
    \draw[\xicolor] (12,0) to (13,0) node[above] {$\xi$} ;
    \draw[\witnesscolor] (12,-1) to (13,-1) node[above] {$\vw$} ;
    \draw[\nonquerycolor] (12,-2) to (13,-2) node[above] {$\vq$} ;
    \draw[purple] (12,-3) to (13,-3) node[above] {$v'$} ;
\end{tikzpicture}
\]
\negbigskip
}

The transducer $S_A\circ S_B$ works as follows; see \rf{fig:functional}.  It starts in
\begin{equation}
\label{eqn:functional_witness}
\xi \oplus \vw\oplus \vq \oplus v' \oplus \ww \oplus \wq.
\end{equation}
It applies the input oracle $O$ to the third and the last terms, which corresponds to the first steps in both~\rf{eqn:functional_1} and~\rf{eqn:functional_2}.
This gives
\begin{equation}
\label{eqn:functional_A}
\xi \oplus \vw\oplus (I\otimes O)\vq \oplus v' \oplus \ww \oplus (I\otimes I_B\otimes O)\wq.
\end{equation}
Then it performs $I\otimes \Sw_B$, which is the second operation in~\rf{eqn:functional_2}, and which gives
\begin{equation}
\label{eqn:functional_B}
\xi \oplus \vw\oplus (I\otimes O)\vq \oplus (I\otimes O') v' \oplus \ww \oplus \wq.
\end{equation}
After that, $\Sw_A$ is performed, which is the last operation of~\rf{eqn:functional_1}, and we get the final state
\begin{equation}
\label{eqn:functional_C}
\tau \oplus \vw\oplus \vq \oplus v' \oplus \ww \oplus \wq.
\end{equation}

The transduction~\rf{eqn:functionalTransductionComplexity} and the query~\rf{eqn:functionalQueryComplexity} complexities follow from~\rf{eqn:functional_witness}.
In order to get the time complexity, we have to elaborate on the placement of all these transducers in registers.
Since $S_A$ and $S_B$ are aligned in $O$, we assume they use the same value $\ket R|1>$ of the same register to denote its execution.
Concerning $O'$, which is an oracle only for $S_A$, we assume it is indicated by another qubit $\cR'$.
Let $\cP_A$ and $\cP_B$ be the privacy qubits of $S_A$ and $S_B$, respectively; we assume they are different.

Let us write down the values of these indicator qubits for the various subspaces of the composed transducer:
\[
\begin{array}{l|cccccc|}
        & \cH_A	& \cLw_A & \cLq_A & \cL'_A \text{ and } \cH_B 	& \cLw_B & \cLq_B \\\hline
\cP_A 	& 0		& 1		 & 1	  & 1			   	& 1		& 1\\
\cR 	& 0		& 0		 & 1	  & 0			   	& 0		& 1\\
\cR'	& 0		& 0		 & 0	  & 1				& 1		& 1\\
\cP_B	& 0		& 0		 & 0      & 0				& 1		& 1 \\\hline
\end{array}
\]
Other than that, we assume we can embed all these subspaces into registers also preserving $\cL'_A = \cLt\otimes \cH_B$.
The privacy qubit of $S_A\circ S_B$ is $\cP_A$ and its query register is $\cR$.

Let us go through the steps of $S_A\circ S_B$.
The application of the query to $O$ to get to~\rf{eqn:functional_A} is conditioned on $\ket R|1>$ as required, and it is possible because we assume $S_A$ and $S_B$ are aligned in $O$.
The application of $I\otimes \Sw_B$ to get to~\rf{eqn:functional_B} is conditioned on $\ket R'|1>$.
Finally, the application of $\Sw_A$ to get to the final state~\rf{eqn:functional_C} is conditioned on $\ket P_B|0>$.
This gives the required time complexity estimate.
\pfend

\section{Transducers from Programs}
\label{sec:programs->transducers}

In this section, we describe how to convert a quantum program $A$ like in~\rf{eqn:program} into a canonical transducer $S_A$, and give the corresponding corollaries.
We consider both the circuit and the QRAG models.
The general idea is similar in both cases.  
First, we obtain transducers for individual gates, and then compose them using either \rf{prp:sequentialsequential} for the circuit model, or \rf{prp:sequentialParallel} for the QRAG model.

\subsection{General Assumptions}
\label{sec:assumptions}
We are given a program $A$, and our goal is to construct a canonical transducer $S_A$ whose transduction action is identical to the action of $A$.

One thing to observe is that we have to modify our assumptions on the execution of input oracles by introducing an additional register related to $\cR$.
There are two main reasons for that.
First, the initial state $\xi$ of $S_A$ can be any state in $\cH$.
In particular, it can use $\cHq$ which is in contradiction with the assumption of \rf{sec:canonicalDefinition} that $\xi$ must be contained in $\cHw$.
Second, as described in \rf{sec:introAlgorithm->Transducer}, the catalyst in the QRAG case is the history state~\rf{eqn:introHistoryState}.
Various $\psi_t$ can use $\cHq$, and we would like to protect them from the application of the input oracle in $S_A$.

Because of that, we introduce an additional qubit $\ctR$.
The input oracle in the canonical transducer $S_A$ is applied conditioned by $\ket \ctR|1>$.
More precisely, the $i$-th input oracle is controlled by $\ket \ctR|1> \ket R|i>$.

Additionally, we assume that all the queries made by the program $A$ are aligned, see \rf{defn:alignment}, and, in case of several programs, they are aligned by their shared input oracles.
This is necessary since the composition results of \rf{sec:properties} require the transducers to be aligned.

\subsection{Building Blocks}
\label{sec:buildingBlocks}

Here we describe transducers corresponding to elementary operations.
The following proposition is trivial.

\begin{prp}[Trivial Transducer]
\label{prp:gates}
Let $A$ be a quantum algorithm (without oracle calls) that implements some unitary in $\cH$.
It can be considered as a transducer with $\cL$ being empty, $T(A)$ being the time complexity of $A$, and both $W(A,\xi)$ and $q(A,\xi)$ equal to zero.
\end{prp}

Oracle execution is more tricky because of the canonicity assumption.
Recall the query $\tO  = \Iw \oplus I\otimes O$ from~\rf{eqn:query} taking place in $\cH = \cHw\oplus\cHq$ with $\cHq = \cHt\otimes\cM$.
Let us divide $\xi = \xi^\circ \oplus \xi^\bullet$ accordingly.

\begin{prp}[Oracle]
\label{prp:oracle}
For a fixed embedding $\cH = \cHw\oplus\cHt\otimes\cM$ as above, there exists a canonical transducer $S_{\mathrm Q}$ such that $\xi\transduce{S_{\mathrm Q}(O)}\tO\xi$ for all $\xi\in\cH$ and unitaries $O\colon \cM\to\cM$.
The transducer satisfies $W(S_{\mathrm Q}, O,\xi) = \|\xi^\bullet\|^2$, $q(S_{\mathrm Q}, O, \xi) = \xi^\bullet$, and $T(S_{\mathrm Q})=O(1)$.
\end{prp}

\pfstart
The private space of $S_{\mathrm Q}$ will be $\cL = \cLq = \cLt\otimes\cM$ with $\cLt$ equal to $\cHt$.
The catalyst is $v = \vq = \xi^\bullet$.
The transducer $S_{\mathrm Q}$ first applies the input oracle to $\cLq$, which is achieved by conditioning on $\ket\ctR|1>$:
\begin{multline*}
\ket P |0> \ket\ctR|0> \ket H |\xi> + \ket P|1> \ket\ctR|1>\ket L |\xi^\bullet>
\maps{\tO}
\ket P |0> \ket\ctR|0> \ket H |\xi> + \ket P|1> \ket\ctR|1> \ketA L |(I\otimes O)\xi^\bullet>\\
=
\ket P|0> \ket\ctR|0> \ket R|0> \ket |\xi^\circ> + \ket P|0> \ket\ctR|0> \ket R|1> \ket |\xi^\bullet>
+ \ket P|1> \ket\ctR|1> \ket R|1> \ket |(I\otimes O)\xi^\bullet>.
\end{multline*}
Now apply C-NOT to both $\cP$ and $\ctR$ controlled by $\ket R|1>$.
This gives
\[
\ket P|0> \ket\ctR|0> \ket R|0> \ket |\xi^\circ> + \ket P|0> \ket\ctR|0> \ket R|1> \ket |(I\otimes O)\xi^\bullet>
+ \ket P|1> \ket\ctR|1> \ket R|1> \ket |\xi^\bullet>
=
\ket P |0> \ket\ctR|0> \ketA H |\tO\xi> + \ket P|1> \ket\ctR|1>\ketA L |\xi^\bullet>.
\qedhere
\]
\pfend

\subsection{Circuit Model}
\label{sec:alg2TransducersCircuit}

For the circuit model, we have the following result, which is a first half of \rf{thm:introProg->Transducer}.

\begin{thm}[Program to Transducer, Circuit Model]
\label{thm:program->transducerCircuit}
Let $A=A(O)$ be a quantum program in some space $\cH$ assuming the circuit model.
We assume all the queries in $A$ are aligned, see \rf{defn:alignment}.
Then, there exists a canonical transducer $S_A(O)$ with the following properties.
\begin{itemize}\itemsep=0pt
\item $S_A(O)\DownTransduce_\cH = A(O)$ for all input oracles $O$.
\item For any input oracle $O$ and the initial state $\xi$, we have $q(S_A, O, \xi) = q(A, O, \xi)$.  In particular, $L(S_A, O, \xi) = L(A, O, \xi)$.
\item The catalyst $v(S_A, O,\xi)$ is equal to the query state $q(S_A, O, \xi)$.  In particular, $W(S_A, O, \xi) = L(A, O, \xi)$.
\item Finally, the transducer $S_A$ can be implemented in the circuit model, and $T(S_A) = \TC(A) + \OO(Q(A)) = \OO(T(A))$, where the complexity measures of $A$ are as in \rf{sec:prelimCircuit}.
\end{itemize}
\end{thm}

\pfstart
Take the representation of the query algorithm as in~\rf{eqn:preAlgorithm}.
We interpret each unitary $U_t$ as a transducer using \rf{prp:gates}, and transform each query into an independent copy of the transducer $S_{\mathrm Q}$ from~\rf{prp:oracle}.

Now, we apply sequential composition of \rf{prp:sequentialsequential} with $m=2Q(A)+1$.
The total query state $q(A, O, \xi)$ is defined as the direct sum of all the queries given to the input oracle, hence, the third and the second points follow from~\rf{eqn:sequentialWitnessComplexity}.
The time complexity follows from~\rf{eqn:sequentialTimeComplexity} using that $T(A) = Q(A) + \sum_t T(U_t)$.
\pfend

While canonical transducers are nice from the theoretical point of view, designing one from scratch is usually inconvenient.
The following result shows that we can convert any transducer into a canonical form with a slight increase in complexity.

Let $S = S(O)$ be a transducer in $\cH\oplus \cL$ not in the canonical form.
Recall from \rf{sec:transducerDefinition} that we defined its total query state $q_\cH(S,O,\xi)$, and Las Vegas query complexity $L_\cH(S, O, \xi)$ on $\xi\in \cH$ as that of $S$, considered as a usual quantum algorithm in $\cH\oplus \cL$, on the initial state $\xi\oplus v(S(O), \xi)$, where $v\sA[S(O),\xi]$ is as in \rf{thm:transduce}.
We assume all the queries made in $S$ are aligned.

\begin{prp}[Transforming Transducers into Canonical Form]
\label{prp:canoning}
For a non-canonical transducer $S = S(O)$ as above, there exists a canonical transducer $S'= S'(O)$ with the same transduction action and such that, for all $O$ and $\xi$,
\begin{equation}
\label{eqn:canoningComplexity}
q(S', O, \xi) = q_\cH(S, O, \xi),
\qquad
W(S', O, \xi) = W(S(O), \xi) + L_\cH(S, O, \xi),
\end{equation}
and $T(S') = \OO(T(S))$.
\end{prp}

\pfstart
Let $S_S = S_S(O)$ be a canonical transducer as obtained in \rf{thm:program->transducerCircuit} from $S=S(O)$ considered as a quantum program.
Its public space is $\cH\oplus \cL$, and $S_S(O)\DownTransduce_{\cH\oplus \cL} = S(O)$.

Let $\xi' = \xi \oplus v\sA[S(O),\xi]$.  By \rf{thm:program->transducerCircuit}, we have $T(S_S) = \OO\sA[T(S)]$, and 
\begin{equation}
\label{eqn:canoning1}
v_{\cH\oplus \cL}(S_S, O, \xi') 
= q_{\cH\oplus \cL}(S_S, O, \xi') 
= q_{\cH\oplus \cL}(S, O, \xi')
= q_\cH(S, O, \xi),
\end{equation}
where, in the third expression, we consider $S$ as a usual quantum program in $\cH\oplus \cL$.

We obtain $S' = S'(O)$ as the transduction action of $S_S(O)$ on $\cH$, using the transitivity of transduction, \rf{prp:transitivityOfTransduction}.
First, $S_S(O)$ is in canonical form when considered as a transducer with the public space $\cH\oplus \cL$.
A fortiori, it is canonical also on the public space $\cH$.
Next, its time complexity does not change, hence, $T(S') = T(S_S) = \OO\sA[T(S)]$.
For the query state, from~\rf{eqn:transitivityWitness} and~\rf{eqn:canoning1}, we get
\[
q_\cH(S', O, \xi) = q_{\cH\oplus\cL} (S_S, O, \xi') = q_\cH(S, O, \xi).
\]
Finally, for the transduction complexity, we can utilise~\rf{eqn:transitivityWitness} in the following way:
\begin{align*}
W(S', O, \xi) = W_\cH\sA[S_S(O), \xi] &= W_\cH\sA[S_S(O)\DownTransduce_{\cH\oplus\cL}, \xi] + W_{\cH\oplus\cL} \sA[ S_S(O), \xi'] \\&= W_\cH(S(O), \xi) + L_\cH(S, O, \xi),
\end{align*}
where we used~\rf{eqn:canoning1} in the last equality.
\pfend

The remaining corollaries of \rf{thm:program->transducerCircuit} were already proven in \rf{sec:overview}: Theorems~\ref{thm:introQueryCompression} and~\ref{thm:introCompositionCircuit}.

Let us note that \rf{thm:program->transducerCircuit} might not always be the best way to get a transducer $S_A$ from a quantum program $A$ in the circuit model.
One can use additional structure of $A$ to get better transducers.
For instance, if $A$ repeatedly uses the same sequence of gates, one can define it as a new oracle, thus reducing the time complexity of the transducer, see, e.g., \rf{prp:iteratedFunctionsImproved}.
If $A$ contains a large loop, one can use \rf{prp:sequentialParallel}.
Obtaining efficient transducers in the circuit model for specific cases seems like an interesting research direction.

\subsection{QRAG Model}
\label{sec:alg2TransducersQRAG}
Assuming the QRAG model, we get the remaining half of \rf{thm:introProg->Transducer}.
\begin{thm}
\label{thm:program->transducerQRAG}
Let $A=A(O)$ be a quantum program in some space $\cH$, and assume we have QRAM access to the description of $A$.
Then, there exists a canonical transducer $S_A(O)$ with the following properties.
\begin{itemize}\itemsep=0pt
\item $S_A(O)\DownTransduce_\cH = A(O)$ for all input oracles $O$.
\item For any input oracle $O$ and initial state $\xi$, we have $q(S_A, O, \xi) = q(A, O, \xi)$.  In particular, $L(S_A, O, \xi) = L(A, O, \xi)$.
\item Also, $W(S_A, O, \xi) \le T(A) \|\xi\|^2$.
\item Finally, in the QRAG model, we have $T(S_A) = \OO(\timeR)$ as defined in~\rf{eqn:timeR}.
\end{itemize}
\end{thm}

\pfstart
The proof parallels that of \rf{thm:program->transducerCircuit} using \rf{prp:sequentialParallel} instead of \rf{prp:sequentialsequential}, but since this special case might be of interest, we give an explicit construction here, slightly simplifying it along the way.

It is convenient to assume a different but essentially equivalent form of a quantum algorithm.
Namely, we assume that the algorithm $A$ is given as
\begin{equation}
\label{eqn:program->alternative}
A(O) = G_{m-1}\tO^{b_{m-1}}G_{m-2}\cdots \tO^{b_2}G_1\tO^{b_1}G_0,
\end{equation}
where each $G_i$ is a gate (not a query), $\tO$ is the query operator~\rf{eqn:query}, and each $b_t$ is a bit that indicates whether there is a query before the $t$-th gate ($b_0$ is always 0).
A program like in~\rf{eqn:program} can be transformed into~\rf{eqn:program->alternative} with $m\le T+1$ adding identity $G_t$ if necessary.
We assume we have QRAM access to an array specifying the gates $G_t$ as well as to the array of $b_t$.

Let the algorithm $A$ go through the following sequence of states on the initial state $\xi$ and the oracle $O$:
\begin{equation}
\label{eqn:program->sequenceOfStates}
\xi = \psi_0 \maps{G_0} \psi_1 \maps{G_1\tO^{b_1}} \psi_2 \maps{G_2\tO^{b_2}}\cdots \maps {G_{m-1}\tO^{b_{m-1}}} \psi_{m} = \tau.
\end{equation}
At high level, the action of $S_A$ is
\begin{equation}
\label{eqn:program->Transducer}
\xi \oplus v = 
\sum_{t=0}^{m-1} \ket |t> \ket |\psi_t> 
\maps{\tO^{b_t}}
\sum_{t=0}^{m-1} \ket |t> \ketA |\tO^{b_t}\psi_t>
\maps{{G_t}}
\sum_{t=0}^{m-1} \ket |t> \ket |\psi_{t+1}>
\longmapsto
\sum_{t=1}^{m} \ket |t> \ket |\psi_{t}>
 = \tau \oplus v,
\end{equation}
where, we first apply the input oracle conditioned on $b_t$, then $G_t$ conditioned on $t$ using~\rf{thm:select}, and then increment $t$ by 1 modulo $m$.

\mycutecommand{\psiw}{\psi^\circ}
\mycutecommand{\psiq}{\psi^\bullet}

Recall that we have a decomposition $\cH = \cHw \oplus \cHq$ of the space of $A$, which is indicated by the query register $\cR$.
As mentioned in~\rf{sec:assumptions}, the transducer $S_A$ uses a different query qubit $\ctR$.
We will use notation $\ctH = \ctR\otimes \cH$.
Let $D$ denote the C-NOT on $\ctR$ controlled by $\cR$.
Then, for $\psi = \psiw\oplus \psiq \in\cH$, we have
\[
\ket \ctH|\psi> = \ket \ctR|0> \ket R|0>\ket |\psiw> + \ket \ctR|0> \ket R|1>\ket |\psiq>
\qqand
\ket \ctH|D\psi> = \ket \ctR|0> \ket R|0>\ket |\psiw> + \ket \ctR|1> \ket R|1>\ket |\psiq>.
\]
Let $\cT$ be an $m$-qudit with operations modulo $m$, and $\cP$ be the privacy qubit of $S_A$.

Let us go through the steps of~\rf{eqn:program->Transducer}.
The initial coupling is given by
\[
\xi\oplus v = \ket P|0> \ket T|0>\ket \ctH|\psi_0> + \sum_{t=1}^{m-1} \ket P|1> \ket T|t> \ketA \ctH|D^{b_t}\psi_t>.
\]
First, as required by the canonical form and our assumptions in \rf{sec:assumptions}, we apply the input oracle $O$ conditioned on $\ket \ctR|1>$.
This acts as $\tO$ on $D\psi_t$, but does not change $\psi_t$.
Then, we apply the operation $D$ controlled on $b_t$ (which can be accessed using the QRAM).
After that, we apply C-NOT to $\reg P$ controlled by $\ket T|0>$.
This gives the second state in the sequence~\rf{eqn:program->Transducer}:
\[
\ket P|1> \sum_{t=0}^{m-1} \ket T|t> \ketA \ctH|\tO^{b_t}\psi_t>.
\]
Now, we apply the last two operations from~\rf{eqn:program->Transducer}: $G_t$ conditioned on $\ket T|t>$, and increment of $\cT$ by 1 modulo $m$.
This gives
\[
\ket P|1> \sum_{t=1}^{m} \ket T|t> \ket \ctH|\psi_t>.
\]
Now we apply the operation $D$ controlled on $b_t$ and apply C-NOT to the register $\reg P$ controlled by $\ket T|0> = \ket T|m>$.
This gives the final state
\[
\ket P|0> \ket T|0>\ket \ctH|\psi_{m}> + \sum_{t=1}^{m-1} \ket P|1> \ket T|t> \ketA \ctH|D^{b_t}\psi_t>
=
\tau\oplus v.
\]

Since $\|\psi_t\| = \|\xi\|$, we have that
$
W(S_A, O, \xi) = (m-1) \|\xi\|^2 \le T(A) \|\xi\|^2.
$
It is clear that $q(S_A, O, \xi) = q(A, O, \xi)$, and the time complexity of $S_A$ is $\OO(\timeR)$ thanks in particular to \rf{thm:select}.
\pfend

In applications like in \rf{sec:introComposition}, we often have to apply this construction in parallel.
The following easy modification of the proof takes care of that.

\begin{prp}
\label{prp:program->transducerParallel}
Let $A_1,\dots,A_n$ be quantum programs in some space $\cH$, all using the same oracle $O$.
Assume we have QRAM access to their joint description.
Then, the direct sum $\bigoplus_{i=1}^n S_{A_i}$ of transducers defined in \rf{thm:program->transducerQRAG} can be implemented in time $\OO(\timeR)$.
\end{prp}

By the joint access, we mean that we have access to the list $m_1,\dots,m_n$ of the parameters $m$ in~\rf{eqn:program->alternative}, as well as access to $b_t$ and $G_t$ of $A_i$ as a double array with indices $i$ and $t$.

\pfstart[Proof of \rf{prp:program->transducerParallel}]
We execute the transducers $S_{A_i}$ in parallel using a register $\cJ$ to store the value of $i$.
Accessing $b_t$ and $G_t$ now is double-indexed by $i$ and $t$, and by our general assumption of~\rf{eqn:timeR} it takes time $\OO(\timeR)$ to access them.
For incrementation of $\cT$ modulo $m$, we use the array containing $m_i$.
Other than that, it is a standard word-sided operation and takes time $\OO(\timeR)$.
\pfend

The consequences of this result were already considered in \rf{sec:overview}: Theorems~\ref{thm:introQueryCompression}, \ref{thm:introCompositionQRAG}, and~\ref{thm:introCompositionTree}.

\section{Example III: Iterated Functions}
\label{sec:iterated}

In this section, we give a more detailed proof of \rf{thm:introIterated}.
We will use notation of \rf{sec:function} for the transducer $S_f$ built from the adversary bound.
However, we assume a more general variant of a canonical transducer $S_f$ for evaluation of $f\colon [q]^n\to[q]$.
Its public space is $\bC^q$, its admissible space is spanned by $\ket |0>$, and $\ket |0> \transduce{S_f(\bi{O_x})} \ket |f(x)>$ for all $x\in [q]^n$.
We assume the initial coupling of $S_f$ on $\bi{O_x}$ is given by
\begin{equation}
\label{eqn:composed_vx}
\ket|0>\oplus v_x
=
\ket P|0>\ket R|0> \ket |0> + 
\ket P|1>\ket R|0> \ket |\vw_x> +
\ket P|1> \sum_{i\in [n]} \ket R|i> \skB[ \ket B|\uparrow>\ket Q|0> \ketA W|v^\uparrow_{x,i}> + \ket B|\downarrow>\ket Q|x_i>\ketA W|v^\downarrow_{x,i}> ]
\end{equation}
for some vectors $\vw_x$, $v^\uparrow_{x,i}$, and $v^\downarrow_{x,i}$, where the registers are as in~\rf{eqn:functionWitness}.
The difference between~\rf{eqn:composed_vx} and~\rf{eqn:functionWitness}, however, is addition of the term $\vw_x$ that is not processed by the input oracle.
Note that~\rf{eqn:composed_vx} gives the general form of a canonical transducer that executes the input oracle on the admissible subspace: the constituent $O_{x,i}$ from~\rf{eqn:function_inputConstituent} on $\ket|0>$ and $O_{x,i}^*$ on $\ket|x_i>$.

Recall the definition of the composed function $f\circ g$  from~\rf{eqn:introComposedFunction}:
\begin{equation}
\label{eqn:composedFunction}
\begin{aligned}
\sS[f\circ g]&(z_{1,1}, \dots,z_{1,m},\;\; z_{2,1},\dots,z_{2,m},\;\;\dots\dots,\;\;z_{n,1},\dots,z_{n,m})\\
&= f\sA[
g(z_{1,1}, \dots,z_{1,m}), 
g(z_{2,1}, \dots,z_{2,m}),
\dots,
g(z_{n,1}, \dots,z_{n,m})].
\end{aligned}
\end{equation}
We define
\begin{equation}
\label{eqn:composed_y}
\yy_i = (z_{i,1}, \dots,z_{i,m})
\qqand
x = \sA[g(\yy_1), g(\yy_2),\dots g(\yy_n)]
\end{equation}
so that $f(x) = \sS[f\circ g](z)$.
Recall also notation $W_x(S) = W\sA[S, \bi{O_x}, \ket|0>]$ and $W(S) = \max_x W_x(S)$ from~\rf{eqn:function_Wx}.

\begin{prp}
\label{prp:composedFunction}
Let $S_f$ and $S_g$ be canonical transducers for the functions $f$ and $g$ of the form  described above.
Then, there exists a canonical transducer $S_{f\circ g}$ for the composed function $f\circ g$ from~\rf{eqn:composedFunction} that has the same form, and such that
\begin{equation}
\label{eqn:composedFunctionTransduction}
W_z(S_{f\circ g}) = W_x (S_f) + \sum_{i=1}^n L^{(i)}_x(S_f)\cdot W_{\yy_i}(S_g) \le W(S_f) + L(S_f)\cdot W(S_g)
\end{equation}
and
\begin{equation}
\label{eqn:composedFunctionQuery}
L_z(S_{f\circ g}) = \sum_{i=1}^n L^{(i)}_x(S_f)\cdot L_{\yy_i}(S_g) \le L(S_f)\cdot L(S_g),
\end{equation}
for every $z\in[q]^{nm}$.
Time complexity satisfies $T(S_{f\circ g}) = \TC(S_f) + 2\TC(S_g)$.
\end{prp}

\pfstart
As the first step, we create a bidirectional version $\bi{S_g}$ of $S_g$.
Its public space is $\reg B\otimes \reg Q$, and its admissible space on the input oracle $\bi{O_y}$ is spanned by $\ket B|\uparrow>\ket Q|0>$ and $\ket B|\downarrow> \ket Q|g(y)>$.
Its corresponding transduction action is 
\[
\ket B|\uparrow>\ket Q|0> \transduce{} \ket B|\uparrow>\ket Q|g(y)>
\qqand
\ket B|\downarrow>\ket Q|g(y)> \transduce{} \ket B|\downarrow>\ket Q|0>.
\]
Its time complexity is $2\TC(S_g)$ and, in notation at the end of~\rf{sec:introCanonical}:
\[
\Wmax\sA[\bi{S_g}, \bi{O_y}] = W_y(S_g)
\qqand
\Lmax\sA[\bi{S_g}, \bi{O_y}] = L_y(S_g).
\]

We obtain $\bi{S_g}$ as follows.
By swapping $\uparrow$ and $\downarrow$ if necessary, we may identify $\bi{O_y}^*$ and $\bi{O_y}$.
Then from \rf{prp:inverse}, we obtain a transducer $S^{-1}_g$ with transduction action $\ket |g(y)> \transduce{} \ket |0>$ on $\bi{O_y}$, whose complexity is identical to $S_g$.
We may assume $S_g$ and $S^{-1}_g$ are aligned.
We get $\bi{S_g}$ as $S_g\oplus S_g^{-1}$, where the direct sum is done via the register $\reg B$.
By the definition of $\Wmax$, we have that for some unit vector $(\alpha, \beta)\in\bC^2$:
\begin{align*}
\Wmax (\bi{S_g}, \bi{O_y}) 
&= W\sB[\bi{S_g}, \bi{O_y}, \alpha\ket B|\uparrow>\ket Q|0> + \beta \ket B|\downarrow>\ket Q|g(y)>]\\
&= |\alpha|^2 W\sA[S_g, \bi{O_y}, \ket |0>] + |\beta|^2 W\sA[S_g^{-1}, \bi{O_y}, \ket |g(y)> ]
= W_y(S_g),
\end{align*}
where we used Propositions~\ref{prp:parallel} and~\ref{prp:inverse}, as well as~\rf{eqn:rescaling}.
Query complexity derivation is similar.

As the next step, we construct the transducer $I_n\otimes \bi{S_g}$, where $I_n$ acts on the span of $\ket R|i>$ with $i>0$.
Moreover, as in \rf{rem:directSumDifferentOracle}, we assume the input oracle uses the register $\reg R$.
Therefore, the transduction action of $I_n\otimes \bi{S_g}$ on the input oracle 
\[
\bi{O_z} = \bigoplus_{i=1}^n \bi{O_{\yy_i}}
\]
is identical to the action of $\bi{O_x}$ on its admissible subspace.

Finally, we get $S_{f\circ g}$ as $S_f\circ \sA[I_n\otimes \bi{S_g}]$.
The time complexity estimate follows from \rf{prp:functional}.
For the transduction complexity, we obtain from~\rf{eqn:compositionTransductionMultipleUpper}, using that $S_f$ makes only admissible queries to $\bi{O_x}$ and does not have direct access to $\bi{O_z}$:
\begin{align*}
W_z(S_{f\circ g}) &= W\sA[S_{f\circ g}, \bi{O_z}, \ket |0>] \\
&\le W\sA[S_f, \bi{O_x}, \ket |0>] + \sum_{i=1}^n \Wmax\sA[ \bi{S_g}, \bi{O_{\yy_i}} ] L^{(i)} (S_f, \bi{O_x}, \ket |0>) \\
&= W_x (S_f) + \sum_{i=1}^n W_{\yy_i}(S_g)\cdot L^{(i)}_x(S_f).
\end{align*}
The last inequality in~\rf{eqn:composedFunctionTransduction} follows from $L_x(S_f) = \sum_i L^{(i)}_x(S_f)$.
The estimate~\rf{eqn:composedFunctionQuery} is similar, where we again use that $S_f$ does not make direct queries to $\bi{O_z}$.
Finally, since $\bi{S_g}$ makes only admissible queries to $O_y$, we get that $S_{f\circ g}$ makes only admissible queries to $O_z$.
\pfend

Now we can prove a more detailed version of \rf{thm:introIterated}.

\begin{thm}
\label{thm:iteratedFunctions}
Let $S_f$ be a canonical transducer for a function $f$ of the form as in~\rf{eqn:composed_vx}, and such that $L = L(S_f) = 1 + \Omega(1)$.  Let $W = W(S_f)$ and $T=T(S_f)$, assuming the circuit model.
Then, for each $d$, there exists a bounded-error quantum algorithm evaluating the iterated function $f^{(d)}$ in Monte Carlo query complexity $\OO(L^d)$ and time complexity $\OO\sA[d\cdot TW L^{d-1}] = \OO_f(d\cdot L^d)$ in the circuit model.
\end{thm}

\def\Sf#1{S_{f^{(#1)}}}
\pfstart
Define $\Sf d$ as $S_f$ composed with itself $d$ times using \rf{prp:composedFunction}.
By induction, we have that $L\sA[\Sf d] \le L^d$, and
\[
W\sA[\Sf d] \le \sA[1 + L + \cdots + L^{d-1}] W  = \OO\sA[L^{d-1}W].
\]
Similarly, its time complexity is at most $2d\cdot \TC(S_f) = \OO(d\cdot T)$.
The theorem follows from \rf{thm:optimalImplementation}.
\pfend

The time complexity of the previous theorem can be slightly improved.
Let $s$ denote the initial space complexity of the transducer $S_f$, i.e, the number of qubits used to encode the initial coupling $\ket |0>\oplus v_x$.
It can be much smaller than the time complexity of $S_f$, as well as its space complexity, which is the total number of qubits used by $S_f$.

\begin{prp}
\label{prp:iteratedFunctionsImproved}
The time complexity of \rf{thm:iteratedFunctions} can be improved to $\OO\sA[(sd+T)W L^{d-1}]$, where $s$ is the initial space complexity of the transducer $S_f$.
\end{prp}

We only sketch the proof, as it does not improve the overall asymptotic $\OO_f(d\cdot L^d)$.

\pfstart[Proof sketch of \rf{prp:iteratedFunctionsImproved}]
By studying the proof of \rf{thm:iteratedFunctions}, we can observe that 
the work unitary of $\Sf d$ consists of repeated applications of $\bi{S_f}$ to different registers.
We can treat $\bi{S_f}$ as an oracle.
Then, the work unitary of $\Sf d$ becomes a non-canonical transducer with the oracle $\bi{S_f}$.
It is non-aligned, but we can make it aligned in additional time $s\cdot d$ by switching the registers before each execution of $\bi{S_f}$.
We can convert it into the canonical form using \rf{prp:canoning}, which 
increases the transduction complexity by at most a constant factor.
However, now it takes only one execution of $\bi{S_f}$ and $\OO(sd)$ other operations to implement the transducer.
The statement again follows from \rf{thm:optimalImplementation}.
\pfend

\section{Perturbed Transducers}
\label{sec:perturbedTransducers}

It is quite common in quantum algorithm to use subroutines that impose some error.
The total correctness of the algorithm then follows from a variant of \rf{lem:surgery}, assuming that the error of each constituent is small enough.

In most cases in this paper, like in Sections~\ref{sec:walks} and~\ref{sec:function}, the transduction action is the exact implementation of the required transformation.
However, it is not always feasible, and it makes sense to study what happens if a transducer satisfies the condition~\rf{eqn:transduce} only approximately.

It is worth noting that small errors in a transducer can result in large errors in the corresponding transduction action.
In other words, it is possible that $\xi \oplus v \maps{S} \tau' \oplus v'$ with $v$ close to $v'$, but $\xi\transduce{S} \tau$ with $\tau$ being far from $\tau'$.
For example, let both $\cH$ and $\cL$ be 1-dimensional, and assume $S$ acts as
\[
S\colon 
\begin{pmatrix}
0\\ v
\end{pmatrix}
\mapsto
\begin{pmatrix}
a\\\sqrt{v^2-a^2}
\end{pmatrix}
\]
for some positive real $v$ and $a$.
Observe that
\[
v - \sqrt{v^2-a^2} = \frac{a^2}{v + \sqrt{v^2-a^2}}
\]
can be very small for large $v$, even if $a$ is substantial.
Thus, on the basis of $v\approx v - \sqrt{v^2-a^2}$ we may be tempted to assume that $S$ approximately transduces $0$ into $a$, but this is very far from the true transduction action $0\transduce{S}0$.

\mycutecommand{\tv}{\widetilde v}
\mycutecommand{\ttau}{\widetilde\tau}

\subsection{Definition}
\label{sec:perturbedDefinition}
In order to solve this issue, we incorporate a perturbation, in the sense of \rf{lem:surgery}, into the transducer $S$.
Let $S$ be a unitary in $\cH\oplus \cL$ that maps
\[
\tS\colon \xi\oplus v \mapsto \ttau \oplus \tv .
\]
We assume its idealised version maps
\[
S\colon \xi \oplus v \mapsto \tau \oplus v
\]
with the perturbation
\[
\delta(S, \xi) = \normA| (\tau\oplus v) - (\ttau \oplus \tv) |.
\]

We call $S$ the idealised or perturbed transducer, and $\tS$ the approximate transducer.
We say that $S$ transduces $\xi$ into $\tau$, and we keep notation $\xi\transduce{S} \tau$, $v(S,\xi) = v$, $W(S, \xi) = \|v\|^2$, and $\tau(S,\xi) = \tau$.

This also extends to the canonical form of transducer in \rf{sec:canonicalDefinition}, where we decompose $v = \vw\oplus \vq$, and let $q(S,O,\xi) = \vq$ and $L(S,O,\xi) = \|\vq\|^2$.
As before, we write $\delta(S, O, \xi)$ instead of $\delta(S(O),\xi)$, and similarly for other pieces of notation.

As for usual quantum algorithms, we can use approximate transducers instead of the idealised ones, as long as we carefully keep track of the perturbations.
In the remaining part of this section, we will briefly study the main results of this paper under the perturbation lenses.

\subsection{Implementation}

We have the following variant of \rf{thm:pumping}.

\begin{thm}
\label{thm:pumpingApproximate}
Under the assumptions of \rf{sec:perturbedDefinition}, for every positive integer $K$, there exists a quantum algorithm that transforms $\xi$ into $\tau'$ such that
\[
\normA|\tau' - \tau(S,\xi)| \le 2 \sqrt{\frac{W(S,\xi)}{K}} + \sqrt{K}\delta(S,\xi)
\]
for every $\tS\colon \cH\oplus \cL\to\cH\oplus \cL$, perturbed version $S$, and initial state $\xi\in\cH$.
The algorithm conditionally executes $\tS$ as a black box $K$ times, and uses $\OO(K)$ other elementary operations.
\end{thm}

\pfstart
The algorithm is identical to that of \rf{thm:pumping}.
The analysis is similar with the only difference that, on Step~2(a), in order to obtain the mapping
\[
\frac1{\sqrt K} \ket T|t> \ket P |0> \ket H|\xi> + \frac1{\sqrt K} \ket T|t>\ket P|1>\ket L|v>
\;\longmapsto\;
\frac1{\sqrt K} \ket T|t> \ket P |0> \ket H|\tau> + \frac1{\sqrt K} \ket T|t>\ket P|1>\ket L|v>
\]
as in \rf{eqn:pumpingOneStep}, we introduce a perturbation of size $\delta(S, \xi)/\sqrt{K}$.
By \rf{lem:surgery}, the total perturbation is then
\[
2 \frac{\|v\|}{\sqrt K} + K\cdot \frac{\delta(S,\xi)}{\sqrt{K}},
\]
which gives the required estimate.
\pfend

Again, this theorem can be reformulated as follows.
\begin{cor}
\label{cor:approximatePumping}
For all $W, \eps>0$, there exists a quantum algorithm that executes $\tS$ as a black box $K = \OO(1+W/\eps^2)$ times, uses $\OO(K)$ other elementary operations, and $\eps$-approximately transforms $\xi$ into $\tau(S, \xi)$ for all $S$, $\tS$, and the initial state $\xi$ that satisfy $W(S,\xi)\le W$ and $\delta(S, \xi)\le \frac{\eps}{2\sqrt{K}}$.
\end{cor}

We get a version of \rf{thm:optimalImplementation} in a similar fashion.

\begin{prp}
\label{prp:approximateImplementation}
\rf{thm:optimalImplementation} works assuming a perturbed transducer $S$.  The estimate is
\[
\normA|\tau' - \tau(S, O, \xi)| \le \frac 2{\sqrt K} \sqrt{W(S, O, \xi) + \sum_{i=1}^r\s[\frac{K}{K^{(i)}} -1] L^{(i)}(S, O, \xi) } + \sqrt K \delta(S, O, \xi).
\]
\rf{cor:optimalImplementation} also works assuming a perturbed transducer $S$ under the additional assumption of $\delta(S, O, \xi) \le \frac\eps{2\sqrt K}$.
\end{prp}

\pfstart
The proof is analogous to \rf{thm:pumpingApproximate}:  we use perturbation of size
$\delta(S, O, \xi)/\sqrt{K}$ to get from~\rf{eqn:optimal1} to~\rf{eqn:optimal2}.
\pfend

\subsection{Composition}
The composition of perturbed transducers exactly follows the corresponding constructions of \rf{sec:properties}, where we use \rf{lem:surgery} to evaluate the total perturbation.
This gives us the following proposition.

\begin{prp}
\label{prp:perturbationComposition}
Propositions~\ref{prp:parallel}, \ref{prp:sequentialsequential}, \ref{prp:sequentialParallel} and~\ref{prp:functional} work assuming perturbed transducers.
We get the following estimates:
\begin{equation}
\label{eqn:perturbationLinear}
\delta(S, O, c\xi) = |c| \delta (S, O, \xi),
\end{equation}
\begin{equation}
\label{eqn:perturbationParallel}
\delta(S, O, \xi) \le \sqrt{\sum_{i=1}^m \delta(S_i, O, \xi_i)^2}
\end{equation}
for~\rf{prp:parallel}, and
\begin{equation}
\label{eqn:perturbationSequential}
\delta(S, O, \xi) \le \sum_{t=1}^m \delta(S_t, O, \psi_t)
\end{equation}
for Propositions~\ref{prp:sequentialsequential} and \ref{prp:sequentialParallel}.
For \rf{prp:functional}, we have
\begin{equation}
\label{eqn:perturbationFunctional}
\delta(S_A\circ S_B, O, \xi)
\le 
\delta\sA[S_A, O\oplus O',\xi] + \delta\sA[S_B, O, q^{(1)} (S_A, O\oplus O', \xi) ].
\end{equation}
\end{prp}

\pfstart
Eq.~\rf{eqn:perturbationLinear} follows by linearity.
All the remaining estimates follow from the corresponding proofs in \rf{sec:properties} replacing each transducer with its approximate version and using \rf{lem:surgery}.
The estimate~\rf{eqn:perturbationParallel} follows from the observation that perturbations in terms of the direct sum act on orthogonal subspaces.
\pfend

\section{Purifiers}
\label{sec:purifiers}

In this section, we formulate and prove the formal version of \rf{thm:introPurifier}, as well as draw some consequences of it for composition of bounded-error algorithms.

\subsection{Boolean Case}

Recall our settings from \rf{sec:introPurifier}.
The input oracle performs the transformation
\begin{equation}
\label{eqn:purifierInput}
O_\psi\colon \ket M |0> \mapsto \ket M|\psi> = \ket B|0>\ket N|\psi_0> + \ket B|1> \ket N|\psi_1>
\end{equation}
for some unit vector $\psi$ in some space $\cM = \cB\otimes \cN$ with $\cB = \bC^2$, and it is promised that there exist constants $0\le c-d < c + d \le 1$ such that
\begin{equation}
\label{eqn:purifierCases}
\text{either \qquad
$\|\psi_1\|^2 \le c-d$
\qquad  or \qquad 
$\|\psi_1\|^2 \ge c+d$,
}
\end{equation}
which corresponds to $f(\psi)=0$ and $f(\psi)=1$, respectively.
Let
\begin{equation}
\label{eqn:purifierd}
\mu = \sqrt{(1-c)^2 - d^2} + \sqrt{c^2 - d^2} < 1.
\end{equation}
We say the input oracle $O_\psi$ is \emph{admissible} if it satisfies~\rf{eqn:purifierInput} and~\rf{eqn:purifierCases}.

\mycutecommand{\deltapur}{\delta_{\mathrm{pur}}}

\begin{thm}
\label{thm:purifierBoolean}
Let $D$ be a positive integer.
In the above assumptions, there exists a perturbed transducer $\purifier$ on the 1-dimensional public space and with bidirectional access to $O_\psi$ which satisfies the following conditions:
\begin{itemize}\itemsep=0pt
\item 
It transduces $\ket |0>$ into $(-1)^{f(\psi)} \ket |0>$ for all admissible input oracles $O_\psi$.
In particular, every vector in its public space is admissible.
\item On every admissible input oracle and normalised input state, its perturbation is at most $\deltapur = 2 \mu^{D-1}$.
\item Its transduction and query complexities, $\Wmax(\purifier)$ and $\Lmax(\purifier)$, are $\OO\sA[1/(1-\mu)] = \OO(1)$.
\item It executes the input oracle on the admissible subspace only: $O_\psi$ on $\ket M|0>$ and $O_\psi^*$ on $\ket M|\psi>$.
\end{itemize}
In the circuit model, the purifier can be implemented in time $\OO(D\cdot s)$, where $s$ is the number of qubits used in $\cM$.
In the QRAG model and assuming the RA input oracle, the purifier can be implemented in time $\OO(\timeR)$.
\end{thm}

Here we used $\Wmax(\purifier)$ to denote maximal $W(\purifier, \bi{O_\psi}, \ket|0>)$ over all admissible input oracles $O_\psi$.
$\Lmax(\purifier)$ is defined similarly.

\pfstart
The outline of the proof was already given in \rf{sec:introPurifier}.
We will assume that $D$ is even for concreteness, the case of odd $D$ being analogous.
Recall the parameters
\[
a = \sqrt[4]{\frac{1-c+d}{1-c-d}}
\qqand
b = \sqrt[4]{\frac{c+d}{c-d}}.
\]
and the following vector in $\cM$:
\[
\wpsi = 
\begin{cases}
 \frac 1a \ket |0> \ket |\psi_0> + b \ket|1>\ket|\psi_1>,& \text{if $f(\psi)=0$};\\
a \ket |0> \ket |\psi_0> + \frac 1b \ket|1>\ket|\psi_1>,& \text{if $f(\psi)=1$}.\\
\end{cases}
\]
We have the following important estimate.  If $f(\psi)=0$:
\begin{equation}
\label{eqn:purifierPsiSize1}
\|\wpsi\|^2 
= \frac1{a^2} \|\psi_0\|^2 + b^2 \|\psi_1\|^2
\le \frac1{a^2} (1-c+d) + b^2 (c-d)=
\mu,
\end{equation}
where we used that $1/a^2 < 1 < b^2$.
Similarly, for $f(\psi)=1$:
\begin{equation}
\label{eqn:purifierPsiSize2}
\|\wpsi\|^2 
= {a^2} \|\psi_0\|^2 + \frac1{b^2} \|\psi_1\|^2
\le {a^2} (1-c-d) + \frac1{b^2} (c+d)=
\mu.
\end{equation}

We describe the transducer as non-canonical, and we will turn it into the canonical form later.
The space of the transducer is $\cD\otimes \cM^{\otimes D-1}$, where $\cD$ is a $D$-qudit.
The one-dimensional public space is spanned by $\xi = \ket D|0> \ket |0>^{\otimes D-1}$.
The initial coupling is given by
\begin{equation}
\label{eqn:purifier xi+v}
\xi \oplus v = \sum_{i=0}^{D-1} (-1)^{i\cdot f(\psi)} \ket D|i> 
\ketO|\wpsi>^{\otimes i} \ket|0>^{\otimes D-i-1}.
\end{equation}
The transduction complexity is 
\begin{equation}
\label{eqn:zetaNorm}
\|v\|^2 \le \|\xi\oplus v\|^2 \le \sum_{i=0}^{\infty} \normA|\wpsi|^{2i} = \frac{1}{1-\normA|\wpsi|^2} \le \frac1{1-\mu}.
\end{equation}

\myfigure{\label{fig:purifier}}
{
A purifier is a multidimensional quantum walk on the line graph seen above.
The edge between the vertices $i$ and $i+1$ corresponds to the subspace $\ket D|i>\otimes \cM^{\otimes D-1}$ as indicated by the state below the edge.
The local reflections on the vertices $0$ and $D$ are identities.
The local reflection at the vertex $i=1,\dots, D-1$ acts on the subspace $\spn\{\ket D|i-1>, \ket D|i>\}\otimes \cM^{\otimes D-1}$.
The expressions above the edges give the initial coupling from~\rf{eqn:purifier xi+v}.
}
{
\[
\begin{tikzpicture}[auto]
\node (0) at (0,0) [circle, draw] {$0$};
\node (1) at (2,0) [circle, draw] {$1$};
\node (2) at (5,0) [circle, draw] {$2$};
\node (3) at (7.5,0) [circle, draw] {$3$};
\node (D-1) at (11,0) [ellipse, draw] {$D-1$};
\node (D) at (14.5, 0) [ellipse, draw] {$D$};
\draw (0) to node [midway, above, blue] {$\scriptstyle\ket |0>^{\otimes D-1} $} 
             node [midway, below, gray ] {$\ket D|0>$} 
(1);
\draw (1) to node [midway, above, blue] {$\scriptstyle(-1)^{f(\psi)}\ketO |\wpsi>\ket |0>^{\otimes D-2} $} 
             node [midway, below, gray ] {$\ket D|1>$} 
(2);
\draw (2) to node [midway, above, blue] {$\scriptstyle\ketO |\wpsi>^{\otimes 2}\ket |0>^{\otimes D-3} $} 
             node [midway, below, gray ] {$\ket D|2>$} 
(3);
\draw (3) to (8.5,0);
\node (dots) at (9,0) {$\cdots$};
\draw(D-1) to (9.5,0);
\draw (D-1) to node [midway, above, blue] {$\scriptstyle (-1)^{f(\psi)}\ketO |\wpsi>^{\otimes D-1} $} 
             node [midway, below, gray ] {$\ket D|D-1>$} 
(D);
\end{tikzpicture}
\]
}

Let us now describe the action of the transducer.
The transducer is a multidimensional quantum walk on the line graph, see \rf{fig:purifier}.
It is a product of two reflections $R_1$ and $R_2$.
The reflection $R_1$ is the product of the local reflections on the odd vertices, $i=1,3,5,\dots, D-1$.
The reflection $R_2$ is the product of the local reflections on the even vertices, $i=2,4,\dots, D-2$.

The local reflection for the vertex $i=1,\dots,D-1$ is as follows.
It acts in $\spn\{\ket D|i-1>, \ket D|i>\}\otimes \cM^{\otimes D-1}$.
Define a qubit $\cA$ whose value $0$ corresponds to $\ket D|i-1>$ and $1$ to $\ket D|i>$.
If $i$ is odd, this could be the least significant qubit of $\cD$.
Let $\cM^{(i)} = \cB^{(i)}\otimes \cN^{(i)}$ be the $i$-th multiplier in the tensor product $\cM^{\otimes D-1}$.
The reflection is as follows:

\begin{enumerate}\itemsep=0pt
\item Execute the input oracle $O_\psi$ on $\cM^{(i)}$ conditioned on $\ket A|0> = \ket D|i-1>$.
\item Execute a two-qubit unitary on $\cA\otimes \cB^{(i)}$, which is the reflection about the span of the states
\[
\sA[a\ket A|0> + \ket A|1> ]\ket B|0>
\qqand
\sA[\ket A|0> + b \ket A|1> ]\ket B|1>.
\]
\item Execute the inverse oracle $O^*_\psi$ on $\cM^{(i)}$ conditioned on $\ket A|0> = \ket D|i-1>$.
\end{enumerate}

\begin{clm}
\label{clm:purifier}
The local reflection for the vertex $i$ multiplies the corresponding part of the state in~\rf{eqn:purifier xi+v}
\[
\ket D|i-1> \ketO|\wpsi>^{\otimes i-1} \ket|0>^{\otimes D-i} +
(-1)^{f(\psi)}\ket D|i> \ketO|\wpsi>^{\otimes i} \ket|0>^{\otimes D-i-1}
\]
by the phase $(-1)^{f(\psi)}$.
\end{clm}

\pfstart
After application of the oracle in Step 1, we get the state
\[
\ketO|\wpsi>^{\otimes i-1}  
\otimes 
\skB[\ket D|i-1> \ket M |\psi> + (-1)^{f(\psi)} \ket D|i> \ketA M|\wpsi>] 
\otimes 
\ket|0>^{\otimes D-i-1}.
\]
The local reflection acts on the state in the square brackets, which can be rewritten as
\begin{equation}
\label{eqn:purifierSub1}
\sB[\ket A|0> + \frac 1a \ket A|1>]\ket B|0> \ket N|\psi_0> +
\sB[\ket A|0> + b \ket A|1> ]\ket B|1> \ket N|\psi_1>
\end{equation}
if $f(\psi)=0$, and
\begin{equation}
\label{eqn:purifierSub2}
\sB[\ket A|0> - a \ket A|1> ]\ket B|0> \ket N|\psi_0> +
\sB[\ket A|0> - \frac 1b \ket A|1> ]\ket B|1> \ket N|\psi_1>
\end{equation}
if $f(\psi)=1$.
It is easy to see that the operation on Step 2 does not change the state in~\rf{eqn:purifierSub1} and negates the one in~\rf{eqn:purifierSub2}, from which the claim follows.
\pfend

From the claim, it immediately follows that, if $f(\psi)=0$, the transducer does not change the state~\rf{eqn:purifier xi+v}.
Therefore, in this case $\ket |0> \transduce{\purifier} \ket |0>$.

On the other hand, if $f(\psi)=1$, then $R_1$ reflects the whole state $\xi\oplus v$, and $R_2$ reflects all the terms in the sum except for $i=0$ and $i=D-1$.  Thus, the final state is
\[
-\xi \oplus v - 2 (-1)^{(D-1)f(\psi)}\ket D|D-1> \ketO|\wpsi>^{\otimes D-1} .
\]
This can be interpreted as transducing $\ket |0>$ into $-\ket |0>$ with a perturbation of size at most $2\mu^{D-1}$ by~\rf{eqn:purifierPsiSize1} and~\rf{eqn:purifierPsiSize2}.

One can see that the local reflection for the vertex $i$ applies the input oracle and its inverse on the part of the state $v$ in the subspace $\ket D|i-1>\otimes \cM^{\otimes D-1}$.
Hence, the query complexity is at most $2\|\xi\oplus v\|^2 \le 2/(1-\mu)$ by~\rf{eqn:zetaNorm}.
The transformation into canonical form, \rf{prp:canoning}, adds the query complexity to the transduction complexity, hence, the latter stays $\OO\sA[1/(1-\mu)]$.

Let us estimate time complexity.  We start with the circuit model.
The queries in transducer of \rf{thm:purifierBoolean} are not aligned, as they are applied to different copies of $\cM$.
In order to make them aligned, as required by \rf{prp:canoning}, the register $\cM^{(i)}$ should be moved to some specific array of qubits shared by all the local reflections.
This takes time $\OO(s)$ per each local reflection.
Step 2 of the local reflection can be implemented in constant time.
Thus, each local reflection takes time $\OO(s)$.
There are $D-1$ local reflections performed.
By \rf{lem:automaton} with all $\phi_i$ being absent, the whole transducer can be implemented in time $\OO(D\cdot s)$.
Transformation into canonical form in \rf{prp:canoning} keeps the time complexity of the transducer essentially the same.

Now consider the QRAG model.
All local reflections in $R_1$ can be performed in parallel, and the same is true for $R_2$.
The first and the third operation in the local reflection are implemented by the RA input oracle.
The second operation can be performed in $\OO(\timeR)$ time by \rf{thm:select}.
\pfend

\subsection{Non-Boolean Case}
This time let $\cM = \cB\otimes \cN$ be a space with $\cB = \bC^p$.
Let $O_\psi$ be an oracle that performs the following state generation:
\begin{equation}
\label{eqn:purifierInput2}
O_\psi\colon \ket M |0> \mapsto \ket M|\psi> 
= 
\sum_{j=0}^{p-1} \ket B|j>\ket N|\psi_j>.
\end{equation}
Let $d>0$ be a constant.
We assume that for every $\psi$ there exists (unique) $f(\psi)\in[p]$ such that
\begin{equation}
\label{eqn:purifierCases2}
\|\psi_{f(\psi)}\|^2 \ge \frac12+d.
\end{equation}
Define
\[
\mu = 2 \sqrt{1/4 - d^2} < 1,
\]
which is the same as in~\rf{eqn:purifierd} for $c=1/2$.
We treat $\cB = \bC^p$ as composed out of $\ell = \log p$ qubits.
We do \emph{not} assume that $\ell = \OO(\timeR)$ here, as functions, in principle, can have very long output.
Again, we call every input oracle $O_\psi$ satisfying~\rf{eqn:purifierInput2} and~\rf{eqn:purifierCases2} admissible.

\begin{thm}
\label{thm:purifierGeneral}
Let $D$ be a positive integer.
Under the above assumptions, there exists a perturbed transducer $\purifier$ on the public space $\bC^p$ and with bidirectional access to $O_\psi$ which satisfies the following conditions:
\begin{itemize}\itemsep=0pt
\item For all $b\in\bool^\ell$,  and admissible $O_\psi$, it transduces $\ket |b>$ into $\ketA|b\oplus f(\psi)>$, where $\oplus$ stands for the bit-wise XOR.
In particular, every initial vector in $\bC^p$ is admissible.
\item On any admissible input oracle and unit initial vector, its perturbation is at most $\deltapur = 2 \mu^{D-1}$.
\item Its query complexity satisfies $\Lmax(\purifier) = \OO\sA[1/(1-\mu)] = \OO(1)$.
\item It executes the input oracle on the admissible subspace only: $O_\psi$ on $\ket M|0>$ and $O_\psi^*$ on $\ket M|\psi>$.
\end{itemize}
For the time and the transduction complexity, we have the following estimates:
\begin{itemize}\itemsep=0pt
\item In the circuit model, $\Wmax(\purifier) = \OO\sA[1/(1-\mu)] = \OO(1)$ and $T(\purifier) = \OO(s\cdot D)$, where $s$ is the number of qubits used in $\cM$.
\item In the QRAG model, assuming the RA input oracle, we have $\Wmax(\purifier) = \OO(\log p/(1-\mu)) = \OO(\log p)$ and $T(\purifier) = \OO(\timeR)$.
\end{itemize}
\end{thm}

Here we use $\Wmax(\purifier)$ to denote maximal $W(\purifier, \bi{O_\psi}, \xi)$ over all admissible input oracles $O_\psi$ and admissible normalised initial states $\xi$.
$\Lmax(\purifier)$ is defined similarly.

\pfstart
We reduce the non-Boolean case to the Boolean case of \rf{thm:purifierBoolean} by encoding the value into the phase and using the Bernstein-Vazirani algorithm~\cite{bernstein:quantumComplexity} to decode it back.

For simplicity of notation, we will assume $p = 2^\ell$ so that $O_\psi$ in~\rf{eqn:purifierInput2} just does not use the extra dimensions.
For $a,b\in [p]$, we denote by $a\odot b\in\{0,1\}$ their inner product when considered as elements of $\bF_2^\ell$.

Denote the public space $\bC^p$ of $\purifier$ by $\cJ$.
We first apply the Hadamard $H^{\otimes \ell}$ to perform the following transformation
\begin{equation}
\label{eqn:purifierGenA}
H^{\otimes \ell}\colon \ket J|b> \longmapsto \frac1{\sqrt{p}} \sum_{i=0}^{p-1} (-1)^{i\odot b} \ket J|i>.
\end{equation}

Consider the following procedure $E$ that evaluates the inner product between $i$ and the output of the oracle into an additional qubit $\cZ$:
\begin{equation*}
E(O_\psi) \colon \ket J|i> \ket M|0>  \ket Z|0>
\maps{O_\psi} 
\ket J|i> \sum_{j=0}^{p-1} \ket B|j>\ket N|\psi_j> \ket Z|0>
\longmapsto
\ket J|i> \sum_{j=0}^{p-1} \ket B|j>  \ket N|\psi_j> \ket Z |i\odot j>.
\end{equation*}
We consider it as a direct sum $E = \bigoplus_{i\in[p]} E^{(i)}$ with
\begin{equation}
\label{eqn:purifierGenB}
E^{(i)}(O_\psi) \colon \ket M|0>\ket Z|0>
\longmapsto
\sum_{j=0}^{p-1} \ket B|j>  \ket N|\psi_j> \ket Z |i\odot j>.
\end{equation}

We convert them into canonical transducers $S_E$ and $S^{(i)}_E$.
We have $\Lmax\sA[S_E^{(i)}] = 1$, where the admissible subspace is $\ket M|0>\ket Z|0>$.
Using Propositions~\ref{prp:inverse} and~\ref{prp:parallel}, we get transducers
\[
\bi{S_E}\sA[\bi{O_\psi}] = S_E(O_\psi) \oplus S_E^{-1}(O^*_\psi)
\qqand
\bi{S_E}^{(i)}\sA[\bi{O_\psi}] = S_E^{(i)}(O_\psi) \oplus (S_E^{(i)})^{-1}(O^*_\psi).
\]
Again,
$\Lmax\sA[\bi{S_E}^{(i)}] = 1$.
We still have $\bi{S_E} = \bigoplus_i \bi{S_E}^{(i)}$ with the help of \rf{rem:directSumDifferentOracle}.
It is also clear that $\bi{S_E}$ only executes the input oracle $\bi{O_\psi}$ on the admissible subspace.

We treat $E^{(i)}(O_\psi)$ as an oracle encoding a Boolean value into the register $\cZ$.
By~\rf{eqn:purifierCases2}, we get that $E^{(i)}(O_\psi)$ evaluates $i\odot f(\psi)$ with bounded error.
Take the purifier $\purifier'$ from \rf{thm:purifierBoolean} with $c=1/2$ and the same values of $d$ and $D$.
This purifier satisfies
\[
\purifier'\sB[\bi{E^{(i)}(O_\psi)}] \colon \ket |0>\transduce{} (-1)^{i\odot f(\psi)}\ket |0> .
\]
Taking direct sum over all $i\in[p]$, we get that a transducer
\begin{equation}
\label{eqn:purifierItimesPurifier}
\sS[I_\cJ\otimes \purifier'] \sB[\bi{E(O_\psi)}]
=
\bigoplus_{i=0}^{p-1} \purifier'\sB[\bi{E^{(i)}(O_\psi)}]
\end{equation}
performs the following transduction:
\begin{equation}
\label{eqn:purifierGen1}
\frac1{\sqrt{p}} \sum_{i=0}^{p-1} (-1)^{i\odot b}\ket J|i>
\transduce{}
\frac1{\sqrt{p}} \sum_{i=0}^{p-1} (-1)^{i\odot b + i\odot f(\psi)}\ket J|i> .
\end{equation}

Finally, we again apply $H^{\otimes \ell}$ to get
\begin{equation}
\label{eqn:purifierGen2}
H^{\otimes \ell}\colon \frac1{\sqrt{p}} \sum_{i=0}^{p-1} (-1)^{i\odot b + i\odot f(\psi)}\ket J|i>
\longmapsto \ket J |b \oplus f(\psi)>.
\end{equation}

Combining~\rf{eqn:purifierGenA}, \rf{eqn:purifierGen1} and~\rf{eqn:purifierGen2}, we see that we can use sequential composition of \rf{prp:sequentialsequential} to get
\[
\purifier = S_{H^{\otimes \ell}} * \sA[(I_\cJ\otimes \purifier')\circ \bi{S_E}] * S_{H^{\otimes \ell}},
\]
where $S_{H^{\otimes \ell}}$ is a transducer from \rf{prp:gates} with transduction action $H^{\otimes \ell}$ and no input oracle.

Let us estimate query complexity on a unit vector $\xi\in\cJ$.
We have the following estimate, where we explain individual lines after the equation.
\begin{align*}
L(\purifier, \bi{O_\psi}, \xi) 
&= L\sA[(I_\cJ\otimes \purifier')\circ \bi{S_E}, \bi{O_\psi}, H^{\otimes \ell}\xi]\\
& = \sum_{i=0}^{p-1} L\sA[\purifier'\circ \bi{S_E}^{(i)}, \bi{O_\psi}, \phi_i]\\
& \le \sum_{i=0}^{p-1} \Lmax\sA[\purifier'] \Lmax\sA[\bi{S_E}^{(i)}] \|\phi_i\|^2\\
& \le \OO(1/(1-\mu)) \sum_{i=0}^{p-1} \|\phi_i\|^2 = \OO(1/(1-\mu)).
\end{align*}
On the first line, we used sequential composition of \rf{prp:sequentialsequential} and that the transducer $S_{H^{\otimes \ell}}$ does not use the input oracle.
On the second line, we decomposed $H^{\otimes \ell}\xi = \bigoplus_i \phi_i$, and used~\rf{eqn:purifierItimesPurifier} and \rf{prp:parallel}.
On the third line, we used functional composition of \rf{prp:functional}, Eq.~\rf{eqn:rescaling}, and that $\purifier'$ does not have direct access to $O_\psi$ and executes its input oracle $E^{(i)}$ on the admissible subspace only.

In a similar way, we have
\begin{align}
W(\purifier, \bi{O_\psi}, \xi) 
&= 2\Wmax(S_{H^{\otimes \ell}}) + W\sA[(I_\cJ\otimes \purifier')\circ \bi{S_E}, \bi{O_\psi}, H^{\otimes \ell}\xi] \notag\\
& = 2\Wmax(S_{H^{\otimes \ell}}) + \sum_{i=0}^{p-1} W\sA[\purifier'\circ \bi{S_E}^{(i)}, \bi{O_\psi}, \phi_i]\notag\\
& \le 2\Wmax(S_{H^{\otimes \ell}}) + 
\sum_{i=0}^{p-1} \skB[{\Wmax\sA[\purifier'] + \Lmax\sA[\purifier'] \Wmax\sA[\bi{S_E}^{(i)}]}] \|\phi_i\|^2 \notag\\
& \le 2\Wmax(S_{H^{\otimes \ell}}) + \Wmax (\purifier') + \Lmax(\purifier')\max_i\Wmax\sA[\bi{S_E}^{(i)}]\notag\\
& \le 2\Wmax(S_{H^{\otimes \ell}}) + \OO\sA[1/(1-\mu)]\max_i\sB[{1+\Wmax\sA[\bi{S_E}^{(i)}]}].
\label{eqn:purifierTransductionComplexity}
\end{align}
For the perturbation, we have using \rf{prp:perturbationComposition}:
\begin{align*}
\delta(\purifier, \bi{O_\psi}, \xi) 
&= \delta\sA[(I_\cJ\otimes \purifier')\circ \bi{S_E}, \bi{O_\psi}, H^{\otimes \ell}\xi]\\
& = \sqrt{\sum_{i=0}^{p-1} \delta\sA[\purifier',  E^{(i)}(O_\psi), \phi_i]^2} 
 \le \sqrt{\sum_{i=0}^{p-1} \sA[ \deltapur \|\phi_i\| ]^2} = \deltapur.\\
\end{align*}
Finally,
\begin{equation}
\label{eqn:purifierTimeComplexity}
T(\purifier) = 2 \TC(S_{H^{\otimes \ell}}) + \TC(\purifier') + \TC(\bi{S_E}) + \OO(1).
\end{equation}

In the circuit model, the transduction complexities of $S_{H^{\otimes \ell}}$ and $\bi{S_E}^{(i)}$ are 0 and 1, respectively, and we get the required estimate from~\rf{eqn:purifierTransductionComplexity}.
Also, both $T(S_{H^{\otimes\ell}})$ and $T(\bi{S_E})$ are $\OO(\log p)$, and $T(\purifier') = \OO(s\cdot D)$.  Since $s\ge \log p$, we get the required estimate from~\rf{eqn:purifierTimeComplexity}.

In the QRAG model, we have both $\Wmax(S_{H^{\otimes \ell}})$ and $\Wmax(\bi{S_E}^{(i)})$ bounded by $\OO(\log p)$, which gives the required estimate on the transduction complexity.
For the time complexity, all the terms in~\rf{eqn:purifierTimeComplexity} are $\OO(\timeR)$, which shows that $T(\purifier) = \OO(\timeR)$.
\pfend

\subsection{Composition of Bounded-Error Algorithms}
\label{sec:purifierComposition}

Purifier can be composed with algorithms evaluating functions with bounded error to reduce the error.
In this section, we mention some examples.

Since we are ignoring the constant factors, we will assume the standard version of the input oracle: $O_x \colon \ket |i>\ket|b> \mapsto \ket|i>\ket|b\oplus x_i>$.
In particular, it is its own inverse, and we use $O_x$ instead of $\bi{O_x}$ everywhere in this section.

Let us again recall the composed function $f\circ g$  from~\rf{eqn:composedFunction}:
\begin{equation}
\label{eqn:composedFunctionCopy}
\begin{aligned}
\sS[f\circ g]&(z_{1,1}, \dots,z_{1,m},\;\; z_{2,1},\dots,z_{2,m},\;\;\dots\dots,\;\;z_{n,1},\dots,z_{n,m})\\
&= f\sA[
g(z_{1,1}, \dots,z_{1,m}), 
g(z_{2,1}, \dots,z_{2,m}),
\dots,
g(z_{n,1}, \dots,z_{n,m})].
\end{aligned}
\end{equation}
and
\begin{equation}
\label{eqn:composed_yCopy}
\yy_i = (z_{i,1}, \dots,z_{i,m})
\qqand
x = \sA[g(\yy_1), g(\yy_2),\dots g(\yy_n)]
\end{equation}
so that $f(x) = \sS[f\circ g](z)$.
The following result is essentially \rf{thm:introCompositionFunctionCircuit}.

\begin{thm}
\label{thm:compositionFunctionCircuit}
Let $A$ and $B$ be quantum algorithms in the circuit model that evaluate functions $f$ and $g$, respectively, with bounded error.
Then, there exists an algorithm in the circuit model that evaluates the function $f\circ g$ with bounded error in time complexity
\begin{equation}
\label{eqn:compositionFunctionCircuit}
\OO(L)\sA[T(A) + T(B) + s\log L]
\end{equation}
where $L$ is the worst-case Las Vegas query complexity of $A$, and $s$ is the space complexity of $B$.
The algorithm makes $\OO\sA[L\cdot Q(B)]$ queries, where $Q(B)$ is the usual Monte Carlo query complexity of $B$.
\end{thm}

\pfstart
First, use \rf{thm:program->transducerCircuit} to get a transducer $S_A$ whose transduction action is identical to the execution of $A$.
Its time complexity $T(S_A) = \OO\sA[T(A)]$ and its transduction and query complexities are bounded by $L$.

The algorithm $B$ on the input oracle $O_y$ evaluates $g(y)$ with bounded error.
We obtain the algorithm $B^{-1}$ with the input oracle $O_y^* = O_y$ whose action is the inverse of $B$.
Combining the two via direct sum, we get the algorithm $\bi{B}(O_y) = B(O_y)\oplus B^{-1}(O_y)$.
Let $\purifier$ be the corresponding purifier from \rf{thm:purifierGeneral}, and $D$ and $\deltapur$ be the parameters therein.
The transduction action of $\purifier$ on the input oracle $\bi{B}(O_y)$ is $\ket|b> \transduce{} \ket |b\oplus g(y)>$.

By~\rf{eqn:composed_yCopy}, we have $O_z = \bigoplus_i O_{\yy_i}$.
Let us denote by $\bi{B}(O_z)$ the algorithm $\sS[I_n\otimes \bi B](O_z) = \bigoplus_i \bi{B} (O_{\yy_i})$.
By \rf{cor:byIdentity} with \rf{rem:directSumDifferentOracle}, the transducer $I_n\otimes\purifier$ on the input oracle $\bi{B}(O_z)$ performs the transduction $\ket |i>\ket|b>\transduce{} \ket |i>\ket |b\oplus g(\yy_i)>$ for every $i\in[n]$.
In other words, its transduction action is $O_x$.

Now consider the transducer 
\[
S = S_A\circ (I_n\otimes \purifier)
\]
with the input oracle $\bi{B}(O_z)$.
By the definition of functional composition, it transduces $\ket |0>$ into $\sS[g\circ f](z)$.
By \rf{eqn:compositionTransductionMultipleUpper}, its transduction complexity is at most
\[
W(S_A, O_x, \ket |0>) + \sum_i L^{(i)} (S_A, O_x, \ket |0>) \cdot \Wmax\sA[\purifier, \bi{B}(O_{\yy_i})]  = \OO(L).
\]

We obtain the required algorithm by using \rf{cor:approximatePumping} with $\eps = \Theta(1)$ on the above transducer $S$.
The transducer $S$ is executed $\OO(L)$ times.
Each execution takes times $\OO\sA[T(A) + sD]$ to execute the transducer and $\OO\sA[T(B)]$ to execute the input oracle.
Therefore, the total time complexity is
\[
\OO(L) \sA[T(A) + T(B) + sD]
\]
and the query complexity is $\OO(L)Q(B)$.

It remains to estimate $D$.
We may assume the error of $A$ is a small enough constant.
By \rf{cor:approximatePumping}, in order to get a bounded-error algorithm from the transducer $S_A\circ (I_n\otimes \purifier)$, we should have
\[
\delta \sA[I_n\otimes \purifier, \bi{B}(O_z), q(S_A, O_x, \ket |0>)] = \OO\sA[1/\sqrt{L}].
\]
for a small enough constant.
Using~\rf{eqn:perturbationLinear} and that $\norm|{q(S_A, O_x, \ket |0>)}| \le \sqrt{L}$, we get that it suffices to have $\deltapur = \OO(1/L)$.
Therefore, we can take $D = \OO(\log L)$, which finishes the proof.
\pfend

Let us now proceed with QRAG case, \rf{thm:introCompositionFunctionQRAG}.
Recall the function from~\rf{eqn:randomComposedFunctionCopy}:
\begin{equation}
\label{eqn:randomComposedFunctionCopy2}
f\sA[
g_1(z_{1,1}, \dots,z_{1,m}), 
g_2(z_{2,1}, \dots,z_{2,m}),
\dots,
g_n(z_{n,1}, \dots,z_{n,m})].
\end{equation}
with notation
\[
\yy_i = (z_{i,1}, \dots,z_{i,m})
\qqand
x = \sA[g_1(\yy_1), g_2(\yy_2),\dots g_n(\yy_n)].
\]

For simplicity, we assume all functions $f$ and $g$ use the same input and output alphabet $q$.
Let $A$ and $B_i$ be quantum algorithms that evaluate $f$ and $g_i$, respectively.
To simplify expressions, we assume that $T(B_i)\ge \log q$.  In other words, we spend at least 1 iteration per bit of the output.
Also, all the algorithms have the same upper bound on permissible error.
We use an approach similar to \rf{sec:iterated} on iterated functions, so that we are able to compose several layer of functions.

As in \rf{sec:function}, we denote $L_x(A) = L(A, O_x, \ket|0>)$ and similarly for other notation.
This time, however the input oracle  $O_x \colon \ket |i>\ket|b> \mapsto \ket|i>\ket|b\oplus x_i>$ is uniquely defined.
As we can make the perturbation of a purifier as small as necessary without increasing complexity, we ignore the perturbations in the following implicitly assuming they are small enough.

\begin{thm}
\label{thm:compositionFunctionQRAG}
Let $S_f$ be a perturbed transducer evaluating the function $f$, and $B_1,\dots, B_n$ be algorithms evaluating the functions $g_1,\dots,g_n$ with bounded-error.
Assuming the QRAG model with RA input oracle, and QRAM access to the description of $B_1,\dots, B_n$, there exists a perturbed transducer $S$ evaluating the function~\rf{eqn:randomComposedFunctionCopy2} with the following parameters.
Its transduction complexity is
\begin{equation}
\label{eqn:compositionFunctionQRAG}
W(S, O_z, \ket|0>) = W(S_f, O_x, \ket |0>) + \sum_{i=1}^n \OO\sA[ L^{(i)}_x(S_f) T(B_i)]
\end{equation}
its query complexity is
\begin{equation}
\label{eqn:compositionFunctionQRAGQuery}
L^{(i,j)} (S, O_z, \ket|0>) = \OO\sA[ L^{(i)}_x(S_f) L^{(j)}_{\yy_i}(B_i)  ]
\end{equation}
and its time complexity is $T(S) = \TC(S_f) + \OO(\timeR)$.
\end{thm}

\pfstart
From \rf{thm:program->transducerQRAG}, for each $i$, we obtain a transducer $S_{B_i}$ whose transduction action is identical to the action of $B_i$.
Its transduction complexity $\Wmax(S_{B_i}) = \OO\sA[T(B_i)]$ and query state is identical to that of $B_i$.
Also, we obtain $\bi{S_{B_i}} = S_{B_i} \oplus S_{B_i}^*$ whose complexity is identical to $S_{B_i}$.
Let $\purifier$ be the corresponding purifier.
We have that the transduction action of $\purifier\circ \bi{S_{B_i}}$ on an input oracle $O_y$ is $\ket|b> \transduce{} \ket |b\oplus g_i(y)>$.

Using direct sum, we get that the transducer 
\[
S_g = (I_n\otimes \purifier)\circ\bigoplus_i \bi{S_{B_i}} = \bigoplus_i \purifier\circ \bi{S_{B_i}}
\]
on the input oracle $O_z = \bigoplus_i O_{\yy_i}$ transduces $\ket |i>\ket|b>\transduce{} \ket |i>\ket |b\oplus g_i(\yy_i)>$ for every $i\in[n]$.
Thus, its transduction action is $O_x$,
and the transducer $S = S_f\circ S_g$ evaluates the function in~\rf{eqn:randomComposedFunctionCopy2}.

Let us estimate its transduction complexity.
First by~\rf{eqn:compositionTransductionUpper}:
\[
\Wmax(\purifier\circ \bi{S_{B_i}}, O_{\yy_i}) \le
\Wmax(\purifier) + \Lmax(\purifier) \Wmax(\bi{S_{B_i}}) 
= \OO(\log p) + \OO(T(B_i)) = \OO(T(B_i))
\]
using our assumption on $T(B_i)\ge\log p$.
Therefore, by~\rf{eqn:compositionTransductionMultipleUpper}:
\begin{align*}
W\sA[S, O_z, \ket|0>] 
&\le W(S_f, O_x, \ket |0>) + \sum_i L^{(i)}_x(S_f)\Wmax\sA[\purifier\circ \bi{S_{B_i}}, O_{\yy_i}] \\
&\le W(S_f, O_x, \ket |0>) + \sum_i L^{(i)}_x(S_f) \OO(T(B_i)) .
\end{align*}

For the query complexity, we use a partial-query variant of~\rf{eqn:compositionLasVegasUpper} and that the purifier only executes the subroutine on the admissible initial states to obtain:
\[
\Lmax^{(j)}\sA[\purifier\circ \bi{S_{B_i}}, O_{\yy_i}]
\le \Lmax(\purifier) L^{(j)}\sA[S_{B_i}, O_{\yy_i}, \ket |0>]
= \OO\sA[L^{(j)}_{\yy_i} (B_i)].
\]
Hence, by~\rf{eqn:compositionLasVegasMultipleUpper}:
\[
L^{(i,j)}\sA[S, O_z, \ket|0>] 
\le L^{(i)}_x(A) \Lmax^{(j)}\sA[\purifier\circ \bi{S_{B_i}}, O_{\yy_i}]
\le L^{(i)}_x(A) \OO\sA[L^{(j)}_{\yy_i} (B_i)].
\]

The time complexity of $\purifier$ is $\OO(\timeR)$.
Also, the direct sum $\bigoplus_i \bi{S_{B_i}}$ can be implemented in $\OO(\timeR)$ by \rf{prp:program->transducerParallel}.
Thus, the time complexity of $S$ is $\TC(S_f) + \OO(\timeR)$.
\pfend

We get \rf{thm:introCompositionFunctionQRAG} from this theorem by using a transducer $S_f$ obtained from the program $A$ using \rf{thm:program->transducerQRAG}, and then applying \rf{prp:approximateImplementation} to the resulting transducer.
Moreover, we get that the algorithm executes the input oracle $O_z$
\[
\OO\sB[ \max_z \sum_{i=1}^n L^{(i)}_x(A) L_{\yy_i}(B_i)  ]
\]
times.

\rf{thm:compositionFunctionQRAG} can be used multiple times in a row to obtain a composed transducer for a tree of functions similar to the one in \rf{thm:introCompositionTree}.
If $d$ is the depth of the tree, the query complexity grows by the factor of $C^d$, where $C$ is the constant in~\rf{eqn:compositionFunctionQRAGQuery}.
This growth is completely analogous to what one obtains using span programs for the query complexity.
The contribution to the time complexity from the subroutines on layer $\ell$ also grows by the factor of $C^\ell$ because they are multiplied by the corresponding query complexity in~\rf{eqn:compositionFunctionQRAG}.
The time complexity of the final transducer is $\OO(d\timeR)$.

\subsection*{Acknowledgements}
We would like to thank Titouan Carette for bringing references~\cite{bartha:quantumTuringAutomata, AndresMartinez:phd} to our attention.
We are grateful to anonymous referees for their useful suggestions on the presentation of this paper.	

AB is supported by the Latvian Quantum Initiative under European Union Recovery and Resilience Facility project no. 2.3.1.1.i.0/1/22/I/CFLA/001 and the QuantERA project QOPT.

SJ is supported by NWO Klein project number OCENW.Klein.061; and ARO contract no W911NF2010327. SJ is funded by the European Union (ERC, ASC-Q, 101040624). Views and opinions expressed are however those of the author(s) only and do not necessarily reflect those of the European Union or the European Research Council. Neither the European Union nor the granting authority can be held responsible for them. 
SJ is supported by the project Divide \& Quantum  (with project number 1389.20.241) of the research programme NWA-ORC which is (partly) financed by the Dutch Research Council (NWO).
SJ is a CIFAR Fellow in the Quantum Information Science Program. 

\bibliographystyle{habbrvM}
{
\small
\bibliography{belov}
}

\end{document}